\documentclass[12pt]{article}

\newcommand{\blind}{0}

\usepackage[margin=1in]{geometry}

\usepackage{amsfonts}
\usepackage{amsgen,amsmath,amstext,amsbsy,amsopn,amssymb,subfigure, amsthm}
\usepackage{graphicx}
\usepackage{comment}
\usepackage{enumitem}
\usepackage{natbib}
\usepackage{booktabs}
\usepackage{blkarray}
\usepackage{algorithm}
\usepackage{algpseudocode}
\usepackage{url}
\usepackage{longtable}
\usepackage{mleftright}
\usepackage{multirow}
\usepackage[usenames]{color}

\usepackage{hyperref}
\hypersetup{
    colorlinks=true,
    linkcolor=blue,
    filecolor=blue,
    urlcolor=red,
    citecolor=blue,
}

\allowdisplaybreaks

\newtheorem{Example}{Example}  

\newtheorem{Theorem}{Theorem}
\newtheorem*{Theorem*}{Theorem}
\newtheorem{Lemma}{Lemma}
\newtheorem{Remark}{Remark}
\newtheorem{Corollary}{Corollary}
\newtheorem{Assumption}{Assumption}
\newtheorem*{Assumption*}{Assumption}
\newtheorem{Proposition}{Proposition}

\newcommand{\be}{\begin{equation}}
	\newcommand{\ee}{\end{equation}}
 \newcommand{\bs}{\begin{split}}
	\newcommand{\es}{\end{split}}
\newcommand{\bea}{\begin{eqnarray}}
	\newcommand{\eea}{\end{eqnarray}}
\newcommand{\beas}{\begin{eqnarray*}}
	\newcommand{\eeas}{\end{eqnarray*}}

\newcommand{\EE}{\mathbb{E}}

\DeclareMathAlphabet\mathbfcal{OMS}{cmsy}{b}{n}

\allowdisplaybreaks

\def\spacingset#1{\renewcommand{\baselinestretch}{#1}\small\normalsize}

\begin{document}
\spacingset{1}

\ifnum\blind=0\relax
{
  \title{\bfseries {RECAP}: Representation-Based Calibration with Amortized Pooling for Frequentist Simulation-Based Inference}
  \author{Runshi Tang\\
    Department of Statistics, Columbia University, USA\\
    and\\
    Genevera I. Allen\\
    Department of Statistics, Columbia University, USA}
  \date{}
  \maketitle
}
\else
{
  \bigskip
  \bigskip
  \bigskip
  \begin{center}
    {\LARGE\bfseries {RECAP}: Representation-Based Calibration with Amortized Pooling for Frequentist Simulation-Based Inference}
  \end{center}
  \medskip
}
\fi

\bigskip
\begin{abstract}
Simulation-based inference enables uncertainty quantification in broad contexts when a simulation model is available but when a test statistic's sampling distribution is not tractable.
But, frequentist calibration in this setting can be costly over high-dimensional parameter spaces.
In this paper, we introduce \texttt{RECAP} (REpresentation-based Calibration with Amortized Pooling), a general framework for frequentist simulation-based inference.
\texttt{RECAP} exploits low-dimensional structure in how the null distribution varies with parameters and is particularly useful when the parameter dimension is high or when different parameter configurations have redundant or compensating effects.
\texttt{RECAP} pools simulated scores according to a representation that preserves this distribution exactly or approximately, allowing distant parameter values to share calibration information.
Representations may be derived from model structure or learned independently.
Empirical quantiles of the pooled scores then yield tests and confidence sets, for which we establish theoretical coverage guarantees.
In the simulation studies, we show that \texttt{RECAP} has lower coverage error than existing approaches and power close to that of the point-wise oracle, while maintaining the computational benefits of amortized calibration.  
An application to population genomics illustrates learned low-dimensional calibration and interpretable uncertainty quantification.
\end{abstract}


\vfill

\newpage
\spacingset{1.55}

\section{Introduction}
\label{sec:introduction}

Many scientific questions are studied using large computational simulators that describe how complex physical systems behave. For example, weather forecasting relies on numerical models of atmospheric dynamics \citep{bauer2015quiet}, while climate scientists use Earth system models that couple processes in the atmosphere, oceans, and land \citep{schneider2017earth}. 
In particle physics, simulations are used to model high-energy interactions and the resulting detector measurements \citep{cranmer2020frontier}. 
Although these simulators can produce predictions, they often do not provide an explicit, tractable likelihood for the observed data. This raises a fundamental statistical question: how can we use observations to learn about unknown model parameters and quantify the associated uncertainty?

When a simulator includes stochastic mechanisms, it implicitly defines a sampling distribution, even if the corresponding likelihood cannot be evaluated. This motivates \emph{simulation-based inference (SBI)}, which uses data simulated under different parameter values to perform statistical inference and quantify uncertainty \citep{digglegratton1984montecarlo}. Simulation-based methods are also useful when the likelihood is available but the sampling distribution of a test statistic is difficult to characterize. Although asymptotic theory can provide useful approximations, their accuracy may be limited in finite samples or for certain parameter configurations \citep{dalmasso2024lf2i,alkadhim2024alffi}.

We consider an observed dataset $X\sim P_{\theta^*}$, where $\theta^*\in\Theta\subseteq\mathbb R^d$ is fixed and unknown, and assume that independent datasets can be simulated from $P_\theta$ at any candidate $\theta\in\Theta$.
For each candidate $\theta\in\Theta$, we test $H_0:\theta^*=\theta$ using a statistic $T(X,\theta)$ with distribution $F_\theta(t)$, whose large values indicate disagreement with the candidate model.
Let $q_{\alpha,\theta}$ denote its $(1-\alpha)$ quantile under $P_\theta$, where $\alpha\in(0,1)$.
If $q_{\alpha,\theta}$ were known for all $\theta$, holding the observed data fixed and retaining all nonrejected candidates would yield $\mathcal C_\alpha(X)=\{\theta\in\Theta:T(X,\theta)\leq q_{\alpha,\theta}\}$.
This confidence set has coverage at least $1-\alpha$ at every true parameter value.

SBI concerns settings in which $q_{\alpha,\theta}$ is not analytically available.
Bayesian SBI, including approximate Bayesian computation and neural posterior estimation, approximates the posterior distribution of $\theta$ given $X$ under a specified prior \citep{Beaumont2002,cranmer2020frontier}.
However, posterior credibility at level $1-\alpha$ does not generally imply repeated-sampling coverage at that level for every $\theta$ or the prespecified $\theta^*$.
Frequentist SBI instead seeks estimates $\widehat q_{\alpha,\theta}$ such that $\widehat{\mathcal C}_\alpha(X)=\{\theta\in\Theta:T(X,\theta)\leq\widehat q_{\alpha,\theta}\}$ achieves fixed-parameter coverage $\mathbb P_{\theta^*}\{\theta^*\in\widehat{\mathcal C}_\alpha(X)\}\geq 1-\alpha$, or, more strongly, uniform coverage $\inf_{\theta\in\Theta}\mathbb P_\theta\{\theta\in\widehat{\mathcal C}_\alpha(X)\}\geq 1-\alpha$.

Frequentist SBI has classical roots in Monte Carlo testing \citep{dwass1957modified} and the bootstrap \citep{efron1979bootstrap}.
Bootstrap approximates $F_{\theta^*}$ through empirical or fitted resampling, whose validity depends on reproducing the dependence relevant to $T(X,\theta^*)$.
Standard i.i.d.\ resampling need not reproduce $F_{\theta^*}$ for structured observations, for example a matrix whose noise entries share a random scale, while specialized bootstrap schemes require additional structural assumptions \citep{kunsch1989jackknife}.

The parametric bootstrap instead simulates from a fitted model: given $\widehat\theta=\widehat\theta(X)$, it generates $X_b^*\sim P_{\widehat\theta}$, uses $T(X_b^*,\widehat\theta)$ to estimate its quantile $q_{\alpha,\widehat\theta}$, and approximates the target $q_{\alpha,\theta^*}$ by $q_{\alpha,\widehat\theta}$ \citep{davisonhinkley1997bootstrap}.
This approximation can be inaccurate when $F_\theta$ is sensitive to parameter estimation error or the problem is nonregular \citep{andrews2000inconsistency}, and increasing the simulation budget does not eliminate the discrepancy between $q_{\alpha,\widehat\theta}$ and $q_{\alpha,\theta^*}$.

Direct pointwise simulation avoids the plug-in approximation by generating $X_b^*\sim P_\theta$ separately at each candidate $\theta$ and obtaining $\widehat q_{\alpha,\theta}$ from the scores $T(X_b^*,\theta)$.
Inverting these valid point-null tests yields fixed-parameter coverage \citep{bolvikenskovlund1996confidence}, but repeating calibration across many candidate parameter values can be computationally expensive.

Modern frequentist SBI amortizes calibration using a common simulation bank $(\theta_i,T(X_i,\theta_i))$, with $X_i\sim P_{\theta_i}$.
\texttt{ACORE} and \texttt{WALDO}, within the \texttt{LF2I} framework \citep{dalmasso2020acore,dalmasso2024lf2i,masserano2023waldo}, and \texttt{ALFFI} \citep{alkadhim2024alffi} take a \emph{global} regression approach, estimating $\theta\mapsto q_{\alpha,\theta}$ or $(t,\theta)\mapsto F_\theta(t)$ over $\Theta$.
Their coverage guarantees depend on accurate estimation of these conditional quantities, which can be demanding when $d$ is large.
Alternatively, \citet{cabezas2026tmlr-conformal} proposed \texttt{TRUST} and \texttt{TRUST++} for \emph{local} calibration using tree- and forest-based partitions.
Specifically, they learn partitions by regressing $T(X_i,\theta_i)$ on $\theta_i$ and calibrate using scores pooled within tree leaves or according to forest proximity.
However, a mean-regression criterion targeting $\EE_\theta[T(X,\theta)]$ can miss directions affecting only the variance or tails of $F_\theta$, even when these determine $q_{\alpha,\theta}$.
This motivates learning a geometry that preserves $F_\theta$ rather than only its mean.

The main contribution of this work is a novel framework for frequentist SBI that exploits low-dimensional structure.
Specifically, we introduce \texttt{RECAP}---REpresentation-based Calibration with Amortized Pooling---based on the relation $F_\theta(t)=H_0(t,g(\theta))$ for some function $H_0$, or an approximation thereof, where $g:\Theta\to\mathbb R^s$ and $s\leq d$.
The representation $g(\theta)$ captures how the score distribution varies with $\theta$, allowing parameters far apart in the original space to share calibration information.
In \texttt{RECAP}, simulated scores are pooled by neighborhoods according to $g(\theta)$ and then $q_{\alpha,\theta}$ is estimated by an adjusted empirical quantile.

In many statistical models, an exact $g$ can be derived from the model structure. Examples are provided in Section~\ref{subsec_examples}. When $g$ is unknown, we propose to learn $g$ from independent training simulations using sufficient dimension reduction (SDR). SDR seeks a low-dimensional representation of the parameter that preserves the conditional distribution of the scores. We develop general coverage guarantees for \texttt{RECAP} and further specialize them to a concrete implementation that learns a linear representation using the \texttt{RMAVE} ensemble \citep{yin2011ensemble} from the SDR literature and calibrates with a moving-ball pooling rule.

Consider a motivating example, matrix denoising, which illustrates the benefits of this representation.
Let $Y=S+\tau W$, where $S\in\mathbb R^{20\times20}$ has rank two, $W$ has independent $N(0,1/20)$ entries, and $\tau$ is an independent random scale shared by all entries.
For the rank-two truncated-SVD estimator $\widehat S(Y)$, consider the statistic $T(Y,S)=\|\widehat S(Y)-S\|_{\mathrm F}$.
By orthogonal invariance, its null distribution depends on $S$ only through its two nonzero singular values, $g(S)=(\lambda_1(S),\lambda_2(S))$.
Thus, matrices far apart in Frobenius distance can have identical score distributions, enabling calibration in two dimensions. Figure~\ref{fig:intro-matrix-denoising} reports the mean absolute coverage error, mean power, and simulation cost for twelve fixed signal matrices, comparing our proposed framework with other comparators.
\texttt{RECAP} nearly matches the point-wise oracle in coverage and power and has substantially lower coverage error than the non-oracle comparators.
See Section~\ref{sec:simulation-matrix-denoising} for explanation of the comparators and other details.

\begin{figure}[htbp]
    \centering
    \subfigure[Coverage error.]{%
        \includegraphics[width=0.36\textwidth]{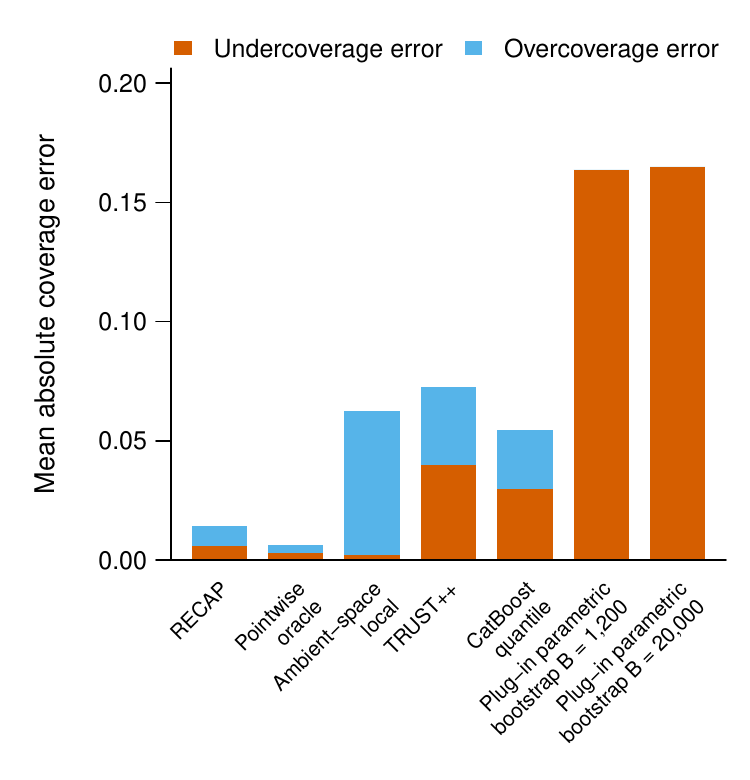}}%
    \hfill
    \subfigure[Mean power.]{%
        \includegraphics[width=0.24\textwidth]{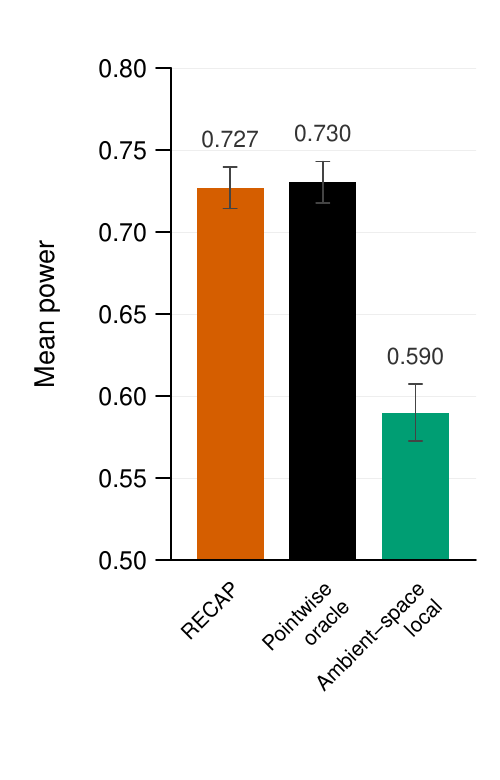}}%
    \hfill
    \subfigure[Simulation cost.]{%
        \includegraphics[width=0.36\textwidth]{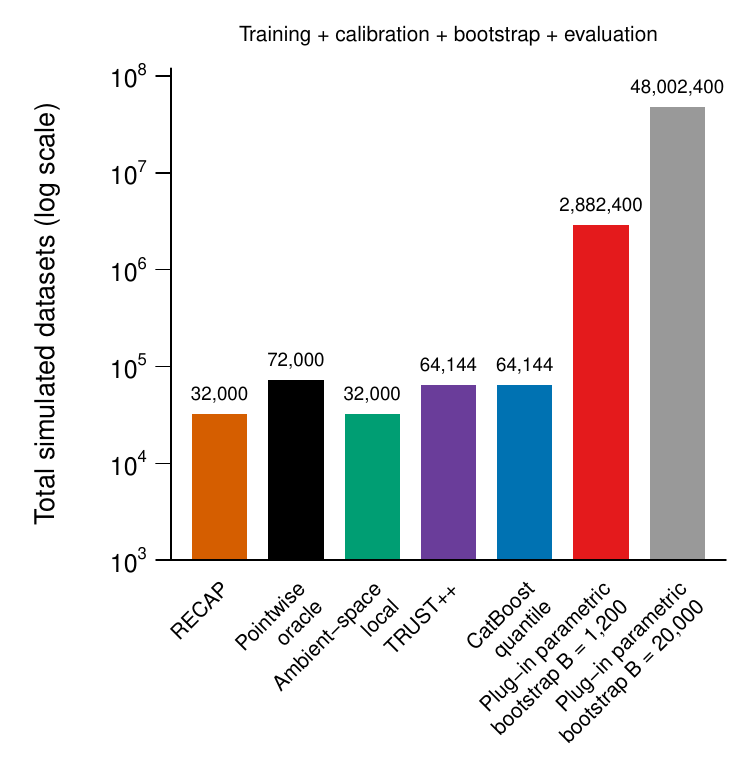}}
    \caption{Matrix denoising at nominal coverage $0.90$.
    (a) Mean absolute coverage error across twelve targets.
    (b) Mean power with $95\%$ intervals. Methods with high undercoverage error are excluded. 
    (c) Total simulated datasets. }
    \label{fig:intro-matrix-denoising}
\end{figure}

\subsection*{Organization}

The remainder of the paper is organized as follows. 
Section~\ref{sec_method} introduces the general \texttt{RECAP} framework and illustrates its use through examples.
Section~\ref{sec:statistical-theory} establishes coverage guarantees.
Section~\ref{sec:simulation-study} evaluates \texttt{RECAP} in simulation studies involving known, model-informed, and learned calibration representations.
Section~\ref{sec:real} applies the method to high-dimensional demographic inference in population genomics.
The final section discusses limitations and future directions.
The supplement provides additional examples, proofs, and implementation details.

\section{Method}
\label{sec_method}

\subsection{The general \texttt{RECAP} framework}
\label{subsec_recap_framework}

We begin by introducing the general \texttt{RECAP} framework. 
Let $X\sim P_\theta$, where $\theta\in\Theta\subseteq\mathbb R^d$, and assume that a simulator is available to generate independent realizations from $P_\theta$ for every $\theta\in\Theta$.
Let $\pi$ be a user-specified prior distribution on $\Theta$ used to sample simulation parameters.
Throughout, we use $\theta$ to denote a deterministic parameter value and $\vartheta$ to denote a random simulation parameter drawn from $\pi$.
Independently of $X$, generate $M$ independent simulation pairs $(\vartheta_i,X_i)$ according to
\[
\vartheta_i \overset{\mathrm{i.i.d.}}{\sim} \pi,
\qquad
X_i\mid\vartheta_i\sim P_{\vartheta_i},
\qquad i=1,\ldots,M.
\]
For a given real-valued test statistic $T(X,\theta)$, whose larger values indicate greater disagreement between $X$ and $\theta$, compute the scores $Z_i=T(X_i,\vartheta_i)$ and form the simulation bank $\{(\vartheta_i,Z_i)\}_{i=1}^M$.
Denote the null cumulative distribution function by $F_\theta(t):=\mathbb P_\theta\{T(X,\theta)\leq t\}$, which is also the conditional cumulative distribution function of $Z_i$ given $\vartheta_i=\theta$.
For $\alpha\in(0,1)$, define its $(1-\alpha)$ quantile as $q_{\alpha,\theta}:=\inf\{t\in\mathbb R:F_\theta(t)\geq1-\alpha\}$.

The aim of amortized frequentist SBI is to estimate the candidate-specific quantiles $q_{\alpha,\theta}$ across $\theta\in\Theta$ using the simulation bank $\{(\vartheta_i,Z_i)\}_{i=1}^M$, for a prespecified level $\alpha$ or multiple levels of interest.
When $\vartheta_i=\vartheta_j$, the scores $Z_i$ and $Z_j$ share the same null distribution and can be pooled to estimate its quantiles.
On the other hand, when $\vartheta_i$ and $\vartheta_j$ are far apart in $\Theta$, $F_{\vartheta_i}$ and $F_{\vartheta_j}$ may be similar or even identical, and hence pooling $Z_i$ and $Z_j$ can still improve quantile estimation. 
Thus, for calibration at a candidate $\theta$, this suggests pooling scores whose distributions are close to $F_\theta$, rather than restricting attention to parameters that are close to $\theta$ in $\Theta$.
This is the central idea of \texttt{RECAP}---REpresentation-based Calibration with Amortized Pooling---which aims to organize pooling according to similarity of the null score distributions.

We refer to any measurable map $g:\Theta\to\mathcal G\subseteq\mathbb R^s$, with $s\leq d$, as a \emph{calibration representation}.
We call $g$ an \emph{exact} calibration representation if $F_\theta(t)=H_0(t,g(\theta))$ for every $\theta\in\Theta$ and $t\in\mathbb R$, where $H_0(\cdot,u)$ is a cumulative distribution function for each $u\in\mathcal G$.
In this case, the null score distribution depends on $\theta$ only through $g(\theta)$.
If $H_0(t,u)$ varies continuously with $u$, uniformly over $t$, then nearby representations induce similar null score distributions.
This motivates pooling simulated scores using neighborhoods in $\mathcal G$. 
More generally, we also allow $g$ to be inexact, where the null score distribution is only approximately determined by the representation.

In several classical statistical models, an exact representation is available from model structure, as illustrated in Section~\ref{subsec_examples}.
When $g$ is unknown, representation learning can be approached through sufficient dimension reduction (SDR).
Specifically, the exact representation condition implies $Z_i\perp\!\!\!\perp\vartheta_i\mid g(\vartheta_i)$, so $g(\vartheta_i)$ is a sufficient reduction of $\vartheta_i$ for the conditional distribution of $Z_i$.
This connects the estimation of $g$ to the SDR framework, reviewed in Supplement Section~\ref{subsec_sdr}, with the simulated score as the response and the simulation parameter as the predictor.
From an SDR perspective, the key advantage of \texttt{RECAP} compared with the existing local frequentist SBI method \texttt{TRUST} is that we target the \emph{central subspace}, which characterizes the full conditional distribution of the score, whereas \texttt{TRUST} targets the \emph{central mean subspace}, which characterizes only the conditional mean and need not suffice for calibration. We further illustrate this observation in Example~\ref{ex:calibration-poisson} and Section~\ref{subsec:simulation-poisson}.
Thus, the first step of \texttt{RECAP} is to specify a known representation or obtain an estimate $\widehat g$ from a separate representation-learning simulation bank generated independently of the calibration bank.
We introduce a specific representation-learning procedure in Section~\ref{subsec_learning_representation}.

Given a known or estimated calibration representation, denoted by $g$ for simplicity, the next step in \texttt{RECAP} is to calibrate through a pooling rule in $\mathcal G$.
Specifically, let $\mathcal N(u)\subseteq\mathcal G$ be a measurable pooling region containing $u \in \mathcal G$.
This region may be a neighborhood of $u$.
We discuss a specific moving-ball pooling rule in Section~\ref{subsec_moving_ball}.
The pooling rule and its tuning parameters are specified in advance or determined using training or tuning simulations independent of the calibration bank.
For a given representation $g$ and each candidate $\theta\in\Theta$, define the selected indices and pool size by
\[
\mathcal I_\theta(g)
=
\left\{i\in\{1,\ldots,M\}:g(\vartheta_i)\in\mathcal N(g(\theta))\right\},
\qquad
M_\theta=|\mathcal I_\theta(g)|.
\]
We suppress the dependence on $g$ in the pool size and subsequent cutoff notation for simplicity.

We estimate $q_{\alpha,\theta}$ by an adjusted empirical quantile of the pooled scores $\{Z_i:i\in\mathcal I_\theta(g)\}$.
For $M_\theta>0$, write these scores in increasing order as $Z_{(1)}^\theta\leq\cdots\leq Z_{(M_\theta)}^\theta$.
Set $k_\theta=\lceil(M_\theta+1)(1-\alpha)\rceil$ and define
\begin{equation}
\label{eq:recap-cutoff}
\widehat q_{\alpha,\theta}
=
\begin{cases}
Z_{(k_\theta)}^\theta, & k_\theta\leq M_\theta,\\
+\infty, & k_\theta=M_\theta+1.
\end{cases}
\end{equation}

Once the cutoff rule is constructed, it can be applied to an observed dataset $X_0$ under the model assumption $X_0\sim P_{\theta^*}$, where $\theta^*\in\Theta$ is fixed and unknown.
We assume that $X_0$ is independent of the simulations used for representation learning, tuning, and calibration.
For each candidate $\theta$, we reject the point null $\theta^*=\theta$ when $T(X_0,\theta)>\widehat q_{\alpha,\theta}$.
Inverting these tests yields the $(1-\alpha)$ confidence set
\begin{equation}
\label{eq:recap-confidence-set}
\widehat{\mathcal C}_\alpha(X_0)
=
\left\{\theta\in\Theta:T(X_0,\theta)\leq\widehat q_{\alpha,\theta}\right\}.
\end{equation}
With the representation and pooling rule fixed, the same calibration bank supports inference on current and future independent datasets using the same model and test statistic. 
Thus, computational cost in \texttt{RECAP} is amortized across both candidate parameters and observed datasets.
We summarize the general procedure in Algorithm~\ref{alg:recap-framework} and provide the detailed implementation in Supplement Algorithm~\ref{alg:recap-detailed}.

\begin{algorithm}[ht]
\caption{\texttt{RECAP}: REpresentation-based Calibration with Amortized Pooling}
\label{alg:recap-framework}
\textbf{Input:} Simulator $P_\theta$, statistic $T$, prior $\pi$, calibration budget $M$, and level $\alpha\in(0,1)$.
\begin{enumerate}[leftmargin=*,nosep]
    \item \textbf{Representation and pooling.} Specify a known representation $g$ or estimate it from separate simulations, and fix a pooling rule $\mathcal N$.
    \item \textbf{Simulation bank.} Generate $M$ independent parameter--score pairs $\{(\vartheta_i,Z_i)\}_{i=1}^M$ using $\pi$, $P_\theta$, and $T$.
    \item \textbf{Representation-based calibration.} For each candidate $\theta$, select the indices $\mathcal I_\theta(g)$ for which $g(\vartheta_i)\in\mathcal N(g(\theta))$.
    Calculate the adjusted empirical quantile $\widehat q_{\alpha,\theta}$ from the pooled scores $\{Z_i:i\in\mathcal I_\theta(g)\}$ using \eqref{eq:recap-cutoff}.
    \item \textbf{Inference.} For an independent observed dataset $X_0$, reject $\theta^*=\theta$ when $T(X_0,\theta)>\widehat q_{\alpha,\theta}$, or construct the confidence set in \eqref{eq:recap-confidence-set}.
\end{enumerate}
\end{algorithm}

\subsection{Models with low-dimensional calibration structure}
\label{subsec_examples}

In this section, we first illustrate low-dimensional calibration structure in statistical models and then discuss its potential role in scientific simulators.

\subsubsection*{Statistical models}

\paragraph{Exact representations generated by invariance.}

A common source of exact calibration representations is an invariance shared by the model and the score.
The following proposition formalizes this observation.

\begin{Proposition}[Exact representation from invariance]
\label{prop:invariant-score-representation}
Let a group $\mathcal H$ act measurably on the sample and parameter spaces, and suppose that for every $h\in\mathcal H$ and $\theta\in\Theta$,
\[
X\sim P_\theta
\quad\Longrightarrow\quad
h\cdot X\sim P_{h\cdot\theta},
\qquad
T(h\cdot X,h\cdot\theta)=T(X,\theta).
\]
Then $F_{h\cdot\theta}=F_\theta$.
Thus, any measurable map $g$ satisfying $g(\theta)=g(\theta')$ if and only if $\theta'=h\cdot\theta$ for some $h\in\mathcal H$ gives an exact calibration representation $F_\theta(t)=H_0(t,g(\theta))$.
\end{Proposition}


We illustrate this invariance property in the following examples. 

\begin{Example}[Matrix denoising]
\label{ex:calibration-matrix-denoising}
Consider $Y=S+E$, where $S\in\mathbb R^{p_1\times p_2}$ has rank $r$ and the noise satisfies $QER^\top\overset{d}{=}E$ for every pair of orthogonal matrices $Q$ and $R$ of compatible dimensions.
Let $\widehat S$ be an estimator satisfying $\widehat S(QYR^\top)=Q\widehat S(Y)R^\top$, and consider the score $T(Y,S)=\|\widehat S(Y)-S\|_{\mathrm F}$.
Simultaneously transforming $Y$ and $S$ by left and right orthogonal matrices preserves both the model and the score.
Because two matrices have the same singular values precisely when one can be obtained from the other by left and right orthogonal transformations, an exact representation is
\[
g(S)=\bigl(\lambda_1(S),\ldots,\lambda_r(S)\bigr)^\top,
\]
where the nonzero singular values are ordered.
The motivating example in Section~\ref{sec:introduction} is a special case.
\end{Example}

The calibration representation in Example~\ref{ex:calibration-matrix-denoising} follows from orthogonal invariance.
The same principle applies to principal component analysis \citep{johnstone2001distribution}, canonical correlation analysis \citep{hotelling1936relations}, and matrix regression with a random design \cite{rechtfazelparrilo2010guaranteed,candesplan2011tight}; see Supplement Section~\ref{sec:supp-orthogonal-invariance} for further details.

Another common source of exact calibration representations is permutation invariance.
We illustrate this principle below using spectral clustering in the stochastic block model (SBM).
Related representations for spectral clustering in isotropic Gaussian mixtures \citep{loffler2021optimality} and sparse random-design regression \citep{wainwright2009sharp} are given in Supplement Section~\ref{sec:supp-permutation-invariance}.

\begin{Example}[Spectral clustering in SBM]
\label{ex:calibration-sbm}
Consider an undirected SBM \citep{lei2015consistency} with conditionally independent edges $A_{ij}\mid z,B\sim\operatorname{Bernoulli}(B_{z_i z_j})$, $1\leq i<j\leq n$, where $z\in\{1,\ldots,K\}^n$ is the membership vector and symmetric $B\in[0,1]^{K\times K}$ is the connection probability matrix.
For a vertex-permutation-equivariant clustering procedure $\widehat z(A)$, use the misclustering score
\[
T(A,(z,B))=\frac{1}{n}\min_{\eta\in\mathfrak S_K}\sum_{i=1}^n\mathbf 1\{\widehat z_i(A)\neq\eta(z_i)\},
\]
where $\mathfrak S_K$ denotes permutations of the community labels.
By permutation invariance, an exact representation is $g(z,B)=\operatorname{canon}(n_1,\ldots,n_K,B)$, where $n_k$ is the size of community $k$ and $\operatorname{canon}$ selects a canonical representative under simultaneous relabeling of the sizes and the rows and columns of $B$.
\end{Example}

Overall, there are many types of statistical models that exhibit some form of invariance and hence, exact representations in RECAP can improve calibration.  

\paragraph{Classical models with unknown representation.}

When an exact calibration representation cannot be identified, the model structure may still provide evidence for an underlying low-dimensional calibration representation, and in some cases suggest informative summaries for constructing one.

\begin{Example}[Matrix completion]
\label{ex:calibration-matrix-completion}
Consider recovery of a rank-$r$ matrix $M$ from uniformly sampled entries with sampling probability $q$, and use the score $T(Y,M)=\|\widehat M(Y)-M\|_{\mathrm F}$, where $Y$ includes the observed locations and values.
Although we do not identify an exact calibration representation here, recovery guarantees for standard completion estimators involve quantities such as rank, singular values, coherence, sampling rate, and noise level \citep{candesrecht2009exact,keshavan2010matrix}.
These results suggest possible low-dimensional calibration representations with these summaries. 
In Section~\ref{subsec:simulation-matrix-completion}, we empirically find that incorporating these summaries substantially improves calibration accuracy.
\end{Example}

Besides matrix completion, low-rank recovery provides several other examples, including tensor decomposition \citep{zhang2018tensor}, and tensor completion \citep{xia2021statistically}, whose recovery guarantees similarly involve ranks, signal strengths, and, for completion, incoherence and sampling properties.

The last example comes from the SDR literature, where it is used to illustrate that the central subspace can contain directions absent from the central mean subspace \citep{yin2011ensemble}. 
This distinction is particularly relevant to calibration, which requires distributional information beyond the conditional mean.
\begin{Example}[Heteroscedastic Poisson regression]
\label{ex:calibration-poisson}
Consider independent responses $Y_j\sim\operatorname{Poisson}(\mu_j(\theta))$, where $\mu_j(\theta)=ne_j\exp(x_j^\top\theta)$, $j=1,\ldots,m$, with fixed design vectors $x_j\in\mathbb R^d$ and exposures $ne_j>0$.
For a candidate $\theta$, use the score
\[
T(Y,\theta)=\sum_{j=1}^m\frac{\{Y_j-\mu_j(\theta)\}^2}{\mu_j(\theta)}.
\]
It follows that $\mathbb E(Z\mid\vartheta)=m$, although the conditional distribution $Z\mid \vartheta$ generally varies with $\vartheta$.
This motivates learning a representation that preserves distributional variation beyond the conditional mean.
We investigate this distinction empirically in Section~\ref{subsec:simulation-poisson}.
On the other hand, although the score's finite-sample distribution is available, calculating its quantiles requires numerical calibration, which can be costly across many $\theta$.
And for fixed $m$, the score converges in distribution to $\chi_m^2$ as $n\rightarrow \infty$, but this approximation can be inaccurate for small expected counts.
The exact score distribution, numerical calibration, and supporting moment calculations are given in Supplement Section~\ref{sec:supp-poisson}.
\end{Example}

\subsubsection*{Scientific simulation models with low-dimensional calibration structure}

Large scientific simulators provide further motivation for seeking low-dimensional calibration representations, as different parameter combinations can produce similar outputs.
In climate modeling, \citet{mauritsen2012tuning} obtain alternative configurations of the Max Planck Institute Earth System Model (MPI-ESM) by changing a cloud or convection parameter and adjusting another to maintain similar radiation balances and cloud-water amounts.
In particle physics, \citet{buckley2010systematic} find strong correlations among fitted fragmentation parameters in the \texttt{PYTHIA} event generator.
A more direct example of low-dimensional structure comes from computational fluid dynamics, where \citet{constantine2015exploiting} find that the integrated exit pressure in simulations of the HyShot II scramjet is well approximated by a function of one linear combination of seven uncertain inputs, including inflow conditions and boundary-layer transition locations.

These findings suggest that approximate calibration representations may exist in scientific simulators in real-world applications, and one could investigate this possibility by learning representations from simulated parameter--score pairs and checking calibration on independent simulations.
If the score distribution varies mainly through the learned representation, \texttt{RECAP} can share calibration information across parameter configurations that are far apart in the original space.

\subsection{Learning the calibration representation}
\label{subsec_learning_representation}

We now give a concrete procedure for learning an unknown calibration representation, $g$, within the general \texttt{RECAP} framework of Section~\ref{subsec_recap_framework}. 
In principle, any SDR method targeting the conditional score distribution can be used within \texttt{RECAP}.
Available approaches include both inverse-regression methods \citep[e.g.,][]{li1991sir,cook2000save,liwang2007directional} and forward-regression methods \citep[e.g.,][]{xia2007constructive,wangxia2008sliced,yin2011ensemble}.
We use the characteristic-function \texttt{RMAVE} ensemble \citep{yin2011ensemble} because its weighted least-squares updates are straightforward to implement, it has established convergence guarantees under regularity conditions, and relies on less restrictive assumptions than many classical SDR methods.

Suppose that the null score distribution admits an exact linear representation of the form
\begin{equation}
\label{eq:linear-distributional-representation}
F_\theta(t)=H_0(t,B_0^\top\theta),
\qquad
B_0\in\mathbb R^{d\times s},
\qquad
B_0^\top B_0=I_s.
\end{equation}
We generate $N$ representation-training pairs according to
\begin{equation}
\label{eq:recap-training-bank}
\vartheta_i^{\mathrm{tr}}\overset{\mathrm{i.i.d.}}{\sim}\pi,
\qquad
X_i^{\mathrm{tr}}\mid\vartheta_i^{\mathrm{tr}}\sim P_{\vartheta_i^{\mathrm{tr}}},
\qquad
Z_i^{\mathrm{tr}}=T(X_i^{\mathrm{tr}},\vartheta_i^{\mathrm{tr}}),
\qquad i=1,\ldots,N.
\end{equation}
These simulations are independent of those used for calibration.
We then apply the \texttt{RMAVE} ensemble with dimension $s$ to the training pairs to obtain an estimate $\widehat B$ and define $\widehat g(\theta)=\widehat B^\top\theta$.
The fitted representation is then held fixed for subsequent pooling and calibration.
Algorithm~\ref{alg:recap-representation} summarizes this fixed-dimension fit; selection among candidate dimensions is described below.

\begin{algorithm}[ht]
\caption{Learning a calibration representation by the \texttt{RMAVE} ensemble}
\label{alg:recap-representation}
\textbf{Input:} Simulator $P_\theta$, statistic $T$, prior $\pi$, training budget $N$, dimension $s$, and \texttt{RMAVE} ensemble settings.
\begin{enumerate}[leftmargin=*,nosep]
    \item \textbf{Training simulations.} Generate $\{(\vartheta_i^{\mathrm{tr}},Z_i^{\mathrm{tr}})\}_{i=1}^N$ by \eqref{eq:recap-training-bank}.
    \item \textbf{Subspace estimation.} Apply the \texttt{RMAVE} ensemble criterion \eqref{eq:rmave-ensemble-objective} with predictors $\vartheta_i^{\mathrm{tr}}$ and transformed responses $f_\ell(Z_i^{\mathrm{tr}})$ to obtain $\widehat B$.
    \item \textbf{Calibration representation.} Set $\widehat g(\theta)=\widehat B^\top\theta$.
\end{enumerate}
\end{algorithm}

The linear assumption in \eqref{eq:linear-distributional-representation} holds for models whose null score distributions depend on a few linear combinations of the parameters, including those discussed in  Section~\ref{subsec_examples}.
Approximate linear structure, such as the scientific simulator in  \citep{constantine2015exploiting}, provides further motivation for this choice.
We leave the development of more general nonlinear representation-learning methods to future work.

\paragraph{Selecting the representation dimension.}
To tune the representation dimension $s$, we propose to fit $g_r$ for each $r\in\mathcal S\subseteq\{1,\ldots,d\}$ and compare the candidates using a criterion appropriate to the inferential objective.
Options include transformed-response prediction error \citep{yin2011ensemble} and coverage-oriented criteria evaluated on independent validation simulations.
We favor the latter because they directly assess calibration, particularly when a low-dimensional representation is only approximate.
In the Poisson experiment of Section~\ref{subsec:simulation-poisson}, we select the dimension using mean absolute coverage error (MAE), while in the case study of Section~\ref{sec:real}, we use root mean squared coverage error (RMSE).
Supplement Section~\ref{subsec:supp-mae-selection-statement} provides further details on this selection procedure, and Section~\ref{subsec:theory-learned-representation} summarizes its theoretical guarantees.

\subsection{Moving-ball pooling for calibration}
\label{subsec_moving_ball}

Recall that the goal is to estimate the null score distribution across candidate parameters $\theta$.
With a known or learned calibration representation $g$ from previous steps, we adopt a local approach that pools simulated scores in the low-dimensional representation space.
For each $\theta$, the aim is to include enough scores for accurate estimation while keeping their null distributions close to $F_\theta$.
We use moving-ball pooling to balance these considerations, with neighborhood radii adapting to the local concentration of simulated parameters in the representation space.

Specifically, take $\mathcal N(u)=\{v\in\mathcal G:\|v-u\|_2\leq h(u)\}$, where the radius function $h:\mathcal G\to[0,\infty)$ is fixed independently of the final calibration bank and the observed dataset.
For each candidate $\theta$, denote the selected indices and pool size
\begin{equation}
\label{eq:moving-ball-pool}
\mathcal I_\theta(g)
=
\left\{i\in\{1,\ldots,M\}:\|g(\vartheta_i)-g(\theta)\|_2\leq h(g(\theta))\right\},
\qquad
M_\theta=|\mathcal I_\theta(g)|.
\end{equation}
The cutoff $\widehat q_{\alpha,\theta}$ is then calculated from the selected scores using \eqref{eq:recap-cutoff}.

We determine the neighborhood radius using an independent parameter-only pilot $\vartheta_j^{\mathrm{pil}}\overset{\mathrm{i.i.d.}}{\sim}\pi$, $j=1,\ldots,M_{\mathrm{pil}}$.
For a target mass $\rho\in(0,1)$, define
\begin{equation}
\label{eq:moving-ball-radius}
h(u)
=
\inf\left\{
r\geq0:
\frac{1}{M_{\mathrm{pil}}}
\sum_{j=1}^{M_{\mathrm{pil}}}
\mathbf 1\{\|g(\vartheta_j^{\mathrm{pil}})-u\|_2\leq r\}
\geq\rho
\right\}.
\end{equation}
Here, the target mass $\rho$ controls the trade-off between localization and variation in quantile estimation.
This step requires only parameter draws and representation evaluations, with no additional simulated datasets, and is therefore computationally inexpensive.
Algorithm~\ref{alg:recap-moving-ball} summarizes this construction.

\begin{algorithm}[t]
\caption{Constructing a moving-ball pooling rule}
\label{alg:recap-moving-ball}
\textbf{Input:} A fixed known or estimated representation $g:\Theta\to\mathcal G$, prior $\pi$, pilot size $M_{\mathrm{pil}}\geq1$, and target mass $\rho\in(0,1)$.
\begin{enumerate}[leftmargin=*,nosep]
    \item \textbf{Parameter-only pilot.} Draw $\vartheta_j^{\mathrm{pil}}\overset{\mathrm{i.i.d.}}{\sim}\pi$, $j=1,\ldots,M_{\mathrm{pil}}$.
    \item \textbf{Radius function.} For each $u\in\mathcal G$, define $h(u)$ as in \eqref{eq:moving-ball-radius}.
    \item \textbf{Pooling rule.} Set $\mathcal N(u)=\{v\in\mathcal G:\|v-u\|_2\leq h(u)\}$ for every $u\in\mathcal G$. 
\end{enumerate}
\textbf{Output:} The radius function $h$ and moving-ball pooling rule $\mathcal N$.
\end{algorithm}

Finally, we emphasize that the central contribution of \texttt{RECAP} is a framework that exploits low-dimensional structure in the null score distribution for calibration, while the specific calibration rule appears less important. 
In Section~\ref{sec:simulation-study}, we show that \texttt{RECAP} outperforms analogous calibration in the ambient (original) parameter space, while in Supplement Section~\ref{sec:supp-shared-g}, we show that different calibration methods perform similarly when given the same learned representation $g$.

\subsection{Interpreting uncertainty in representation space}
\label{subsec:representation-inference}

High-dimensional confidence sets can be difficult to interpret in scientific applications, as their geometry may be elongated, curved, or disconnected.
Consequently, coordinatewise projections can take the entire allowable parameter ranges while concealing informative constraints on parameter combinations.
On the other hand, \texttt{RECAP} provides a model- and score-informed representation for summarizing uncertainty about such combinations.

For an exact linear representation $g(\theta)=B_0^\top\theta$, perturbations perpendicular to $\operatorname{span}(B_0)$ leave the null score distribution unchanged, as $F_{\theta+v}=F_\theta$ whenever $B_0^\top v=0$ and $\theta+v\in\Theta$.
For nonlinear representations, the same invariance holds along the level sets of $g$.
Thus, $g$ separates parameter variation that governs the null score distribution from variation irrelevant to calibration.

To obtain a lower-dimensional uncertainty summary, consider
\begin{equation}
\label{eq:representation-confidence-set}
\widehat{\mathcal C}_{\alpha,g}(X_0)
=
\left\{g(\theta):\theta\in\widehat{\mathcal C}_\alpha(X_0)\right\}.
\end{equation}
This set inherits the original lower coverage guarantee, since inclusion of $\theta^*$ implies inclusion of $g(\theta^*)$ in its image.

Projecting the confidence set through $g$ removes variation along directions that do not affect the representation.
In the linear case, for example, variation perpendicular to $\operatorname{span}(B_0)$ disappears after projection.
Thus, even when the original confidence set is wide in many individual coordinates, its image in the representation space can be much more concentrated.
This provides a lower-dimensional summary of uncertainty about the parameter combinations that are most relevant to calibration.
Such combinations can also be scientifically meaningful, as they are shaped by the model structure and the chosen test statistic.
The real-data application in Section~\ref{sec:real} illustrates how such summaries can aid the interpretation of uncertainty in a high-dimensional parameter space.

\section{Statistical theory}
\label{sec:statistical-theory}

In this section, we first establish a fixed-parameter coverage guarantee for Algorithm~\ref{alg:recap-framework}, then specialize to moving-ball pooling and learned linear representations. Supplement Section \ref{sec:supp-coverage-theory} provides further details and the proofs.

\subsection{Finite-simulation fixed-parameter coverage}
\label{subsec:theory-general-coverage}

Let $\mathcal D$ collect all preliminary information used to construct the statistic $T$, representation $g$, and pooling rule $\mathcal N$, including any fitting, tuning, dimension selection, and parameter-only pilot draws.

\begin{Assumption}[Independent stages]
\label{ass:recap-independent-stages}
In Algorithm~\ref{alg:recap-framework}, the statistic $T$, representation $g$, and pooling rule $\mathcal N$ are determined by the preliminary information $\mathcal D$.
The final calibration pairs $(\vartheta_i,X_i)$, $i=1,\ldots,M$, are i.i.d.\ with $\vartheta_i\sim\pi$ and $X_i\mid\vartheta_i\sim P_{\vartheta_i}$.
The preliminary information $\mathcal D$, the final calibration bank $\{(\vartheta_i,X_i)\}_{i=1}^M$, and the observed dataset $X_0$ are mutually independent, with $X_0\sim P_\theta$ when evaluating coverage at $\theta$.
\end{Assumption}

Throughout the coverage analysis, $F_\theta(t)=\mathbb P_\theta\{T(X_0,\theta)\leq t\mid\mathcal D\}$ denotes the null score distribution.
For a fresh pair $(\vartheta,X)$ independent of $\mathcal D$ and $X_0$, set $Z=T(X,\vartheta)$ and define
\begin{equation}
\label{eq:recap-pooling-mixture}
\begin{aligned}
p_\theta&=\mathbb P\{g(\vartheta)\in\mathcal N(g(\theta))\mid\mathcal D\},\\
F_\theta^{\mathrm{mix}}(t)&=\mathbb P\{Z\leq t\mid g(\vartheta)\in\mathcal N(g(\theta)),\mathcal D\},\\
\Delta_\theta&=\sup_{t\in\mathbb R}|F_\theta^{\mathrm{mix}}(t)-F_\theta(t)|,
\end{aligned}
\end{equation}
where the last two quantities are defined when $p_\theta>0$.
Here, $F_\theta^{\mathrm{mix}}$ is the mixture of null score distributions $F_\vartheta$ over parameters satisfying $g(\vartheta)\in\mathcal N(g(\theta))$.
The following theorem requires neither an exact representation nor a particular pooling geometry.

\begin{Theorem}[Finite-simulation coverage]
\label{thm:recap-finite-coverage}
Under Assumption~\ref{ass:recap-independent-stages}, fix $\theta\in\Theta$ with $p_\theta>0$.
Conditional on $\mathcal D$, the confidence set $\widehat{\mathcal C}_\alpha(X_0)$ in \eqref{eq:recap-confidence-set} satisfies
\begin{equation}
\label{eq:recap-coverage-lower}
\mathbb P_\theta\{\theta\in\widehat{\mathcal C}_\alpha(X_0)\mid\mathcal D\}\geq1-\alpha-\Delta_\theta.
\end{equation}
If $F_\theta^{\mathrm{mix}}$ is continuous, then
\begin{equation}
\label{eq:recap-coverage-two-sided}
\left|\mathbb P_\theta\{\theta\in\widehat{\mathcal C}_\alpha(X_0)\mid\mathcal D\}-(1-\alpha)\right|
\leq\Delta_\theta+\frac{1}{(M+1)p_\theta}.
\end{equation}
\end{Theorem}

This theorem shows that the coverage error is bounded by the discrepancy between the mixed score distribution $F_\theta^{\mathrm{mix}}$ and the true null distribution $F_\theta$, plus $1/((M+1)p_\theta)$, where $M$ is the simulation bank size and thus $Mp_\theta$ is the expected effective sample size for calibration.

\subsection{Moving-ball calibration}
\label{subsec:theory-moving-ball}

Recall from Section~\ref{subsec_moving_ball} that $h(u)$ is the radius of the pooling ball centered at $u\in\mathcal G$.
For a candidate $\theta$, write $h_\theta=h(g(\theta))$ for its pooling radius.
For the independent parameter--score pair $(\vartheta,Z)$ introduced above, let $H_g(t,u)=\mathbb P\{Z\leq t\mid g(\vartheta)=u,\mathcal D\}$.
Because conditional distributions are only defined almost surely, $H_g$ need not be unique. We therefore choose any version of $H_g$ and hold it fixed throughout the following analysis.
We measure its discrepancy from the candidate-specific null distributions by
\begin{equation}
\label{eq:recap-representation-discrepancy}
\varepsilon_{\mathrm{rep}}(g)=\sup_{\theta\in\Theta}\sup_{t\in\mathbb R}|F_\theta(t)-H_g(t,g(\theta))|.
\end{equation}
For an exact representation, we have $F_\theta(t)=H_0(t,g(\theta))$, choosing $H_g=H_0$ gives $\varepsilon_{\mathrm{rep}}(g)=0$.

\begin{Corollary}[Moving-ball coverage]
\label{cor:recap-moving-ball}
Under Assumption~\ref{ass:recap-independent-stages}, suppose that, conditional on $\mathcal D$, for some $L_0>0$ and $\beta\in(0,1]$,
\begin{equation}
\label{eq:recap-holder}
\sup_{t\in\mathbb R}|H_g(t,u)-H_g(t,u')|\leq L_0\|u-u'\|_2^\beta,
\qquad u,u'\in g(\Theta).
\end{equation}
For moving-ball pooling and $p_\theta>0$, $\Delta_\theta\leq\varepsilon_{\mathrm{rep}}(g)+L_0h_\theta^\beta$.
Consequently, $\mathbb P_\theta\{\theta\in\widehat{\mathcal C}_\alpha(X_0)\mid\mathcal D\}\geq1-\alpha-\varepsilon_{\mathrm{rep}}(g)-L_0h_\theta^\beta$, and, if $F_\theta^{\mathrm{mix}}$ is continuous,
\begin{equation}
\label{eq:recap-moving-ball-error}
\left|\mathbb P_\theta\{\theta\in\widehat{\mathcal C}_\alpha(X_0)\mid\mathcal D\}-(1-\alpha)\right|
\leq
\varepsilon_{\mathrm{rep}}(g)
+L_0h_\theta^\beta
+\frac{1}{(M+1)p_\theta}.
\end{equation}
\end{Corollary}

\begin{Remark}[Coverage rate for an exact representation]
\label{rem:recap-exact-rate}
For a fixed exact representation of dimension $s\geq1$, suppose additionally that $g(\vartheta)$ has a density continuous and positive at $g(\theta)$.
If $F_\theta^{\mathrm{mix}}$ is continuous, choosing deterministic radii $h_\theta\asymp M^{-1/(\beta+s)}$ in \eqref{eq:recap-moving-ball-error} gives $\left|\mathbb P_\theta\{\theta\in\widehat{\mathcal C}_\alpha(X_0)\mid\mathcal D\}-(1-\alpha)\right|=O(M^{-\beta/(\beta+s)})$.
\end{Remark}

\subsection{Learned representations}
\label{subsec:theory-learned-representation}

Now consider estimating the linear representation in \eqref{eq:linear-distributional-representation} on a compact region $\Theta_0$ containing the target parameters and the support of $\pi$.
For the orthonormal estimate $\widehat B$ from Algorithm~\ref{alg:recap-representation}, define
\begin{equation}
\label{eq:recap-subspace-error}
D_\Theta=\sup_{\theta,\theta'\in\Theta_0}\|\theta-\theta'\|_2,
\qquad
a_N=\inf_{R\in\mathcal O(s)}\|\widehat B-B_0R\|_{\mathrm{op}},
\end{equation}
where $\mathcal O(s)$ is the set of $s\times s$ orthogonal matrices.
The alignment accounts for the fact that the matrix itself is not identifiable.
Recall $N$ is the representation-training budget and $M$ is the final calibration budget. 
\begin{Theorem}[Coverage with a learned linear representation]
\label{thm:recap-learned-coverage}
Under Assumption~\ref{ass:recap-independent-stages}, suppose \eqref{eq:linear-distributional-representation} holds on $\Theta_0$ and $H_0$ satisfies \eqref{eq:recap-holder} on $B_0^\top\Theta_0$, with $H_0$ in place of $H_g$.
Use $g=\widehat g$ and moving-ball pooling, and fix $\theta\in\Theta_0$ with $p_\theta>0$.
Conditional on $\mathcal D$,
\begin{equation}
\label{eq:recap-subspace-transfer}
\Delta_\theta\leq L_0(h_\theta+D_\Theta a_N)^\beta,
\qquad
\mathbb P_\theta\{\theta\in\widehat{\mathcal C}_\alpha(X_0)\mid\mathcal D\}\geq1-\alpha-L_0(h_\theta+D_\Theta a_N)^\beta.
\end{equation}
If $F_\theta^{\mathrm{mix}}$ is continuous, then
\begin{equation}
\label{eq:recap-learned-error}
\left|\mathbb P_\theta\{\theta\in\widehat{\mathcal C}_\alpha(X_0)\mid\mathcal D\}-(1-\alpha)\right|\leq L_0(h_\theta+D_\Theta a_N)^\beta+\frac{1}{(M+1)p_\theta}.
\end{equation}
\end{Theorem}

To simplify the result, further suppose that, with probability tending to one, $c_hh_M\leq h_\theta\leq C_hh_M$ and $p_\theta\geq c_ph_\theta^s$, for fixed positive constants $c_h,C_h,c_p$ and a deterministic sequence $h_M>0$.
If $D_\Theta,L_0$ are fixed, $h_M\to0$, $Mh_M^s\to\infty$, and $F_\theta^{\mathrm{mix}}$ is continuous, then \eqref{eq:recap-learned-error} gives, with probability tending to one,
\begin{equation}
\label{eq:recap-joint-error}
\left|\mathbb P_\theta\{\theta\in\widehat{\mathcal C}_\alpha(X_0)\mid\mathcal D\}-(1-\alpha)\right|\leq C\left(a_N^\beta+h_M^\beta+\frac{1}{Mh_M^s}\right),
\end{equation}
where $C>0$ is independent of $N$ and $M$.
Choosing $h_M\asymp M^{-1/(\beta+s)}$ balances the two calibration terms and gives $\left|\mathbb P_\theta\{\theta\in\widehat{\mathcal C}_\alpha(X_0)\mid\mathcal D\}-(1-\alpha)\right|\leq C\{a_N^\beta+M^{-\beta/(\beta+s)}\}$ with probability tending to one, after enlarging $C$ if necessary. 
The pilot rule in Algorithm \ref{alg:recap-moving-ball} attains these radius and mass orders under Supplement Lemma~\ref{lem:recap-pilot-mass}, provided its local-mass constants hold uniformly over the learned maps with probability tending to one.

Combining Theorem~\ref{thm:recap-learned-coverage} with the convergence guarantee for the \texttt{RMAVE} ensemble \citep[Theorem~6.1]{yin2011ensemble} yields the following rate in the representation-training and calibration budgets.

\begin{Corollary}[Coverage rate with the \texttt{RMAVE} ensemble]
\label{cor:recap-rmave-rate}
Suppose the assumptions of Theorem~\ref{thm:recap-learned-coverage} and Supplement Assumption~\ref{ass:recap-rmave} hold, with fixed $D_\Theta$ and $L_0$.
For a fixed $\theta\in\Theta_0$, suppose additionally that $F_\theta^{\mathrm{mix}}$ is continuous and, with probability tending to one, $c_hh_M\leq h_\theta\leq C_hh_M$ and $p_\theta\geq c_ph_\theta^s$, for fixed positive constants $c_h,C_h,c_p$ and a deterministic sequence $h_M\asymp M^{-1/(\beta+s)}$.
Then, as $N,M\to\infty$,
\begin{equation}
\label{eq:recap-rmave-joint-rate}
\left|\mathbb P_\theta\{\theta\in\widehat{\mathcal C}_\alpha(X_0)\mid\mathcal D\}-(1-\alpha)\right|
=
O_{\mathbb P}\left[\left\{(\log N)N^{-4/(s+4)}+N^{-1/2}\right\}^\beta+M^{-\beta/(\beta+s)}\right].
\end{equation}
\end{Corollary}

Thus, as the representation-training budget $N$ and calibration budget $M$ tend to infinity, the conditional coverage converges in probability to the nominal level $1-\alpha$ at each fixed $\theta$ under the stated assumptions.

\paragraph{Fixed simulation designs.}
The finite-simulation guarantee in Theorem~\ref{thm:recap-learned-coverage} also extends to prespecified calibration locations with arbitrary replication and independent simulator draws.
Conditional on the realized design, the same representation and pooling terms control coverage, while the two-sided bound replaces $1/\{(M+1)p_\theta\}$ by $1/(M_\theta+1)$ and assumes continuity of $F_\theta$.
Such designs allow direct control over simulation allocation when a sampling distribution $\pi$ is difficult to specify or feasible parameter configurations form a sparse or constrained set.
Supplement Subsection~\ref{sec:supp-fixed-design} gives the detailed result.

\paragraph{Coverage with dimension selection.}
Now consider the dimension selection procedure of Section~\ref{subsec_learning_representation} for learned representations. 
Specifically, we compare fitted representations by coverage MAE on independent validation simulations, using a separate bank for selection calibration and a fresh bank for final calibration.
Under the score-law and local-mass conditions in Supplement Proposition~\ref{prop:recap-mae-selection}, the selected procedure's expected MAE under $\pi$ is bounded by the smallest population MAE among the fitted candidates plus validation and calibration-bank approximation errors.
This controls MAE for the full procedure with dimension selection.
The details and proof are given in Supplement Section~\ref{sec:supp-mae-selection}.

\section{Simulation Studies}
\label{sec:simulation-study}

We conduct three main simulation studies using the models in Examples~\ref{ex:calibration-matrix-denoising}, \ref{ex:calibration-matrix-completion}, and \ref{ex:calibration-poisson}, covering exact, model-informed, and learned calibration representations, respectively.
We also study spectral clustering under the SBM in Example~\ref{ex:calibration-sbm}, with results reported in Supplement Section~\ref{subsec:simulation-sbm}.
In this experiment, \texttt{RECAP} closely follows the pointwise oracle in coverage and power (Supplement Figure~\ref{fig:simulation-sbm}).

All experiments use nominal coverage $1-\alpha=0.90$ and evaluate coverage and power at twelve fixed target parameters.
We compare \texttt{RECAP} with ambient-space local calibration, \texttt{TRUST++} \citep{cabezas2026tmlr-conformal}, direct \texttt{CatBoost} quantile regression, a plug-in parametric bootstrap, and a pointwise oracle.
Ambient-space local calibration pools simulated scores using neighborhoods in the original parameter space with the same moving-ball rule as in \texttt{RECAP}.
\texttt{CatBoost} quantile regression follows the strategy of \texttt{LF2I} \citep{dalmasso2024lf2i}, directly learning the map $\theta\mapsto q_{\alpha,\theta}$ over the full parameter space.
We use it as a representative global regression method because it directly models the candidate-specific calibration quantile, without requiring additional structures as in other global approaches \citep{masserano2023waldo,alkadhim2024alffi}.
\texttt{TRUST++} is the local calibration method with forest-based pooling proposed in \citet{cabezas2026tmlr-conformal}.
Hyperparameters in \texttt{TRUST++} and \texttt{CatBoost} are tuned using independent simulations.
Supplement Section~\ref{sec:supp-simulation} gives further details including the complete designs, simulation budgets, and method implementations.
For estimated coverage $\widehat c_i$ at target $\theta_i$, we summarize calibration accuracy by $\operatorname{MAE}=12^{-1}\sum_{i=1}^{12}|\widehat c_i-0.90|$.
The reported coverage is evaluated on further independent observations, with the fitted procedures and reusable calibration banks held fixed.
Each power value averages rejection probabilities for two false candidate parameters, with data generated at the corresponding true parameter and each fixed-method cutoff evaluated at the candidate being tested.
Supplement Section~\ref{subsec:supp-sim-diagnostics} gives further details.
We also note that some specialized inferential procedures are available for each of the four examples we study, but our primary comparisons focus on general-purpose calibration frameworks applicable across all four models; hence, we exclude these specialized procedures from our comparisons. 

\subsection{Matrix denoising}
\label{sec:simulation-matrix-denoising}

Following Example~\ref{ex:calibration-matrix-denoising}, we observe $Y=S+\tau W$, where $S\in\mathbb R^{20\times20}$ has rank two, the entries of $W$ are independent $N(0,1/20)$ variables, and $\tau$ equals $0.5/\sqrt{1.25}$ or $1.5/\sqrt{1.25}$ with equal probability.
The score is $T(Y,S)=\|\widehat S_2(Y)-S\|_{\mathrm F}$, where $\widehat S_2$ is the rank-two truncated-SVD estimator, and \texttt{RECAP} uses the exact representation $g(S)=(\lambda_1(S),\lambda_2(S))^\top$.
\texttt{RECAP}, ambient-space local calibration, \texttt{TRUST++}, and \texttt{CatBoost} use the same final bank of $M=20{,}000$ simulated parameter--score pairs.
The twelve targets combine $\lambda_2\in\{1.25,1.75,2.25,2.75\}$ with $\lambda_1-\lambda_2\in\{0.5,1.0,1.5\}$.
For power, the two false matrices have the same singular values but rotated left singular vectors, at Frobenius distances $3$ and $3.5$ from the truth.
See Supplement Section~\ref{subsec:supp-sim-denoising} for details.

Figure~\ref{fig:simulation-denoising} shows that \texttt{RECAP} closely follows the pointwise oracle in both coverage and power, whereas ambient-space calibration is generally conservative and less powerful.
\texttt{TRUST++} and \texttt{CatBoost} exhibit larger parameter-dependent coverage deviations, including both undercoverage and overcoverage.
The mean powers of \texttt{RECAP}, the pointwise oracle, and ambient-space calibration are $0.727$, $0.730$, and $0.590$, respectively, as summarized in Figure~\ref{fig:intro-matrix-denoising}.
The two plug-in bootstrap budgets give similar results and substantial undercoverage, indicating that increasing the resampling budget does not resolve the plug-in discrepancy in this experiment.

\begin{figure}[t]
    \centering
    \includegraphics[width=0.8\textwidth]{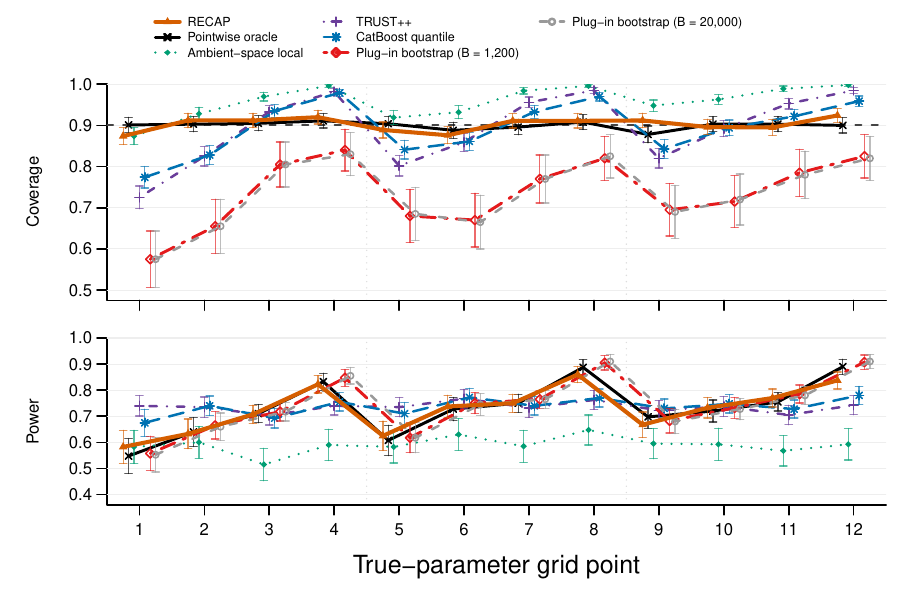}
    \caption{Matrix denoising: coverage and power.
    }
    \label{fig:simulation-denoising}
\end{figure}

\subsection{Noisy matrix completion}
\label{subsec:simulation-matrix-completion}

In the example, we illustrate how model-informed structure can improve calibration in \texttt{RECAP}, without requiring that representation to be exact or learned from simulations.
We construct a fixed, low-dimensional representation for \texttt{RECAP} using signal strength and row and column coherence.
We instantiate Example~\ref{ex:calibration-matrix-completion}, noisy matrix completion, using rank-two matrices in $\mathbb R^{32\times32}$, independent entry-observation probability $0.40$, and Gaussian observation noise with standard deviation $0.35$.
The estimator uses five ridge alternating-least-squares sweeps after spectral initialization, and the score is $\|\widehat M-M\|_{\mathrm F}/32$.
A controlled family varies the singular values $32\rho(1,0.60)$ and the row and column coherence through parameters $\kappa_U$ and $\kappa_V$.
The twelve targets combine $\rho\in\{0.50,0.65,0.80,0.95\}$ with $\kappa_U=\kappa_V\in\{0,0.5,1\}$.
\texttt{RECAP}'s fixed six-coordinate representation comprises normalized signal strength, two log-coherences, and their three pairwise interactions.
Power is evaluated after changing $\rho$ by $0.12$ or $0.18$ toward the proposal midpoint, with all other coordinates fixed.
The complete family and summary construction are given in Supplement Section~\ref{subsec:supp-sim-completion}.

Figure~\ref{fig:simulation-completion} shows that the fixed-summary version of \texttt{RECAP} generally tracks the pointwise oracle in coverage and power, although some coverage deviations remain.
Other calibration methods are generally conservative at lower coherence levels and exhibit undercoverage at several high-coherence targets.
Supplement Figure~\ref{fig:supp-simulation-completion-summary}(a) shows that fixed-summary \texttt{RECAP} has smaller coverage error than the other calibrators and the plug-in bootstrap.
This version of \texttt{RECAP} uses $32{,}000$ simulated datasets in total, compared with $64{,}144$ for each of \texttt{TRUST++} and \texttt{CatBoost}, $192{,}000$ for the pointwise oracle, and approximately $2.88$ million for the bootstrap (panel~(b)).

We also evaluated \texttt{RECAP} with a linear representation learned from the raw matrix entries using the \texttt{RMAVE} ensemble.
It performed much worse than the fixed-summary version in this example.
Note that the signal-strength and coherence summaries are nonlinear functions of the matrix, and supplying them directly was more effective here than using the learned linear projection.
These results illustrate the benefit of incorporating informative features into calibration and motivate nonlinear representation learning as a potential direction for future work.

\begin{figure}[t]
    \centering
    \includegraphics[width=0.8\textwidth]{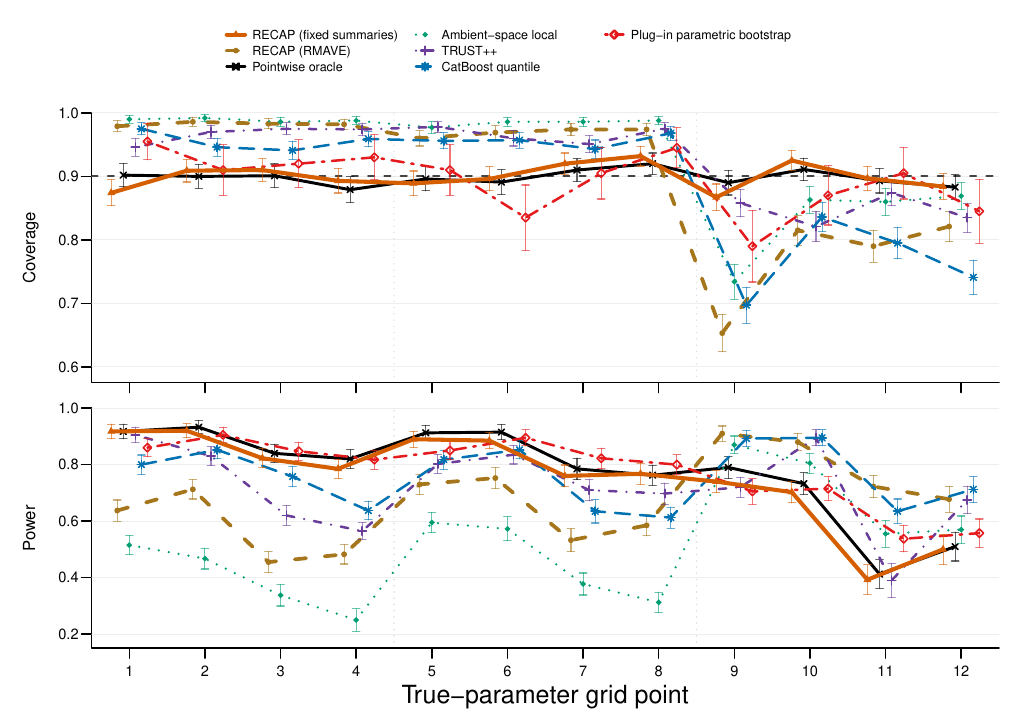}
    \caption{Noisy matrix completion: coverage and power.
    }
    \label{fig:simulation-completion}
\end{figure}

\subsection{Heteroscedastic Poisson regression}
\label{subsec:simulation-poisson}

Our next example on heteroscedastic Poisson regression illustrates how \texttt{RECAP} can learn an unknown calibration representation from simulations. We use $m=d=12$ and $\mu_j(\theta)=e_j\exp(x_j^\top\theta)$ in Example~\ref{ex:calibration-poisson}, with baseline intensities logarithmically spaced from $0.0008$ to $0.016$.
\texttt{RECAP} learns from the standardized twelve-dimensional parameter vector and selects its dimension from $\{1,\ldots,6\}$ using the coverage MAE in Supplement Section \ref{subsec:supp-mae-selection-statement}, evaluated at $24$ independently sampled proposal parameters with $256$ independent observations per parameter.
These selection locations and observations are separate from the twelve-target final evaluation.
The selected dimension is $\widehat s=3$.
We compare to the chi-square benchmark which uses the asymptotic distribution for this problem. 
Supplement Section~\ref{subsec:supp-sim-poisson} gives further details.

In Figure~\ref{fig:simulation-poisson}, \texttt{RECAP} remains close to nominal coverage, with small parameter-dependent deviations, and broadly follows the exact oracle in power.
Ambient-space calibration, \texttt{CatBoost}, the chi-square benchmark, and the plug-in bootstrap systematically undercover, and \texttt{TRUST++} also performs poorly across the target grid. In this example, the mean-regression criterion in \texttt{TRUST++} is uninformative because the conditional mean of the score is constant across parameters, even though its higher moments and quantiles vary with $\vartheta$.

Supplement Figure~\ref{fig:supp-simulation-poisson-summary}(a) shows that \texttt{RECAP} has small undercoverage and overcoverage contributions and substantially smaller MAE than the other nonoracle methods.
Including representation learning, dimension selection, and final evaluation, \texttt{RECAP} uses $276{,}144$ simulated datasets, the same per-method total as \texttt{TRUST++} and \texttt{CatBoost} and far below the approximately $30$ million used by the bootstrap (panel~(b)).

In Supplement Section~\ref{sec:supp-shared-g}, we conduct a separate Poisson experiment in which \texttt{RECAP}, \texttt{TRUST++}, and \texttt{CatBoost} share a learned calibration representation.
Their similar coverage and power suggest that the benefit of \texttt{RECAP} in this experiment comes primarily from exploiting low-dimensional structure in the null score distribution, and the particular calibration method used within that space seems less important.

\begin{figure}[t]
    \centering
    \includegraphics[width=0.8\textwidth]{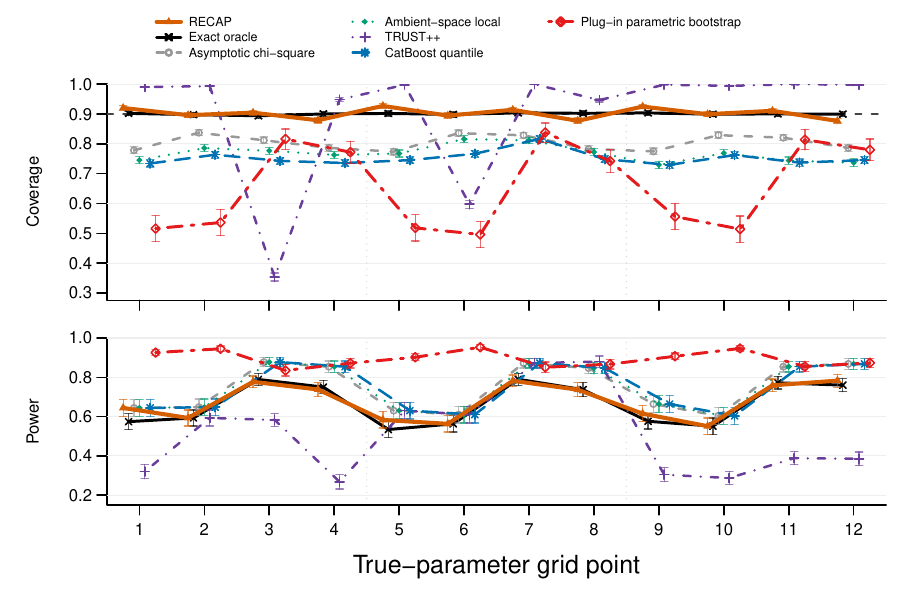}
    \caption{Poisson regression: pointwise coverage and power.
    }
    \label{fig:simulation-poisson}
\end{figure}

\section{Case Study: Coalescent Models in Population Genomics}
\label{sec:real}

Demographic inference in population genetics seeks to recover population-size histories, split times, and migration patterns from genetic variation.  Yet, different demographic parameters can have compensating effects, making joint uncertainty for this problem difficult to characterize \citep{Gutenkunst2009}. 
Although coalescent models can be simulated efficiently \citep{Baumdicker2022}, calibrating inference separately at each candidate parameter can be prohibitively expensive in a high-dimensional model. 
To address these challenges, we use \texttt{RECAP} to learn a low-dimensional representation that allows calibration information to be shared across demographic histories. The learned representation also provides a way to summarize uncertainty in scientifically meaningful combinations of parameters that may be difficult to interpret individually.

\subsection{Data}
\label{subsec:real-data}

We analyze phased genotypes from the NYGC high-coverage GRCh38 release of the 1000 Genomes Project \citep{ByrskaBishop2022,Fairley2020}.
Specifically, we study two African--European population pairs: ESN--GBR, comprising Esan individuals from Nigeria and British individuals from England and Scotland, and GWD--IBS, comprising Gambian individuals from the Western Division of The Gambia and Iberian individuals from Spain.
Each analysis includes 20 unrelated diploid individuals from each population, giving 40 sampled genome copies per population.

For each pair, the observation $X$ contains variant positions and genotypes in four 20-kilobase genomic blocks, one from each of chromosomes 1--4.
The same four blocks are used for both population pairs.
The selection of the four blocks is based on genomic accessibility.
Supplement Table~\ref{tab:real-windows} lists the selected blocks, and Supplement Section~\ref{subsec:supp-real-data} provides further details on data preparation.

Following summary-based approaches to population-genetic inference \citep{Beaumont2002,Bazin2010,Blum2013}, we compute 86 features for each block.
These describe allele-frequency distributions within and across populations, genetic diversity and differentiation, and linkage disequilibrium, which measures associations between alleles at different genomic positions.
Taking the mean and sample standard deviation of each feature across the four blocks gives 172 features.
We then apply componentwise $\operatorname{asinh}$ transformations, standardization, and principal component analysis, retaining 30 components to obtain $Z(X)\in\mathbb R^{30}$.
A separate reference bank of 2,048 simulated parameter--observation pairs determines this preprocessing map, which is held fixed and applied to both simulated and observed data.
Supplement Section~\ref{subsec:supp-real-learning} gives further details.

\subsection{Model}
\label{subsec:real-data-model}

We consider a neutral structured coalescent model with recombination, based on the coalescent model of \citet{Kingman1982} and its extensions for recombination and population structure \citep{Hudson1983,Notohara1990}.
This model traces the ancestry of sampled genome copies backward in time.
We instantiate this model framework to the following two-population demographic history.
Let the two contemporary populations be $A$ and $B$.
We assume they remain split until $T_{\mathrm{split}}$ generations before the present, when they merge into a single ancestral population.
Population sizes are assumed to be piecewise constant over time intervals called epochs.
We consider ten epochs for each contemporary population and four for the ancestral population, resulting in 24 effective population-size parameters.
We allow migration between the two contemporary populations before $T_{\mathrm{split}}$, with separate migration rates in each direction over five time intervals, resulting in ten migration parameters.
The resulting parameter vector $\theta\in\Theta\subseteq\mathbb R^{35}$ contains 20 contemporary effective sizes, four ancestral effective sizes, one split time, and ten migration rates.

Given $\theta$, we generate $X\sim P_\theta$ using the \texttt{msprime} model \citep{Kelleher2016,Baumdicker2022}.
The model first generates random ancestral histories for the 80 chromosome copies in each genomic block, tracing their ancestry backward in time under the specified demographic history.
An ancestral lineage represents the path of a DNA segment through successive generations.
At each genomic position, two lineages merge when they reach a common ancestor.
Migration allows lineages to move between populations, while recombination causes different parts of the same genomic block to follow different lineages.
The model then introduces random mutations along these lineages, producing the variant positions and genotypes in $X$.
The four blocks are simulated independently conditional on $\theta$, with mutation and recombination rates fixed.
Supplement Figure~\ref{fig:supp-real-demographic-model} illustrates the demographic history, and Supplement Section~\ref{subsec:supp-real-model} gives further details on the model and simulation settings.
We assess the model fit to the observed data in \ref{subsec:supp-real-data}.

\subsection{Inference}
\label{subsec:real-application}

Following simulation-based regression approaches \citep{Cabras2015}, we construct a discrepancy between $Z(X)$ and its predicted mean under each candidate $\theta$.
Write $u_j(\theta)=\log(\theta_j/L_j)/\log(U_j/L_j)$ for the normalized log-parameter coordinates, where $L_j$ and $U_j$ are the lower and upper bounds of the support of $\theta_j$ respectively, 
We use the statistic
\begin{equation}
\label{eq:real-score}
T(X,\theta)
=
\left[
\frac{1}{30}\sum_{\ell=1}^{30}
\left\{
\frac{Z_\ell(X)-\widehat m_\ell(u(\theta))-\widehat b_\ell}{\widehat\tau_\ell}
\right\}^{2}
\right]^{1/2},
\end{equation}
where $\widehat m(u(\theta))$ approximates $\mathbb E_\theta[Z(X)]$ using quadratic ridge regression fitted to simulated parameter--summary pairs, and $\widehat b_\ell$ and $\widehat\tau_\ell$ are residual corrections and scales estimated from separate simulations.
All fitted quantities are held fixed before representation learning and calibration.
Supplement Figure~\ref{fig:supp-real-workflow} summarizes the construction.

We learn a calibration representation following the procedure in Section~\ref{subsec_learning_representation}.
Selecting the dimension from $s\in\{1,2,3,4,6,8\}$ on independent simulations gives $\widehat s=1$ and
\begin{equation}
\label{eq:real-representation}
g(\theta)=\widehat B^\top\{\sqrt{3}(2u(\theta)-\mathbf 1)\},
\end{equation}
where $\widehat B\in\mathbb R^{35}$ is the fitted unit-length direction.
Here, the linear transformation in \eqref{eq:real-representation} centers each normalized log-parameter coordinate at zero and scales it to unit variance under the uniform pilot. 
With this learned representation fixed, we generate a fresh calibration bank of $M=16{,}384$ scores.
For each candidate $\theta$, we pool scores whose parameter representations lie in a moving ball centered at $g(\theta)$.
An independent parameter-only pilot determines the radii, targeting approximately 512 pooled scores per candidate.
The adjusted empirical $0.90$ quantile in \eqref{eq:recap-cutoff} gives the cutoff $\widehat q_{0.10,\theta}$.
Supplement Sections~\ref{subsec:supp-real-learning} and \ref{subsec:supp-real-calibration} give further fitting and calibration details. 
We also evaluate RECAP's coverage and power against other comparators in Supplement Section~\ref{subsec:supp-real-validation}.

We now turn to inference using the observed data.
Treating the two population pairs, $X_1$ and $X_2$, as two independent datasets, we apply the same cutoff rule separately to each, obtaining a nominal $90\%$ confidence set $\widehat{\mathcal C}_{0.10}(X_i)$ for its demographic parameter $\theta$ as in \eqref{eq:recap-confidence-set}.
This illustrates the reusability of \texttt{RECAP}: once the statistic, representation, and calibration bank are constructed, they can be applied to future independent datasets under the same model and score without additional calibration simulations.

Following Section~\ref{subsec:representation-inference}, we summarize each confidence set by the interval enclosing its image under the learned representation $g$:
\begin{equation}
\label{eq:real-g-interval}
I_g(X_0)=\left[\inf_{\theta\in\widehat{\mathcal C}_{0.10}(X_0)}g(\theta),\,\sup_{\theta\in\widehat{\mathcal C}_{0.10}(X_0)}g(\theta)\right].
\end{equation}
Figure~\ref{fig:real-application}(a) shows numerical approximations to these intervals, which overlap substantially and are considerably narrower than the full range of $g(\theta)$ over $\Theta$.
By checking the coefficients in $\hat B$, we find that the learned $g(\theta)$ is dominated by two coordinates of $\theta$: deep ancestral size $N_{e,\mathrm{anc},4}$ and split time $T_{\mathrm{split}}$.
$T_{\mathrm{split}}$ is the number of generations before present at which the two contemporary populations merge into their common ancestor when traced backward in time.
$N_{e,\mathrm{anc},4}$ is the effective population size in the oldest ancestral epoch.

To interpret these two parameters jointly, Figures~\ref{fig:real-application}(b,c) show acceptance slices with the other 33 normalized coordinates fixed at $0.5$.
For both population pairs, the downward-sloping lower boundaries indicate that, under this fixed background, a more recent split requires a larger deep ancestral size.
This time--size compensation is consistent with earlier coalescent analyses, where increasing ancestral population size can offset a shorter split time while preserving average sequence divergence \citep{Mailund2011}.
Human demographic studies likewise emphasize the dependence of split-time inference on other demographic parameters, including migration \citep{Gutenkunst2009}.

This application illustrates how \texttt{RECAP} uses a learned low-dimensional representation to share calibration simulations across demographic histories and population pairs and to summarize uncertainty in scientifically meaningful combinations of parameters.

\begin{figure}[tbp]
\centering
\subfigure[Intervals for $g(\theta)$.\label{fig:coal35-g-intervals}]{%
\includegraphics[width=0.3\linewidth]{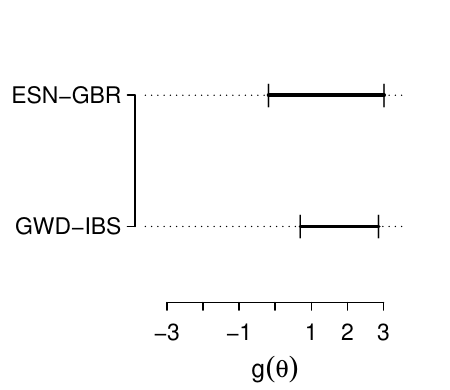}}
\hfill
\subfigure[ESN--GBR.\label{fig:coal35-esn-conditional}]{%
\includegraphics[width=0.3\linewidth]{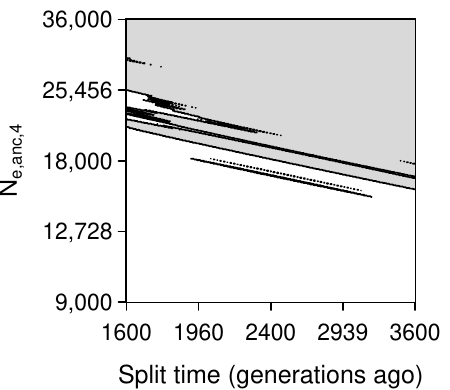}}
\hfill
\subfigure[GWD--IBS.\label{fig:coal35-gwd-conditional}]{%
\includegraphics[width=0.3\linewidth]{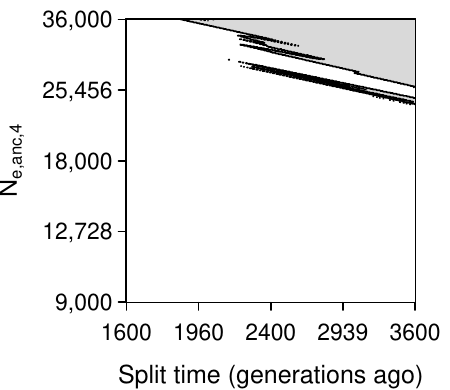}}
\caption{Coalescent-model inference for the two population pairs at nominal coverage $0.90$.
(a) Numerically approximated intervals for $g(\theta)$.
(b,c) Acceptance slices over split time and deep ancestral size, with all other normalized parameters fixed at $0.5$.
Gray denotes acceptance regions.
Supplement Section~\ref{subsec:supp-real-inversion} describes how the confidence sets are approximated numerically.}
\label{fig:real-application}
\end{figure}

\section{Discussion}

In this paper, we introduced \texttt{RECAP}, a general framework for frequentist simulation-based inference that exploits low-dimensional structure. 
The key idea is to calibrate inference in a representation space that preserves, exactly or approximately, the relevant score distribution, so that simulations can be shared across parameters that may be far apart in the original space. 
We considered representations derived from model structure as well as representations learned from simulations through sufficient dimension reduction, and developed finite-simulation coverage guarantees that account for representation error, local pooling, and calibration-bank size. 
In numerical experiments, \texttt{RECAP} provided accurate calibration and competitive power. 
The population genomics application further illustrated the practical use of \texttt{RECAP} in a high-dimensional scientific model and yielded scientifically meaningful findings consistent with existing demographic understanding.

The scope of \texttt{RECAP} extends beyond the models studied in this paper.
SBI is useful when data can be generated from a statistical model, but the likelihood or the sampling distribution of a chosen test statistic is difficult to evaluate.
A simulator is simply a computational way to sample from that model.
Requiring it to describe the data adequately is a modeling requirement shared with all model-based inference, rather than a restriction specific to SBI \citep{cranmer2020frontier}.
The practical requirement is that independent datasets can be simulated at candidate parameter values at an affordable computational cost.
Within this setting, \texttt{RECAP} is particularly useful when the parameter space is high-dimensional but the null score distribution varies mainly through a low-dimensional representation.
Potential applications include a broad range of scientific settings where complex simulators are routinely used, such as particle physics, weather forecasting, climate science, and astrophysics.
In such applications, scientific knowledge or independent training simulations could guide the construction of a calibration representation, allowing \texttt{RECAP} to share simulations across parameters with similar null score distributions.

Several directions remain for future work.
As illustrated by the matrix-completion experiment in Section~\ref{subsec:simulation-matrix-completion}, linear projections may fail to capture useful calibration structure in certain settings. 
Thus, developing nonlinear representation-learning methods is a natural next step.
Alternative pooling rules and calibration methods also deserve study, with attention to their effects on coverage, power, and computational cost.
When only one or a subset of parameters are of interest, an important extension is to develop calibration procedures that account for the remaining nuisance parameters without first constructing a full joint confidence set.
Another useful direction is to understand when and how an existing simulation bank can be safely reused after changing the model or test statistic.
Coverage bounds that account for changes in the null score distribution could clarify when such reuse is justified and when additional simulations or recalibration are needed.
Finally, further applications across different scientific domains would help assess how \texttt{RECAP} performs in a wider range of models and practical settings.

\subsubsection*{Generative AI Use Statement}

Generative AI tools (GPT-5.6 and GPT-6) were used solely in simulation studies, real-data analyses, English-language editing, and proofreading. All AI-assisted outputs were reviewed by the authors, who take full responsibility for the full manuscript.


\begingroup
\spacingset{1.45}
\setlength{\bibsep}{0pt}
\bibliographystyle{agsm}
\bibliography{bib}
\endgroup

\clearpage

\appendix

\begin{center}
{%
  \begin{minipage}{0.88\textwidth}
    \centering
    {\LARGE\bfseries Supplementary Material for}
    {\LARGE\bfseries RECAP: Representation-Based Calibration with Amortized Pooling for Frequentist Simulation-Based Inference}
  \end{minipage}%
}

\ifnum\blind=0\relax
\vspace{1.5em}
{\large Runshi Tang \quad and \quad Genevera I. Allen}\\[0.5em]
Department of Statistics, Columbia University, USA
\fi
\end{center}

\bigskip

\section{Preliminaries}

\subsection{Sufficient dimension reduction}
\label{subsec_sdr}

For a predictor $X\in\mathbb R^p$ and a response $Y\in\mathbb R$, sufficient dimension reduction (SDR) seeks a lower-dimensional representation of $X$ that preserves the conditional distribution of $Y$ \citep{cook1998regression,li2018sdr}.
In linear SDR, $\operatorname{span}(B)$ is a dimension-reduction subspace if
\[
Y\perp\!\!\!\perp X\mid B^\top X.
\]
The \emph{central subspace}, denoted by $\mathcal S_{Y\mid X}$, is the intersection of all dimension-reduction subspaces, provided that this intersection is itself a dimension-reduction subspace.
Its dimension is called the \emph{structural dimension}.
The linearity refers to the projection $B^\top X$, not to the dependence of $Y$ on the projected predictor.

When $\mathbb E|Y|<\infty$, a mean dimension-reduction subspace instead satisfies $\mathbb E(Y\mid X)=\mathbb E(Y\mid B^\top X)$.
The \emph{central mean subspace}, denoted by $\mathcal S_{\mathbb E(Y\mid X)}$, is the intersection of all such subspaces when that intersection also preserves the conditional mean.
We assume that the subspaces discussed below exist.
Then $\mathcal S_{\mathbb E(Y\mid X)}\subseteq\mathcal S_{Y\mid X}$, with potentially strict inclusion: directions affecting conditional variance or other distributional features can be absent from the central mean subspace.

\subsection{\texttt{RMAVE} ensemble}
\label{subsec_rmave}

Minimum average variance estimation (\texttt{MAVE}) estimates the central mean subspace through local linear regression, while refined \texttt{MAVE} (\texttt{RMAVE}) updates the smoothing weights using the estimated reduced space \citep{xia2002adaptive}.
To target the full central subspace, the ensemble approach of \citet{yin2011ensemble} combines regressions for transformed responses.
A \emph{characterizing family} is a collection of transformations whose central mean subspaces together span $\mathcal S_{Y\mid X}$.
One such choice is the characteristic-function family, represented by the real-valued transformations $\cos(\tau Y)$ and $\sin(\tau Y)$ over $\tau\in\mathbb R$.
Their conditional expectations jointly determine the conditional distribution of $Y$ given $X$.

In practice, select finitely many paired sine and cosine transformations, denoted by $f_1,\ldots,f_L$.
Given independent observations $\{(X_i,Y_i)\}_{i=1}^n$ and a chosen dimension $s$, the \texttt{RMAVE} ensemble uses the local least-squares criterion
\begin{equation}
\label{eq:rmave-ensemble-objective}
\min_{\substack{B^\top B=I_s\\ \{a_{j\ell},b_{j\ell}\}}}
\sum_{\ell=1}^L\sum_{j=1}^n\sum_{i=1}^n
w_{ij}
\left\{
f_\ell(Y_i)-a_{j\ell}-b_{j\ell}^\top B^\top(X_i-X_j)
\right\}^2,
\end{equation}
where $B\in\mathbb R^{p\times s}$, $a_{j\ell}\in\mathbb R$, and $b_{j\ell}\in\mathbb R^s$.
At a current basis $\widetilde B$, the normalized weights satisfy $w_{ij}\propto K_h(\widetilde B^\top(X_i-X_j))$ and $\sum_{i=1}^n w_{ij}=1$, where $K_h$ is a nonnegative kernel with bandwidth $h$.
The procedure alternates between updating the local regression coefficients and the common basis with the weights fixed, then recomputes the weights in the updated reduced space.
Unlike ordinary \texttt{RMAVE} applied only to $Y$, this ensemble can detect directions that affect transformed conditional means even when they do not affect $\mathbb E(Y\mid X)$.

\section{Additional calibration representations from orthogonal invariance}
\label{sec:supp-orthogonal-invariance}

We present three additional examples in which orthogonal invariance yields exact calibration representations through Proposition~\ref{prop:invariant-score-representation}.
Throughout, sample sizes, dimensions, and tuning parameters are fixed.

\begin{Example}[Principal component analysis]
\label{ex:calibration-spiked-covariance}
Consider the Gaussian spiked covariance model \citep{johnstone2001distribution}
\[
X_i\overset{\mathrm{i.i.d.}}{\sim}N_p(0,\Sigma),
\qquad
\Sigma=\sigma^2\{I_p+U\operatorname{diag}(\kappa_1,\ldots,\kappa_r)U^\top\},
\qquad i=1,\ldots,n,
\]
where $U^\top U=I_r$ and $\kappa_1\geq\cdots\geq\kappa_r>0$.
For the sample principal subspace $\widehat U$, consider the projection-loss score
\[
T(X,\Sigma)=\|\widehat U\widehat U^\top-UU^\top\|_{\mathrm F}.
\]
Orthogonal rotations preserve this loss, while a common rescaling of the observations leaves the sample principal subspace unchanged.
Consequently, $g(\Sigma)=(\kappa_1,\ldots,\kappa_r)^\top$ is an exact calibration representation.
The score distribution depends only on the $r$ spike strengths, not on the orientation of the population subspace or the overall scale $\sigma^2$.
\end{Example}

\begin{Example}[Canonical correlation analysis]
\label{ex:calibration-canonical-correlation}
Consider the canonical correlation analysis model \citep{hotelling1936relations} with independent observations from a centered jointly Gaussian model after whitening both marginal covariance matrices.
Write the cross-covariance matrix as
\[
\operatorname{Cov}(X,Y)=U\operatorname{diag}(\rho_1,\ldots,\rho_r)V^\top,
\]
where $U$ and $V$ have orthonormal columns and $\rho_1,\ldots,\rho_r$ are the ordered nonzero canonical correlations.
For an equivariant estimator and a canonical-subspace loss invariant under separate orthogonal rotations of the two coordinate systems, the score distribution depends only on the canonical correlations.
Thus, $g(\theta)=(\rho_1,\ldots,\rho_r)^\top$ is an exact calibration representation, irrespective of the canonical directions.
\end{Example}

\begin{Example}[Matrix regression with a random design]
\label{ex:calibration-matrix-sensing}
Consider the matrix regression model \citep{rechtfazelparrilo2010guaranteed,candesplan2011tight}
\[
y_i=\langle A_i,M\rangle+\sigma\varepsilon_i,
\qquad i=1,\ldots,n,
\]
where $M$ has rank $r$, the vectorized design matrices $A_i$ are independent isotropic Gaussian vectors, and the noise law is fixed and independent of the design matrices.
For orthogonal $Q$ and $R$, the transformations $A_i\mapsto QA_iR^\top$ and $M\mapsto QMR^\top$ preserve both the design distribution and the inner products $\langle A_i,M\rangle$.
For an equivariant estimator $\widehat M$ and the Frobenius-loss score $\|\widehat M-M\|_{\mathrm F}$, an exact representation at a fixed noise scale $\sigma>0$ is
\[
g(M)=\left(\frac{\lambda_1(M)}{\sigma},\ldots,\frac{\lambda_r(M)}{\sigma}\right)^\top,
\]
where the nonzero singular values are ordered.
As in matrix denoising, the score distribution does not depend on the singular vectors of $M$.
This representation concerns the joint simulator that generates both the random design matrices and the responses.
\end{Example}

\section{Additional calibration representations from permutation invariance}
\label{sec:supp-permutation-invariance}

We present two additional examples in which permutation invariance contributes to exact calibration representations.
The Gaussian-mixture example also uses translation and orthogonal invariance, while the regression example combines coordinate permutations with sign changes.
Throughout, sample sizes, dimensions, and tuning parameters are fixed.

\begin{Example}[Spectral clustering in isotropic Gaussian mixtures]
\label{ex:calibration-gaussian-mixtures}
Conditionally on fixed labels $z_i\in\{1,\ldots,K\}$, suppose that
\[
Y_i=\mu_{z_i}+\sigma\varepsilon_i,
\qquad
\varepsilon_i\overset{\mathrm{i.i.d.}}{\sim}N_p(0,I_p),
\qquad i=1,\ldots,n,
\]
with fixed $\sigma>0$.
Consider a clustering procedure $\widehat z(Y)$ that is equivariant under observation permutations and invariant, up to relabeling its output, under common translations and orthogonal rotations.
Use the label-invariant misclustering score
\[
T(Y,\theta)
=
\frac{1}{n}\min_{\eta\in\mathfrak S_K}
\sum_{i=1}^n\mathbf 1\{\widehat z_i(Y)\neq\eta(z_i)\},
\]
where $\mathfrak S_K$ is the set of permutations of the component labels.
Translations and rotations remove the absolute location and orientation of the component centers, while observation permutations remove the identities of the observations assigned to each component.
An exact representation is
\[
g(\theta)
=
\operatorname{canon}\left(
n_1,\ldots,n_K,
\left\{\frac{\|\mu_k-\mu_\ell\|_2}{\sigma}\right\}_{1\leq k<\ell\leq K}
\right),
\]
where $n_k=\sum_{i=1}^n\mathbf 1\{z_i=k\}$ and $\operatorname{canon}$ selects a canonical representative under simultaneous relabeling of the component sizes and the indexed pairwise-distance array.
Thus, the high-dimensional centers and observation labels are reduced to component sizes and standardized separation geometry.
This gives an exact calibration structure for isotropic Gaussian-mixture clustering procedures with the stated invariances \citep{loffler2021optimality}.
\end{Example}

\begin{Example}[Sparse random-design regression]
\label{ex:calibration-sparse-regression}
Consider $Y=X\beta+\sigma\varepsilon$, where the entries of $X\in\mathbb R^{n\times p}$ are independent standard Gaussian variables, the noise is independent of $X$, and $\sigma>0$ is fixed.
Let $\widehat\beta$ be a signed-permutation-equivariant sparse estimator.
The joint model is invariant under simultaneous signed permutations of the design columns and coefficient coordinates.
An invariant score, such as $\|\widehat\beta-\beta\|_2$ or support-recovery error, therefore depends on $\beta$ only through its coefficient magnitudes, rather than their locations or signs.
For an $s$-sparse coefficient vector, an exact representation is
\[
g(\beta)
=
\operatorname{sort}\left(
\frac{|\beta_j|}{\sigma}:j\in\operatorname{supp}(\beta)
\right),
\]
with zeros appended to obtain $s$ coordinates when the support has fewer than $s$ elements.
The calibration dimension is therefore at most $s$, even though the unknown support ranges over subsets of $\{1,\ldots,p\}$ \citep{wainwright2009sharp}.
\end{Example}

\section{Detailed implementation of \texttt{RECAP}}
\label{sec:supp-recap-algorithm}

We provide a more detailed implementation of \texttt{RECAP} in Algorithm~\ref{alg:recap-detailed}.
The algorithm combines representation learning (Algorithm~\ref{alg:recap-representation}) and moving-ball pooling (Algorithm~\ref{alg:recap-moving-ball}) with the subsequent calibration and inference steps.
When a known representation is available, the learning step can be skipped.

\begin{algorithm}[htbp]
\caption{Detailed implementation of \texttt{RECAP} with moving-ball pooling}
\label{alg:recap-detailed}
\small
\textbf{Input:} Simulator $P_\theta$, statistic $T$, prior $\pi$, level $\alpha$, calibration budget $M$, pilot size $M_{\mathrm{pil}}$, target neighborhood mass $\rho$, and representation-learning settings if $g$ is unknown.
\begin{algorithmic}[1]
\Statex \textbf{Representation learning and pooling}
\If{$g$ is unknown}
    \State Generate an independent training bank by \eqref{eq:recap-training-bank}.
    \State Apply Algorithm~\ref{alg:recap-representation} to obtain $\widehat g$ for each candidate dimension.
\EndIf
\State For each candidate representation, construct a moving-ball pooling rule using Algorithm~\ref{alg:recap-moving-ball}, where $h(u)$ is the $\lceil M_{\mathrm{pil}}\rho\rceil$th smallest pilot distance to $u$.
\State If needed, select the representation dimension using independent validation simulations as in Supplement Section~\ref{subsec:supp-mae-selection-statement}.

\Statex \textbf{Calibration bank}
\For{$i=1,\ldots,M$}
    \State Draw $\vartheta_i\sim\pi$, simulate $X_i\sim P_{\vartheta_i}$, and store $(g(\vartheta_i),Z_i)$, where $Z_i=T(X_i,\vartheta_i)$.
\EndFor

\Statex \textbf{Calibration and inference}
\For{each candidate $\theta\in\Theta$}
    \State Compute $u=g(\theta)$ and $h(u)$.
    \State Select $\mathcal I_\theta=\{i:\|g(\vartheta_i)-u\|_2\leq h(u)\}$ and set $M_\theta=|\mathcal I_\theta|$.
    \State Set $k_\theta=\lceil(M_\theta+1)(1-\alpha)\rceil$.
    \State Set $\widehat q_{\alpha,\theta}$ to the $k_\theta$th smallest score among $\{Z_i:i\in\mathcal I_\theta\}$, or $+\infty$ if $k_\theta>M_\theta$.
    \State Accept $\theta$ if $T(X_0,\theta)\leq\widehat q_{\alpha,\theta}$; otherwise reject it.
\EndFor
\end{algorithmic}
\textbf{Output:} The confidence set $\widehat{\mathcal C}_\alpha(X_0)$ in \eqref{eq:recap-confidence-set}. The representation, pooling rule, and calibration bank can be reused for future independent datasets.
\end{algorithm}

\section{Heteroscedastic Poisson regression: calibration and moment structure}
\label{sec:supp-poisson}

We provide the calculations related to Example~\ref{ex:calibration-poisson}.
Throughout, $Y_j$ are independent with means $\mu_j(\theta)=ne_j\exp(x_j^\top\theta)>0$, and the design vectors $x_j$ and exposure factors $e_j$ are fixed.

\subsection*{Exact distribution and numerical calibration}

For the Pearson score $T(Y,\theta)=\sum_{j=1}^m\{Y_j-\mu_j(\theta)\}^2/\mu_j(\theta)$, independence gives
\begin{align*}
F_\theta(t)
&=
\exp\left\{-\sum_{j=1}^m\mu_j(\theta)\right\}
\sum_{k_1,\ldots,k_m\geq0}
\left\{\prod_{j=1}^m\frac{\mu_j(\theta)^{k_j}}{k_j!}\right\}\\
&\qquad\times
\mathbf 1\left\{
\sum_{j=1}^m\frac{\{k_j-\mu_j(\theta)\}^2}{\mu_j(\theta)}
\leq t
\right\}.
\end{align*}
Thus, the full mean vector determines the finite-sample score distribution.
Its exact quantile $q_{\alpha,\theta}=\inf\{t:F_\theta(t)\geq1-\alpha\}$ yields the confidence set $\mathcal C_\alpha(Y)=\{\theta:T(Y,\theta)\leq q_{\alpha,\theta}\}$, with
\[
\mathbb P_\theta\{\theta\in\mathcal C_\alpha(Y)\}
=
F_\theta(q_{\alpha,\theta})
\geq1-\alpha.
\]
The inequality can be strict because the score distribution is discrete.

A direct numerical calculation truncates each count at an integer $K_j$, producing $\prod_{j=1}^m(K_j+1)$ joint states.
Let $F_{\theta,K}(t)$ denote the resulting sum without renormalizing the retained probabilities.
Then
\[
0\leq F_\theta(t)-F_{\theta,K}(t)
\leq
\mathbb P_\theta\left\{\bigcup_{j=1}^m\{Y_j>K_j\}\right\}
\leq
\sum_{j=1}^m\mathbb P_\theta(Y_j>K_j).
\]
The first inequality follows because the omitted summands are nonnegative, and the last follows by bounding the indicator of a union by the sum of its indicators.

An exactness certificate is also available.
Choose $K_j\geq\lceil\mu_j(\theta)\rceil$ and define
\[
b_{\theta,K}
=
\min_{1\leq j\leq m}
\frac{\{K_j+1-\mu_j(\theta)\}^2}{\mu_j(\theta)}.
\]
Every omitted count vector has at least one component exceeding its truncation point, whose score contribution is at least $b_{\theta,K}$.
Since all other contributions are nonnegative, $F_{\theta,K}(t)=F_\theta(t)$ for every $t<b_{\theta,K}$.
Consequently, if the quantile obtained from $F_{\theta,K}$ is finite and lies below $b_{\theta,K}$, it equals $q_{\alpha,\theta}$.
This provides a numerical pointwise oracle when feasible, although the component distributions and their convolution generally change with the candidate parameter, requiring the numerical calibration to be recomputed across candidates and making dense confidence-set inversion computationally costly in practice.

\subsection*{Large-exposure approximation}

Fix $m$ and consider a sequence for which $\min_j\mu_j(\theta)\to\infty$, as occurs when $n\to\infty$ with each $e_j\exp(x_j^\top\theta)$ bounded away from zero.
For any fixed $t=(t_1,\ldots,t_m)^\top\in\mathbb R^m$, independence and the Poisson characteristic function give
\begin{align*}
&\mathbb E_\theta\exp\left\{
\mathrm{i}\sum_{j=1}^m t_j\frac{Y_j-\mu_j(\theta)}{\sqrt{\mu_j(\theta)}}
\right\}\\
&\qquad=
\exp\left[
\sum_{j=1}^m\mu_j(\theta)
\left\{
\exp\left(\frac{\mathrm{i}t_j}{\sqrt{\mu_j(\theta)}}\right)
-1-\frac{\mathrm{i}t_j}{\sqrt{\mu_j(\theta)}}
\right\}
\right]\\
&\qquad=
\exp\left\{-\frac{1}{2}\sum_{j=1}^m t_j^2+R_n(t)\right\},
\end{align*}
where the Taylor remainder satisfies
\[
|R_n(t)|
\leq
\frac{1}{6}\sum_{j=1}^m\frac{|t_j|^3}{\sqrt{\mu_j(\theta)}}
\longrightarrow0.
\]
The characteristic-function continuity theorem states that pointwise convergence of characteristic functions to a characteristic function continuous at the origin implies convergence in distribution.
Here, the limit is the characteristic function of $N_m(0,I_m)$, so the vector of standardized residuals converges to that distribution.
Because convergence in distribution is preserved under continuous maps, applying the sum-of-squares map yields
\[
T(Y,\theta)\overset{d}{\longrightarrow}\chi_m^2.
\]
The limiting distribution is continuous, so its $(1-\alpha)$ quantile gives asymptotic coverage $1-\alpha$.

\subsection*{Score moments and inverse-intensity summaries}

First consider $Y\sim\operatorname{Poisson}(\mu)$ and write $D=Y-\mu$ and $W=D^2/\mu$.
The centered moment-generating function is
\[
\mathbb E[\exp(tD)]
=
\exp\{\mu(e^t-1-t)\}
=
\exp\left\{\mu\sum_{\ell=2}^{\infty}\frac{t^\ell}{\ell!}\right\}.
\]
Expanding this expression and matching the coefficients of $t^2$, $t^4$, and $t^6$ gives
\[
\mathbb E[D^2]=\mu,
\qquad
\mathbb E[D^4]=\mu+3\mu^2,
\qquad
\mathbb E[D^6]=\mu+25\mu^2+15\mu^3.
\]
For example, the coefficient of $t^6$ is
\[
\frac{\mu}{6!}
+
\mu^2\left\{\frac{1}{2!4!}+\frac{1}{2(3!)^2}\right\}
+
\frac{\mu^3}{3!(2!)^3},
\]
which becomes $\mu+25\mu^2+15\mu^3$ after multiplication by $6!$.
Therefore,
\[
\mathbb E[W]=1,
\qquad
\mathbb E[W^2]=3+\mu^{-1},
\qquad
\mathbb E[W^3]=15+25\mu^{-1}+\mu^{-2}.
\]
Subtracting the appropriate products of lower moments yields
\begin{align*}
\operatorname{Var}(W)
&=\mathbb E[W^2]-\{\mathbb E[W]\}^2
=2+\mu^{-1},\\
\kappa_3(W)
&=\mathbb E[W^3]-3\mathbb E[W^2]\mathbb E[W]+2\{\mathbb E[W]\}^3
=8+22\mu^{-1}+\mu^{-2}.
\end{align*}

Returning to the regression model, define
\[
s_\ell(\theta)
=
\sum_{j=1}^m\{e_j\exp(x_j^\top\theta)\}^{-\ell},
\qquad \ell\geq1.
\]
Since $\sum_j\mu_j(\theta)^{-\ell}=n^{-\ell}s_\ell(\theta)$ and cumulants add for independent random variables, the preceding identities imply
\begin{equation}
\label{eq:poisson-score-cumulants}
\begin{aligned}
\mathbb E_\theta[T(Y,\theta)]&=m,\\
\operatorname{Var}_\theta(T(Y,\theta))&=2m+\frac{s_1(\theta)}{n},\\
\kappa_3(T(Y,\theta))&=8m+\frac{22s_1(\theta)}{n}+\frac{s_2(\theta)}{n^2}.
\end{aligned}
\end{equation}
Thus, $s_1$ and $s_2$ determine the variance and third cumulant, including their deviations from the $\chi_m^2$ values.

Higher moments require additional summaries in general.
For example, the same expansion gives
\[
\mathbb E[D^8]=\mu+119\mu^2+490\mu^3+105\mu^4.
\]
Using $\mathbb E[W]=1$, the fourth cumulant is
\[
\kappa_4(W)
=
\mathbb E[W^4]-4\mathbb E[W^3]-3\{\mathbb E[W^2]\}^2+12\mathbb E[W^2]-6
=
48+384\mu^{-1}+112\mu^{-2}+\mu^{-3}.
\]
Consequently,
\[
\kappa_4(T(Y,\theta))
=
48m+\frac{384s_1(\theta)}{n}
+\frac{112s_2(\theta)}{n^2}
+\frac{s_3(\theta)}{n^3}.
\]
The fourth cumulant also depends on $s_3(\theta)$, showing that the first two inverse-intensity summaries do not generally determine the full score distribution.

\section{Proofs and technical details for the coverage guarantees}
\label{sec:supp-coverage-theory}

Throughout this section, write $\gamma_\theta=\mathbb P_\theta\{\theta\in\widehat{\mathcal C}_\alpha(X_0)\mid\mathcal D\}$ for the conditional coverage probability, averaging over the final calibration bank and $X_0$ with $\mathcal D$ held fixed.

\paragraph{Proof structure.}
The fixed-parameter coverage arguments have five principal steps.
\begin{enumerate}[leftmargin=*,nosep]
    \item Identify the pooled-score law by independence, then transfer its exchangeable-rank coverage using the discrepancy $\Delta_\theta$.
    \item Bound the pooling discrepancy by the H\"older condition and identify the local mass by integration over a representation-space ball.
    \item Propagate subspace error through the operator-norm triangle inequality, and translate the \texttt{RMAVE} projector rate using an orthogonal alignment.
    \item Analyze the pilot radius through the distribution and moments of a uniform order statistic.
    \item Extend coverage to fixed calibration locations using quantile coupling for independent, potentially nonidentically distributed scores.
\end{enumerate}

\subsection{Proof of Theorem~\ref{thm:recap-finite-coverage}}

\begin{proof}
Fix the information $\mathcal D$ throughout the argument.
By Assumption~\ref{ass:recap-independent-stages}, conditional on $\mathcal D$, the parameter--score pairs $(\vartheta_i,Z_i)$ are i.i.d., and the observed score $T(X_0,\theta)$ has distribution $F_\theta$ and is independent of these pairs.
Let $I_i=\mathbf 1\{g(\vartheta_i)\in\mathcal N(g(\theta))\}$.
The indicators are independent Bernoulli variables with success probability $p_\theta$, so $M_\theta\sim\operatorname{Binomial}(M,p_\theta)$.
Conditional on a particular indicator vector having selected indices $j_1<\cdots<j_m$, independence of the parameter--score pairs gives
\[
\mathbb P\{Z_{j_1}\leq t_1,\ldots,Z_{j_m}\leq t_m\mid I_1,\ldots,I_M,\mathcal D\}
=\prod_{\ell=1}^mF_\theta^{\mathrm{mix}}(t_\ell).
\]
The right-hand side does not depend on the particular selected indices.
Averaging over indicator vectors with $m$ successes therefore shows that, conditional on $M_\theta=m$, the selected scores are independent with common distribution $F_\theta^{\mathrm{mix}}$.

Set $k=\lceil(m+1)(1-\alpha)\rceil$ and introduce an independent score $Z^{\mathrm{mix}}$ with this same distribution.
When $k=m+1$, the cutoff is $+\infty$ and its acceptance probability is $1=k/(m+1)$.
For $k\leq m$, attach independent continuous uniform marks to the $m+1$ scores and order the score--mark pairs lexicographically.
Exchangeability makes the rank of $Z^{\mathrm{mix}}$ uniform on $\{1,\ldots,m+1\}$, as every label has the same probability of occupying any specified rank, and exactly one label occupies it.
If this rank is at most $k$, then $Z^{\mathrm{mix}}$ is no larger than the $k$th calibration score.
Consequently,
\[
\mathbb P\{Z^{\mathrm{mix}}\leq\widehat q_{\alpha,\theta}\mid M_\theta=m,\mathcal D\}
\geq\frac{k}{m+1}.
\]
If $F_\theta^{\mathrm{mix}}$ is continuous, score ties have probability zero and the two events coincide, giving equality.

The observed score is independent of the calibration bank and has distribution $F_\theta$.
Thus, with both distribution functions extended to equal one at $+\infty$,
\[
\begin{aligned}
&\left|\mathbb P_\theta\{T(X_0,\theta)\leq\widehat q_{\alpha,\theta}\mid M_\theta=m,\mathcal D\}
-\mathbb P\{Z^{\mathrm{mix}}\leq\widehat q_{\alpha,\theta}\mid M_\theta=m,\mathcal D\}\right|\\
&\qquad=\left|\mathbb E\left[F_\theta(\widehat q_{\alpha,\theta})-F_\theta^{\mathrm{mix}}(\widehat q_{\alpha,\theta})\mid M_\theta=m,\mathcal D\right]\right|
\leq\Delta_\theta.
\end{aligned}
\]
Since $k/(m+1)\geq1-\alpha$, averaging over $M_\theta$ proves \eqref{eq:recap-coverage-lower}.
Under continuity, $0\leq k/(m+1)-(1-\alpha)\leq1/(m+1)$, so
\[
|\gamma_\theta-(1-\alpha)|\leq\Delta_\theta+\mathbb E\left[\frac{1}{M_\theta+1}\,\middle|\,\mathcal D\right].
\]
Using $1/(m+1)=\int_0^1x^m\,dx$ and the binomial expansion gives
\[
\mathbb E\left[\frac{1}{M_\theta+1}\,\middle|\,\mathcal D\right]
=\int_0^1(1-p_\theta+p_\theta x)^M\,dx
=\frac{1-(1-p_\theta)^{M+1}}{(M+1)p_\theta}.
\]
This proves \eqref{eq:recap-coverage-two-sided}.
When $p_\theta=0$, every indicator is zero almost surely, so the cutoff is $+\infty$ and $\gamma_\theta=1$.
\end{proof}

\subsection{Proof of Corollary~\ref{cor:recap-moving-ball}}

\begin{proof}
Conditional on $\mathcal D$, the definition of $H_g$ and conditioning on the pooling event give
\[
F_\theta^{\mathrm{mix}}(t)
=\mathbb E\left[H_g(t,g(\vartheta))\,\middle|\,\|g(\vartheta)-g(\theta)\|_2\leq h_\theta,\mathcal D\right].
\]
Subtract $F_\theta(t)$, add and subtract $H_g(t,g(\theta))$, and apply \eqref{eq:recap-holder}.
For every $t$,
\[
\begin{aligned}
|F_\theta^{\mathrm{mix}}(t)-F_\theta(t)|
&\leq|H_g(t,g(\theta))-F_\theta(t)|\\
&\quad+\mathbb E\left[|H_g(t,g(\vartheta))-H_g(t,g(\theta))|\,\middle|\,\|g(\vartheta)-g(\theta)\|_2\leq h_\theta,\mathcal D\right]\\
&\leq\varepsilon_{\mathrm{rep}}(g)+L_0h_\theta^\beta.
\end{aligned}
\]
Taking the supremum over $t$ proves the discrepancy bound, and Theorem~\ref{thm:recap-finite-coverage} gives both coverage conclusions.
There is only one representation-error term because the mixture is an average of the actual conditional distribution $H_g$.

For the rate statement following Corollary~\ref{cor:recap-moving-ball}, let $f_g$ be the density of $g(\vartheta)$ and $u_\theta=g(\theta)$.
Continuity and positivity at $u_\theta$ give a radius $r_0>0$ such that $f_g(u_\theta)/2\leq f_g(u)\leq2f_g(u_\theta)$ whenever $\|u-u_\theta\|_2\leq r_0$.
Let $v_s$ be the Lebesgue volume of the unit ball in $\mathbb R^s$.
Changing variables $u=u_\theta+h_\theta v$ in the density integral shows that, for $0<h_\theta\leq r_0$,
\[
\frac{f_g(u_\theta)v_s}{2}h_\theta^s\leq p_\theta\leq2f_g(u_\theta)v_sh_\theta^s.
\]
For an exact representation, \eqref{eq:recap-moving-ball-error} is therefore bounded by a constant times $h_\theta^\beta+(Mh_\theta^s)^{-1}$.
If $h_\theta\asymp M^{-1/(\beta+s)}$, then
\[
h_\theta^\beta\asymp M^{-\beta/(\beta+s)},
\qquad
\frac{1}{Mh_\theta^s}\asymp M^{-1+s/(\beta+s)}=M^{-\beta/(\beta+s)}.
\]
This proves the asserted rate.
\end{proof}

\subsection{Proof of Theorem~\ref{thm:recap-learned-coverage}}

\begin{proof}
Fix $\mathcal D$ and choose $R\in\mathcal O(s)$ attaining the infimum in \eqref{eq:recap-subspace-error}; existence follows from compactness of $\mathcal O(s)$ and continuity of the norm.
For every calibration parameter $\vartheta\in\Theta_0$ selected by the learned ball,
\[
\begin{aligned}
\|B_0^\top(\vartheta-\theta)\|_2
&=\|R^\top B_0^\top(\vartheta-\theta)\|_2\\
&\leq\|\widehat B^\top(\vartheta-\theta)\|_2
+\|\widehat B-B_0R\|_{\mathrm{op}}\|\vartheta-\theta\|_2\\
&\leq h_\theta+D_\Theta a_N.
\end{aligned}
\]
The exact representation and the H\"older condition therefore imply
\[
\sup_t|F_\vartheta(t)-F_\theta(t)|\leq L_0(h_\theta+D_\Theta a_N)^\beta.
\]
Averaging $F_\vartheta$ over the selected parameter distribution proves the first inequality in \eqref{eq:recap-subspace-transfer}.
The remaining finite-simulation conclusions follow from Theorem~\ref{thm:recap-finite-coverage}.

For $x,y\geq0$ and $0<\beta\leq1$, $(x+y)^\beta\leq x^\beta+y^\beta$.
Indeed, when $x+y>0$, put $t=x/(x+y)$ and use $t^\beta+(1-t)^\beta\geq t+(1-t)=1$; the case $x+y=0$ is immediate.
On the radius and mass event specified after Theorem~\ref{thm:recap-learned-coverage}, \eqref{eq:recap-learned-error} consequently yields
\[
|\gamma_\theta-(1-\alpha)|
\leq L_0C_h^\beta h_M^\beta+L_0D_\Theta^\beta a_N^\beta
+\frac{1}{(M+1)c_pc_h^sh_M^s}.
\]
Taking $C=\max\{L_0C_h^\beta,L_0D_\Theta^\beta,(c_pc_h^s)^{-1}\}$ proves \eqref{eq:recap-joint-error} on this event, whose probability tends to one.
If $h_M\asymp M^{-1/(\beta+s)}$, then $h_M^\beta\asymp M^{-\beta/(\beta+s)}$ and $(Mh_M^s)^{-1}\asymp M^{-1+s/(\beta+s)}=M^{-\beta/(\beta+s)}$, which gives the optimized bound after enlarging $C$ if necessary.
\end{proof}

\subsection{\texttt{RMAVE} assumptions and the subspace-rate specialization}
\label{subsec:supp-rmave-rate}

\begin{Assumption}[\texttt{RMAVE} conditions in \citet{yin2011ensemble}]
\label{ass:recap-rmave}
Fix $d,s,L$ and prespecified paired sine and cosine transformations whose central mean subspaces span $\operatorname{span}(B_0)$.
Use the \texttt{RMAVE-FC} estimator specified in Section~4.2 of \citet{yin2011ensemble}, including its initialization, trimming, and bandwidth iteration, with dimension $s$.
Assume their conditions (C1)--(C5), with $X,Y,p,d_0,n$ replaced by $\vartheta,Z,d,s,N$, respectively, and (C3) imposed on the selected finite ensemble.
Specifically, (C1) requires bounded predictor support, bounded second derivatives of the predictor density, and bounded derivatives of the conditional first and second predictor moments near the central subspace; (C2) requires the stated fourth-order conditional-density smoothness.
Condition (C3) requires, for every $c>0$ with a nonempty constraint set,
\[
\inf_{\substack{B^\top B=I_s\\\|BB^\top-B_0B_0^\top\|_{\mathrm{op}}\geq c}}
\sum_{\ell=1}^L\mathbb E\left[\left\{\mathbb E(f_\ell(Z)\mid B^\top\vartheta)-\mathbb E(f_\ell(Z)\mid B_0^\top\vartheta)\right\}^2\right]>0.
\]
Condition (C4) requires a symmetric, compactly supported univariate kernel density with bounded second derivative.
Condition (C5) uses $b_{N,0}\asymp N^{-1/(d+4)}$ and $b_{N,j}=\max\{\zeta b_{N,j-1},b_N\}$, where $\zeta\in(1/2,1)$ and $b_N\asymp N^{-1/(s+4)}$.
\end{Assumption}

Under this assumption and the exact score-SDR model, Theorem~6.1 of \citet{yin2011ensemble} gives
\[
\|\widehat B\widehat B^\top-B_0B_0^\top\|_{\mathrm{op}}=O_{\mathbb P}\left(b_N^4+\frac{\log N}{Nb_N^s}+N^{-1/2}\right).
\]
Its population target is $\operatorname{span}(B_0)$ because the fixed ensemble spans that subspace.
To transfer this projector rate to $a_N$, take a singular-value decomposition $\widehat B^\top B_0=U\operatorname{diag}(\sigma_1,\ldots,\sigma_s)V^\top$, with $0\leq\sigma_s\leq\cdots\leq\sigma_1\leq1$, and set $R=VU^\top$.
Then
\[
(\widehat B-B_0R)^\top(\widehat B-B_0R)
=2I_s-2U\operatorname{diag}(\sigma_1,\ldots,\sigma_s)U^\top,
\]
so $a_N^2\leq2(1-\sigma_s)$.
Moreover,
\[
\begin{aligned}
1-\sigma_s^2
&=\|(I_d-\widehat B\widehat B^\top)B_0\|_{\mathrm{op}}^2\\
&=\|(B_0B_0^\top-\widehat B\widehat B^\top)B_0\|_{\mathrm{op}}^2
\leq\|B_0B_0^\top-\widehat B\widehat B^\top\|_{\mathrm{op}}^2.
\end{aligned}
\]
Since $1-\sigma_s\leq1-\sigma_s^2$, this proves
\[
a_N
\leq
\sqrt2\|\widehat B\widehat B^\top-B_0B_0^\top\|_{\mathrm{op}}
=
O_{\mathbb P}\left(
b_N^4+\frac{\log N}{Nb_N^s}+N^{-1/2}
\right).
\]
Because $b_N\asymp N^{-1/(s+4)}$,
\[
b_N^4\asymp N^{-4/(s+4)},
\qquad
\frac{\log N}{Nb_N^s}\asymp(\log N)N^{-4/(s+4)},
\]
and hence
\[
a_N
=
O_{\mathbb P}\left\{
(\log N)N^{-4/(s+4)}+N^{-1/2}
\right\}.
\]
Therefore, $a_N^\beta=O_{\mathbb P}\!\left[\{(\log N)N^{-4/(s+4)}+N^{-1/2}\}^\beta\right]$.
Substituting this bound into \eqref{eq:recap-joint-error} and using $h_M\asymp M^{-1/(\beta+s)}$, so that both $h_M^\beta$ and $(Mh_M^s)^{-1}$ are of order $M^{-\beta/(\beta+s)}$, gives
\[
|\gamma_\theta-(1-\alpha)|
=
O_{\mathbb P}\left[
\left\{(\log N)N^{-4/(s+4)}+N^{-1/2}\right\}^\beta
+
M^{-\beta/(\beta+s)}
\right],
\]
which is \eqref{eq:recap-rmave-joint-rate}.
For $s\leq3$, $4/(s+4)>1/2$, and therefore $(\log N)N^{1/2-4/(s+4)}\to0$.

\subsection{Pilot-selected radii}
\label{subsec:supp-pilot-radii}

The following lemma links Algorithm~\ref{alg:recap-moving-ball} to the radius and local-mass conditions used to derive \eqref{eq:recap-joint-error}.

\begin{Lemma}[Pilot mass and radius]
\label{lem:recap-pilot-mass}
Fix $g$ and $\theta$ independently of the pilot, and let
\[
A_\theta(r)=\mathbb P\{\|g(\vartheta)-g(\theta)\|_2\leq r\},\qquad\vartheta\sim\pi.
\]
Assume the distance distribution is continuous, including no atom at zero.
Set $m=M_{\mathrm{pil}}$, choose $\rho\in(0,1)$ before drawing the pilot, and let $k=\lceil m\rho\rceil$.
For the radius in \eqref{eq:moving-ball-radius}, the true pooling mass satisfies
\begin{equation}
\label{eq:recap-pilot-beta}
p_\theta=A_\theta(h_\theta)\sim\operatorname{Beta}(k,m+1-k).
\end{equation}
If $\rho\to0$ and $m\rho\to\infty$, then $p_\theta/\rho\to1$ in probability.
If, additionally, there exist constants $c_-,c_+,r_0>0$ such that
\begin{equation}
\label{eq:recap-pilot-local-mass}
c_-r^s\leq A_\theta(r)\leq c_+r^s,\qquad0<r\leq r_0,
\end{equation}
then, with probability tending to one,
\begin{equation}
\label{eq:recap-pilot-radius-order}
\frac{\rho}{2}\leq p_\theta\leq2\rho,
\qquad
\left(\frac{\rho}{2c_+}\right)^{1/s}\leq h_\theta\leq\left(\frac{2\rho}{c_-}\right)^{1/s}.
\end{equation}
For an independently learned $g$, the distributional statement holds conditional on its training information; the asymptotic conclusions also hold if continuity and \eqref{eq:recap-pilot-local-mass} hold with probability tending to one, with the same constants.
\end{Lemma}

\begin{proof}
Write $D_j=\|g(\vartheta_j^{\mathrm{pil}})-g(\theta)\|_2$ and $V_j=A_\theta(D_j)$.
For $v\in(0,1)$, continuity gives a radius $r_v$ with $A_\theta(r_v)=v$.
Flat portions of $A_\theta$ have zero distance probability, so $\mathbb P\{A_\theta(D_j)\leq v\}=\mathbb P(D_j\leq r_v)=v$.
Thus, $V_1,\ldots,V_m$ are independent uniform variables.
Monotonicity of $A_\theta$ gives $p_\theta=A_\theta(D_{(k)})=V_{(k)}$.
For $0<v<1$,
\[
\mathbb P(V_{(k)}\leq v)=\sum_{j=k}^m\binom mjv^j(1-v)^{m-j}.
\]
Differentiating this finite sum and canceling adjacent terms gives the density
\[
\frac{m!}{(k-1)!(m-k)!}v^{k-1}(1-v)^{m-k},
\]
which proves \eqref{eq:recap-pilot-beta}.
Integrating against this density, using $\int_0^1v^a(1-v)^b\,dv=a!b!/(a+b+1)!$ for nonnegative integers $a,b$, gives
\[
\mathbb E(p_\theta)=\frac{k}{m+1},
\qquad
\mathbb E(p_\theta^2)=\frac{k(k+1)}{(m+1)(m+2)},
\qquad
\operatorname{Var}(p_\theta)=\frac{k(m+1-k)}{(m+1)^2(m+2)}.
\]
The integral identity follows by repeated integration by parts, reducing $a$ to zero and then integrating $(1-v)^{a+b}$.
Since $m\rho\leq k<m\rho+1$, $|k/(m+1)-\rho|\leq1/(m+1)$ and $\operatorname{Var}(p_\theta)\leq(m\rho+1)/m^2$.
Consequently,
\[
\mathbb E\left[\left(\frac{p_\theta}{\rho}-1\right)^2\right]
\leq\frac{1}{m\rho}+\frac{2}{m^2\rho^2}\longrightarrow0.
\]
For every $\epsilon>0$, the probability of $|p_\theta/\rho-1|>\epsilon$ is at most this expectation divided by $\epsilon^2$, because the corresponding indicator is bounded by $(p_\theta/\rho-1)^2/\epsilon^2$.
This proves the mass convergence.

Since $\rho\to0$ and $A_\theta(r_0)\geq c_-r_0^s>0$, the mass convergence implies $p_\theta<A_\theta(r_0)$ with probability tending to one.
Monotonicity then implies $h_\theta<r_0$.
On this event, \eqref{eq:recap-pilot-local-mass} gives $(p_\theta/c_+)^{1/s}\leq h_\theta\leq(p_\theta/c_-)^{1/s}$.
Combining this with $\rho/2\leq p_\theta\leq2\rho$ proves \eqref{eq:recap-pilot-radius-order}.
For learned representations, repeat the argument conditional on training.
The displayed moment bound is independent of $g$, and the assumed common local-mass constants make the radius argument uniform over the training event; its complement has probability tending to zero.
\end{proof}

For a known exact representation satisfying Corollary~\ref{cor:recap-moving-ball} and continuous pooled-score distributions, \eqref{eq:recap-pilot-radius-order} gives, with probability tending to one,
\[
|\gamma_\theta-(1-\alpha)|\leq L_0\left(\frac{2\rho}{c_-}\right)^{\beta/s}+\frac{2}{(M+1)\rho}.
\]
For the learned linear representation in Theorem~\ref{thm:recap-learned-coverage}, under the learned-map conditions of Lemma~\ref{lem:recap-pilot-mass} and continuity of $F_\theta^{\mathrm{mix}}$, the same reasoning gives, with probability tending to one,
\[
|\gamma_\theta-(1-\alpha)|\leq L_0D_\Theta^\beta a_N^\beta+L_0\left(\frac{2\rho}{c_-}\right)^{\beta/s}+\frac{2}{(M+1)\rho}.
\]
Choosing $\rho\asymp M^{-s/(\beta+s)}$ makes both calibration terms of order $M^{-\beta/(\beta+s)}$.
The pilot requirement is then $M_{\mathrm{pil}}M^{-s/(\beta+s)}\to\infty$.

\subsection{Learned-representation coverage under fixed simulation designs}
\label{sec:supp-fixed-design}

We extend Theorem~\ref{thm:recap-learned-coverage} to calibration simulations at prescribed locations, allowing arbitrary replication at each location.
Retain the compact region $\Theta_0$, the orthonormal matrices $B_0,\widehat B\in\mathbb R^{d\times s}$, and $D_\Theta,a_N$ from \eqref{eq:recap-subspace-error}.
Let $v_1,\ldots,v_J\in\Theta_0$ be calibration locations with positive integer replication counts $r_1,\ldots,r_J$ and total budget $M=\sum_{j=1}^Jr_j$.
The design can be deterministic or randomized; for a randomized design, we condition on its realized locations.

Let $\mathcal D_{\mathrm{fix}}$ collect the design and all preliminary information determining the statistic, $\widehat B$, and the radius function $h$.
These quantities are chosen without using the final calibration outcomes or $X_0$; in particular, $h$ may depend on the design locations or an independent parameter-only pilot.
Conditional on $\mathcal D_{\mathrm{fix}}$, assume $X_0\sim P_\theta$ and mutually independent calibration draws
\[
X_{jr}\sim P_{v_j},
\qquad
Z_{jr}=T(X_{jr},v_j),
\qquad
j=1,\ldots,J,\quad r=1,\ldots,r_j,
\]
independent of $X_0$.
For a fitted statistic, $F_v$ denotes its null distribution with the fitting information held fixed.

For a fixed candidate $\theta\in\Theta_0$, recall that $h_\theta=h(\widehat B^\top\theta)$ and define
\begin{equation}
\label{eq:recap-fixed-pool}
\mathcal J_\theta
=
\{j:\|\widehat B^\top(v_j-\theta)\|_2\leq h_\theta\},
\qquad
M_\theta=\sum_{j\in\mathcal J_\theta}r_j.
\end{equation}
Pool all scores at the selected locations and use \eqref{eq:recap-cutoff}, with $k_\theta=\lceil(M_\theta+1)(1-\alpha)\rceil$ and the same $+\infty$ convention.
Write
\[
\gamma_\theta^{\mathrm{fix}}
=
\mathbb P_\theta\{\theta\in\widehat{\mathcal C}_\alpha(X_0)
\mid\mathcal D_{\mathrm{fix}}\}.
\]

\begin{Proposition}[Fixed-design coverage with a learned representation]
\label{prop:recap-learned-fixed-design}
Under the preceding conditional simulation construction, suppose \eqref{eq:linear-distributional-representation} holds on $\Theta_0$ and $H_0$ satisfies \eqref{eq:recap-holder} on $B_0^\top\Theta_0$, with $H_0$ in place of $H_g$.
For each fixed $\theta\in\Theta_0$ with $M_\theta>0$,
\begin{equation}
\label{eq:recap-fixed-learned-lower}
\gamma_\theta^{\mathrm{fix}}
\geq 1-\alpha-L_0(h_\theta+D_\Theta a_N)^\beta.
\end{equation}
If $F_\theta$ is continuous, then
\begin{equation}
\label{eq:recap-fixed-learned-error}
\left|\gamma_\theta^{\mathrm{fix}}-(1-\alpha)\right|
\leq L_0(h_\theta+D_\Theta a_N)^\beta
+\frac{1}{M_\theta+1}.
\end{equation}
If $M_\theta=0$, the cutoff is $+\infty$ and $\gamma_\theta^{\mathrm{fix}}=1$.
\end{Proposition}

\begin{proof}
Condition on $\mathcal D_{\mathrm{fix}}$, and set $m=M_\theta$ and $k=k_\theta$.
If $m=0$, the conclusion is immediate.
If $k=m+1$, the cutoff is $+\infty$ and coverage equals one; moreover, $\alpha<1/(m+1)$, so both bounds hold.
Suppose $k\leq m$.

Choose $R\in\mathcal O(s)$ attaining the infimum defining $a_N$.
For every selected location $v_j$,
\[
\begin{aligned}
\|B_0^\top(v_j-\theta)\|_2
&=\|R^\top B_0^\top(v_j-\theta)\|_2\\
&\leq\|\widehat B^\top(v_j-\theta)\|_2
+\|\widehat B-B_0R\|_{\mathrm{op}}\|v_j-\theta\|_2\\
&\leq h_\theta+D_\Theta a_N.
\end{aligned}
\]
The exact representation and the H\"older condition therefore give
\[
\sup_t|F_{v_j}(t)-F_\theta(t)|
\leq L_0(h_\theta+D_\Theta a_N)^\beta
=:b_\theta.
\]
This is the same pointwise transfer step as in the proof of Theorem~\ref{thm:recap-learned-coverage}.

Relabel the selected scores as $Z_1,\ldots,Z_m$, with distribution functions $G_1,\ldots,G_m$.
Conditional independence permits the realization
\[
Z_i=G_i^{-1}(U_i),
\qquad
U_i\overset{\mathrm{i.i.d.}}{\sim}\operatorname{Unif}(0,1),
\qquad
G_i^{-1}(u)=\inf\{t:G_i(t)\geq u\}.
\]
The generalized inverse satisfies $G_i(Z_i-)\leq U_i\leq G_i(Z_i)$ almost surely.
Since $\sup_t|G_i(t)-F_\theta(t)|\leq b_\theta$,
\[
F_\theta(Z_i)
\geq G_i(Z_i)-b_\theta
\geq U_i-b_\theta.
\]
Monotonicity of $F_\theta$ and of order statistics under coordinatewise inequalities implies $F_\theta(Z_{(k)})\geq U_{(k)}-b_\theta$.
Independence of $X_0$ and the final calibration outcomes now gives
\[
\gamma_\theta^{\mathrm{fix}}
=\mathbb E\{F_\theta(Z_{(k)})\mid\mathcal D_{\mathrm{fix}}\}
\geq\frac{k}{m+1}-b_\theta
\geq1-\alpha-b_\theta,
\]
where $\mathbb E(U_{(k)})=k/(m+1)$.

If $F_\theta$ is continuous, taking left limits also gives
\[
F_\theta(Z_i)=F_\theta(Z_i-)
\leq G_i(Z_i-)+b_\theta
\leq U_i+b_\theta.
\]
The same argument yields $|\gamma_\theta^{\mathrm{fix}}-k/(m+1)|\leq b_\theta$.
Combining this with $0\leq k/(m+1)-(1-\alpha)\leq1/(m+1)$ proves \eqref{eq:recap-fixed-learned-error}.
\end{proof}

\paragraph{Approximate linear representations.}
The exact relation can be relaxed to
\begin{equation}
\label{eq:recap-fixed-linear-approximation}
\sup_{v\in\Theta_0}\sup_t
|F_v(t)-H_0(t,B_0^\top v)|\leq\varepsilon_0,
\end{equation}
where $\varepsilon_0\geq0$ and $H_0(\cdot,u)$ is a distribution function satisfying the same H\"older condition.
The preceding geometric bound and two applications of \eqref{eq:recap-fixed-linear-approximation} give, for each selected $v_j$,
\[
\sup_t|F_{v_j}(t)-F_\theta(t)|
\leq2\varepsilon_0+L_0(h_\theta+D_\Theta a_N)^\beta.
\]
Thus, the same proof gives
\[
\gamma_\theta^{\mathrm{fix}}
\geq1-\alpha-2\varepsilon_0-L_0(h_\theta+D_\Theta a_N)^\beta,
\]
and, if $F_\theta$ is continuous,
\[
\left|\gamma_\theta^{\mathrm{fix}}-(1-\alpha)\right|
\leq2\varepsilon_0+L_0(h_\theta+D_\Theta a_N)^\beta
+\frac{1}{M_\theta+1}.
\]
Here $B_0$ is an orthonormal reference basis for the approximate representation.

\paragraph{Simulation rates.}
These finite-simulation bounds require no density or spacing assumptions on the design and apply to any preliminary orthonormal estimate $\widehat B$ with $s$ columns satisfying the conditional construction.
For fixed $D_\Theta,L_0$, a rate in $M$ requires additional local design control: for example, if $h_\theta\asymp h_M\to0$, $M_\theta\geq cMh_M^s$ for a fixed $c>0$, and $Mh_M^s\to\infty$, then under continuity \eqref{eq:recap-fixed-learned-error} gives the analogue of \eqref{eq:recap-joint-error}, with $\mathcal D_{\mathrm{fix}}$ in place of $\mathcal D$.

\subsection{Coverage accuracy with dimension selection}
\label{sec:supp-mae-selection}

We specify the coverage-MAE selection procedure and analyze its calibration accuracy when validation parameters are independently sampled from $\pi$ and both calibration banks follow Algorithm~\ref{alg:recap-framework}.

\subsubsection{Selection procedure and oracle inequality}
\label{subsec:supp-mae-selection-statement}

Fix $\tau=1-\alpha\in(0,1)$.
Let $\mathcal D_0$ collect all information used to fit the statistic, a nonempty finite family of candidate representations $g_r:\Theta\to\mathbb R^r$, and their moving-ball radius functions $h_r$, for $r\in\mathcal S$.
Write $K=|\mathcal S|$ and condition on $\mathcal D_0$ throughout.
The candidate set, representations, and radius functions are fixed before the selection-calibration bank and validation parameters are generated.
For a fitted statistic, $F_\theta$ denotes its null score distribution with the fitting information held fixed.
Assume all maps, cutoff rules, and integrals below are measurable.

Generate mutually independent selection-calibration and final-calibration banks, of sizes $M_0$ and $M$, respectively, each consisting of i.i.d.\ pairs $(\vartheta,Z)$ with $\vartheta\sim\pi$ and $Z\mid\vartheta\sim F_\vartheta$.
The raw parameter--data banks are independent of $\mathcal D_0$ and $X_0$.
Independently of both banks and $\mathcal D_0$, draw $\vartheta_j^{\mathrm{val}}\overset{\mathrm{i.i.d.}}{\sim}\pi$, $j=1,\ldots,J$, and simulate $R$ conditionally independent datasets $X_{j\ell}^{\mathrm{val}}\sim P_{\vartheta_j^{\mathrm{val}}}$ at each parameter.
The validation scores use the frozen statistic.
Different candidate dimensions may share the same calibration pairs and validation observations.
Use the same candidate representation and radius function for selection and final calibration.
For a generic bank of size $m$, let $\widehat q_{\alpha,r,m}(u)$ be the cutoff obtained by pooling scores with $\|g_r(\vartheta)-u\|_2\leq h_r(u)$ and applying \eqref{eq:recap-cutoff}, including its $+\infty$ convention.
Write $\widehat q_{\alpha,r}^{\mathrm{sel}}$ and $\widehat q_{\alpha,r}^{\mathrm{fin}}$ for the cutoffs from the two banks.
Define the empirical coverage and mean absolute coverage error by
\begin{equation}
\label{eq:recap-selection-mae}
\begin{aligned}
\widehat c_{j,r}^{\mathrm{sel}}
&=\frac1R\sum_{\ell=1}^R
\mathbf 1\!\left\{T(X_{j\ell}^{\mathrm{val}},\vartheta_j^{\mathrm{val}})
\leq\widehat q_{\alpha,r}^{\mathrm{sel}}(g_r(\vartheta_j^{\mathrm{val}}))\right\},\\
\widehat L_r
&=\frac1J\sum_{j=1}^J
\left|\widehat c_{j,r}^{\mathrm{sel}}-(1-\alpha)\right|,
\qquad
\widehat s\in\arg\min_{r\in\mathcal S}\widehat L_r.
\end{aligned}
\end{equation}
Use any measurable tie-breaking rule and construct the final cutoffs from the independent final bank with $g_{\widehat s}$ and its pooling rule held fixed, without reusing selection-calibration scores or validation observations.
Any reported empirical performance is evaluated using further independent observations, and new evaluation parameters are used when assessing performance beyond a prespecified target collection.
Other criteria, such as root mean squared coverage error or maximum undercoverage, can replace MAE; validation parameters may also be prespecified or sampled to emphasize a particular region.
The guarantee below concerns the MAE rule in \eqref{eq:recap-selection-mae} with validation parameters drawn from the same $\pi$ as the calibration banks.

Let $\pi_r$ be the distribution of $g_r(\vartheta)$ under $\pi$, and define
\begin{equation}
\label{eq:recap-mae-population-cdf}
H_r(t,u)=\mathbb P\{Z\leq t\mid g_r(\vartheta)=u,\mathcal D_0\},
\qquad
q_r^0(u)=\inf\{t:H_r(t,u)\geq\tau\}.
\end{equation}
Fix versions of these conditional distributions on $g_r(\Theta)$.
Suppose $H_r(\cdot,u)$ is continuous for $\pi_r$-almost every $u$ and that, on a set of full $\pi_r$ measure, for common constants $L_0>0$ and $\beta\in(0,1]$,
\begin{equation}
\label{eq:recap-mae-holder}
\sup_{t\in\mathbb R}|H_r(t,u)-H_r(t,u')|
\leq L_0\|u-u'\|_2^\beta.
\end{equation}
Suppose also that there exist constants $\overline h_r<\infty$ and $\underline p_r>0$ such that, for $\pi_r$-almost every $u$,
\begin{equation}
\label{eq:recap-mae-local-mass}
h_r(u)\leq\overline h_r,
\qquad
p_r(u):=\mathbb P_{\vartheta\sim\pi}
\{\|g_r(\vartheta)-u\|_2\leq h_r(u)\}
\geq\underline p_r.
\end{equation}

Define the population benchmark and the integrated MAE of a realized calibrated rule by
\begin{equation}
\label{eq:recap-mae-population-loss}
\begin{aligned}
L_r^0
&=\int_\Theta
\left|F_\theta(q_r^0(g_r(\theta)))-\tau\right|\,\pi(d\theta),\\
L_r^a
&=\int_\Theta
\left|F_\theta(\widehat q_{\alpha,r}^{a}(g_r(\theta)))-\tau\right|\,\pi(d\theta),
\qquad a\in\{\mathrm{sel},\mathrm{fin}\}.
\end{aligned}
\end{equation}
Thus, $L_r^0$ measures the calibration error of the representation-conditional population quantile.
The quantity $L_r^{\mathrm{fin}}$ conditions on the realized final bank.
All distribution functions are extended to equal one at $+\infty$.
For a bank of size $m$, put
\begin{equation}
\label{eq:recap-mae-bank-error}
\omega(m)=\max_{r\in\mathcal S}
\left\{
L_0\overline h_r^\beta
+\frac{1}{2\sqrt{(m+1)\underline p_r}}
+\frac{1}{(m+1)\underline p_r}
\right\}.
\end{equation}

\begin{Proposition}[Coverage-MAE selection and fresh-bank calibration]
\label{prop:recap-mae-selection}
Under the preceding construction and assumptions, select $\widehat s$ by minimizing $\widehat L_r$ in \eqref{eq:recap-selection-mae} over $\mathcal S$, with any measurable tie-breaking rule.
Then
\begin{equation}
\label{eq:recap-mae-oracle}
\mathbb E\!\left[L_{\widehat s}^{\mathrm{fin}}\mid\mathcal D_0\right]
\leq
\min_{r\in\mathcal S}L_r^0
+\frac{1}{2\sqrt R}
+\sqrt{\frac{\log K}{2J}}
+(K+1)\omega(M_0)+\omega(M).
\end{equation}
Moreover,
\begin{equation}
\label{eq:recap-mae-selection-transfer}
\mathbb E\!\left[
\left|L_{\widehat s}^{\mathrm{fin}}-\widehat L_{\widehat s}\right|
\,\middle|\,\mathcal D_0\right]
\leq
\frac{1}{2\sqrt R}
+\sqrt{\frac{\log(2K)}{2J}}
+K\omega(M_0)+\omega(M).
\end{equation}
Both expectations average over selection calibration, validation, and final calibration, with $\mathcal D_0$ held fixed.
\end{Proposition}

The first bound is an oracle inequality for the fitted candidate family.
The second formalizes transfer from the observed selection MAE to the actual integrated MAE after recalibration.

\subsubsection{Proof of Proposition~\ref{prop:recap-mae-selection}}
\label{subsec:supp-mae-selection-proof}

\paragraph{Proof structure.}
The proof has four substantive steps.
\begin{enumerate}[leftmargin=*,nosep]
    \item Control validation bias by the binomial variance and Jensen's inequality, and control selection over finitely many candidates by a bounded-variable exponential-moment inequality proved below.
    \item Express an integrated difference between representation-based cutoff rules through $H_r$ by conditional expectation and monotonicity of cumulative distribution functions.
    \item Bound each bank's deviation from its population quantile using the H\"older assumption, uniform-order-statistic moments, and the exact inverse-moment identity for a binomial pool size.
    \item Combine the empirical minimization inequality with these bounds, using independence of the final bank to avoid an additional candidate-uniform error at deployment.
\end{enumerate}

\begin{proof}
Fix $\mathcal D_0$ and suppress it in expectations unless additional conditioning is displayed.

\textit{Step 1: validation error and empirical minimization.}
We condition on the selection-calibration bank, denoted by $\mathcal C_0$.
For each candidate, write
\[
c_r(\theta)=F_\theta(\widehat q_{\alpha,r}^{\mathrm{sel}}(g_r(\theta))),
\qquad
\mu_r=\mathbb E(\widehat L_r\mid\mathcal C_0).
\]
Given a validation parameter $\theta$, the corresponding coverage estimate $\widehat c$ is the mean of $R$ independent Bernoulli variables with success probability $c_r(\theta)$.
For any such mean with success probability $c$, convexity of the absolute value, its reverse triangle inequality, and Cauchy--Schwarz give
\[
0\leq\mathbb E|\widehat c-\tau|-|c-\tau|
\leq\mathbb E|\widehat c-c|
\leq\sqrt{\frac{c(1-c)}R}
\leq\frac1{2\sqrt R}.
\]
Integrating over the independent validation parameter proves
\begin{equation}
\label{eq:recap-mae-validation-bias}
0\leq\mu_r-L_r^{\mathrm{sel}}\leq\frac1{2\sqrt R}.
\end{equation}

We use the exponential-moment bound
\begin{equation}
\label{eq:recap-mae-bounded-mgf}
\mathbb E\exp\{t(Y-\mathbb EY)\}\leq\exp(t^2/8),
\qquad t\in\mathbb R,
\end{equation}
for every random variable $Y$ taking values in $[0,1]$.
For completeness, let $\phi(t)=\log\mathbb E\exp\{t(Y-\mathbb EY)\}$.
Boundedness permits differentiation under the expectation; $\phi''(t)$ is the variance of $Y$ under the probability law tilted by $\exp(tY)$.
Under this tilted law, $Y$ still belongs to $[0,1]$, so its variance is at most its mean squared distance from $1/2$, which is at most $1/4$.
Since $\phi(0)=\phi'(0)=0$, integrating this second-derivative bound gives $\phi(t)\leq t^2/8$ for every real $t$, proving \eqref{eq:recap-mae-bounded-mgf}.

For each fixed $r$, the $J$ quantities $|\widehat c_{j,r}^{\mathrm{sel}}-\tau|$ are independent given $\mathcal C_0$ and belong to $[0,1]$.
Their independence follows from independent validation parameters and independent simulator draws across parameters.
Applying \eqref{eq:recap-mae-bounded-mgf} to their centered average gives, for $D_r=\mu_r-\widehat L_r$,
\[
\mathbb E\{\exp(tD_r)\mid\mathcal C_0\}\leq\exp\{t^2/(8J)\},
\qquad t\in\mathbb R.
\]
For $t>0$, the pointwise inequality $\exp(t\max_rD_r)\leq\sum_r\exp(tD_r)$ and Jensen's inequality imply
\[
\mathbb E(\max_rD_r\mid\mathcal C_0)
\leq\frac{\log K}{t}+\frac{t}{8J}.
\]
For $K>1$, minimizing at $t=\sqrt{8J\log K}$ gives
\begin{equation}
\label{eq:recap-mae-max-deviation}
\mathbb E(\max_rD_r\mid\mathcal C_0)
\leq\sqrt{\frac{\log K}{2J}},
\qquad
\mathbb E(\max_r|D_r|\mid\mathcal C_0)
\leq\sqrt{\frac{\log(2K)}{2J}}.
\end{equation}
The first inequality also holds for $K=1$, since $\mathbb E(D_1\mid\mathcal C_0)=0$.
The second follows by applying the same argument to the $2K$ variables $D_r$ and $-D_r$.

Choose $r_0\in\arg\min_rL_r^{\mathrm{sel}}$, with a fixed tie-breaking rule, so that $r_0$ is determined by $\mathcal C_0$.
By \eqref{eq:recap-mae-validation-bias} and empirical minimization,
\[
L_{\widehat s}^{\mathrm{sel}}
\leq\mu_{\widehat s}
\leq\widehat L_{\widehat s}+\max_rD_r
\leq\widehat L_{r_0}+\max_rD_r.
\]
Taking expectations over validation and applying \eqref{eq:recap-mae-max-deviation} gives
\begin{equation}
\label{eq:recap-mae-frozen-bank-oracle}
\mathbb E(L_{\widehat s}^{\mathrm{sel}}\mid\mathcal C_0)
\leq\min_rL_r^{\mathrm{sel}}
+\frac1{2\sqrt R}+\sqrt{\frac{\log K}{2J}}.
\end{equation}
The same bias bound and the second inequality in \eqref{eq:recap-mae-max-deviation} give
\begin{equation}
\label{eq:recap-mae-frozen-bank-gap}
\mathbb E\!\left[\max_r|\widehat L_r-L_r^{\mathrm{sel}}|\mid\mathcal C_0\right]
\leq\frac1{2\sqrt R}+\sqrt{\frac{\log(2K)}{2J}}.
\end{equation}

\textit{Step 2: integrated cutoff comparison.}
Fix a candidate $r$ and a realization of a calibration bank.
For two measurable extended-real cutoff functions $a(u)$ and $b(u)$, the difference $F_\theta(a(u))-F_\theta(b(u))$ has a common sign over all parameters satisfying $g_r(\theta)=u$, because each distribution function is nondecreasing.
Conditional expectation at $g_r(\vartheta)=u$ therefore gives
\begin{equation}
\label{eq:recap-mae-threshold-identity}
\begin{aligned}
&\int_\Theta
|F_\theta(a(g_r(\theta)))-F_\theta(b(g_r(\theta)))|\,\pi(d\theta)\\
&\qquad=
\int |H_r(a(u),u)-H_r(b(u),u)|\,\pi_r(du).
\end{aligned}
\end{equation}
The representation-conditional law is unchanged by conditioning on the bank because the fresh pair used to define it is independent of that bank.
This argument can equivalently be written by conditioning the indicator of the interval between $a(u)$ and $b(u)$, so it also applies to measurable cutoffs varying with $u$.
Continuity and properness of $H_r(\cdot,u)$ imply that its $\tau$-quantile $q_r^0(u)$ is finite and satisfies $H_r(q_r^0(u),u)=\tau$ for $\pi_r$-almost every $u$.
Hence, for a generic bank of size $m$, define
\[
W_{r,m}=
\int_\Theta
\left|F_\theta(\widehat q_{\alpha,r,m}(g_r(\theta)))
-F_\theta(q_r^0(g_r(\theta)))\right|\,\pi(d\theta).
\]
Equation~\eqref{eq:recap-mae-threshold-identity} identifies this quantity with
$\int|H_r(\widehat q_{\alpha,r,m}(u),u)-\tau|\,\pi_r(du)$.

\textit{Step 3: calibration-bank approximation.}
For a representation value $u$ satisfying the assumptions, let $G_{r,u}$ be the score distribution conditional on $\|g_r(\vartheta)-u\|_2\leq h_r(u)$.
Conditioning further on $g_r(\vartheta)$ gives
\[
G_{r,u}(t)
=\mathbb E\!\left[H_r(t,g_r(\vartheta))
\,\middle|\,\|g_r(\vartheta)-u\|_2\leq h_r(u)\right].
\]
The H\"older condition thus gives
\begin{equation}
\label{eq:recap-mae-mixture-discrepancy}
\sup_t|G_{r,u}(t)-H_r(t,u)|\leq L_0\overline h_r^\beta.
\end{equation}
Continuity of $H_r(\cdot,v)$ for almost every $v$ and bounded convergence imply continuity of $G_{r,u}$.
For a bank of size $m$, the pool size $n$ is binomial with parameters $m$ and $p_r(u)$, and, conditional on $n$, the selected scores are i.i.d.\ from $G_{r,u}$, by the conditional-indicator calculation in the proof of Theorem~\ref{thm:recap-finite-coverage}.

Put $k=\lceil(n+1)\tau\rceil$.
When $k\leq n$, write each selected score as $G_{r,u}^{-1}(U_i)$ for independent $U_i\sim\operatorname{Unif}(0,1)$, where $G^{-1}(v)=\inf\{t:G(t)\geq v\}$.
This construction has the required distribution because $\mathbb P\{G^{-1}(U_i)\leq t\}=\mathbb P\{U_i\leq G(t)\}=G(t)$.
Continuity of $G$ implies $G(G^{-1}(v))=v$ for $0<v<1$, and monotonicity therefore gives $G_{r,u}(\widehat q_{\alpha,r,m}(u))\overset{d}{=}U_{(k)}$ conditional on $n$.
Here $U_{(k)}$ is the $k$th order statistic of $n$ independent uniforms.
Its distribution function is $\sum_{j=k}^n\binom njv^j(1-v)^{n-j}$ on $0<v<1$, whose derivative is
\[
\frac{n!}{(k-1)!(n-k)!}v^{k-1}(1-v)^{n-k}.
\]
Integration against this density, using the beta integral proved in Lemma~\ref{lem:recap-pilot-mass}, yields
\[
\mathbb EU_{(k)}=\frac{k}{n+1},
\qquad
\mathbb EU_{(k)}^2=\frac{k(k+1)}{(n+1)(n+2)},
\]
and therefore
\[
\operatorname{Var}(U_{(k)})
=\frac{k(n+1-k)}{(n+1)^2(n+2)}
\leq\frac1{4(n+2)}.
\]
Since $0\leq k/(n+1)-\tau\leq1/(n+1)$, Cauchy--Schwarz and the triangle inequality give
\[
\mathbb E|U_{(k)}-\tau|
\leq\frac1{2\sqrt{n+2}}+\frac1{n+1}.
\]
When $k=n+1$, including $n=0$, the cutoff is $+\infty$ and $1-\tau<1/(n+1)$, so the same upper bound holds for its transformed cutoff.
Combining this observation with \eqref{eq:recap-mae-mixture-discrepancy} gives
\[
\mathbb E\!\left[
|H_r(\widehat q_{\alpha,r,m}(u),u)-\tau|\mid n
\right]
\leq L_0\overline h_r^\beta
+\frac1{2\sqrt{n+2}}+\frac1{n+1}.
\]

For $p=p_r(u)>0$, the identity $1/(n+1)=\int_0^1x^n\,dx$ and the binomial expansion give
\[
\mathbb E\frac1{n+1}
=\int_0^1(1-p+px)^m\,dx
=\frac{1-(1-p)^{m+1}}{(m+1)p}
\leq\frac1{(m+1)\underline p_r}.
\]
Cauchy--Schwarz also gives
\[
\mathbb E\frac1{\sqrt{n+2}}
\leq\left(\mathbb E\frac1{n+2}\right)^{1/2}
\leq\frac1{\sqrt{(m+1)\underline p_r}}.
\]
Integrating over $u$ and interchanging nonnegative integrals and expectations now proves
\begin{equation}
\label{eq:recap-mae-single-bank}
\mathbb EW_{r,m}
\leq L_0\overline h_r^\beta
+\frac1{2\sqrt{(m+1)\underline p_r}}
+\frac1{(m+1)\underline p_r}
\leq\omega(m).
\end{equation}
By the triangle inequality for the absolute-value loss,
$|L_r^a-L_r^0|\leq W_r^a$, where $W_r^{\mathrm{sel}}$ and $W_r^{\mathrm{fin}}$ denote $W_{r,m}$ evaluated on the corresponding bank.
In particular,
\begin{equation}
\label{eq:recap-mae-bank-max}
\mathbb E\max_rW_r^{\mathrm{sel}}
\leq\sum_r\mathbb EW_r^{\mathrm{sel}}
\leq K\omega(M_0).
\end{equation}

\textit{Step 4: selection and fresh calibration.}
Choose $r_*\in\arg\min_rL_r^0$, which is fixed conditional on $\mathcal D_0$.
The final bank is independent of $\widehat s$ and all selection information, so \eqref{eq:recap-mae-single-bank} implies
\[
\mathbb EW_{\widehat s}^{\mathrm{fin}}
=\sum_r\mathbb P(\widehat s=r)\mathbb EW_r^{\mathrm{fin}}
\leq\omega(M).
\]
Using \eqref{eq:recap-mae-frozen-bank-oracle} and \eqref{eq:recap-mae-bank-max}, we obtain
\[
\begin{aligned}
\mathbb EL_{\widehat s}^{\mathrm{fin}}
&\leq\mathbb EL_{\widehat s}^0+\omega(M)\\
&\leq\mathbb EL_{\widehat s}^{\mathrm{sel}}
+K\omega(M_0)+\omega(M)\\
&\leq\mathbb E\min_rL_r^{\mathrm{sel}}
+\frac1{2\sqrt R}+\sqrt{\frac{\log K}{2J}}
+K\omega(M_0)+\omega(M)\\
&\leq L_{r_*}^0
+\frac1{2\sqrt R}+\sqrt{\frac{\log K}{2J}}
+(K+1)\omega(M_0)+\omega(M),
\end{aligned}
\]
where the last line uses $\min_rL_r^{\mathrm{sel}}\leq L_{r_*}^{\mathrm{sel}}\leq L_{r_*}^0+W_{r_*}^{\mathrm{sel}}$ and \eqref{eq:recap-mae-single-bank}.
This proves \eqref{eq:recap-mae-oracle}.
Finally, the pointwise inequality
\[
|L_{\widehat s}^{\mathrm{fin}}-\widehat L_{\widehat s}|
\leq W_{\widehat s}^{\mathrm{fin}}
+\max_rW_r^{\mathrm{sel}}
+\max_r|L_r^{\mathrm{sel}}-\widehat L_r|
\]
and \eqref{eq:recap-mae-frozen-bank-gap} prove \eqref{eq:recap-mae-selection-transfer} after taking expectations.
\end{proof}

\paragraph{Consequences and scope.}
For fixed $K$, the remainder in \eqref{eq:recap-mae-oracle} vanishes if $J,R\to\infty$, $\max_r\overline h_r\to0$, and both $\min_rM_0\underline p_r$ and $\min_rM\underline p_r$ diverge, with $L_0$ and $\beta$ fixed.
If the smallest population-candidate MAE also tends to zero, the selected procedure's expected integrated MAE tends to zero.
For random preliminary fits, the same conclusion holds after averaging over $\mathcal D_0$ if these conditions and $\min_rL_r^0\to0$ hold in probability, because the conditional expected MAE is bounded by one and by the right-hand side of \eqref{eq:recap-mae-oracle}.
Indeed, for any bounded nonnegative quantity $V\leq1$, $\mathbb EV\leq\epsilon+\mathbb P(V>\epsilon)$, so convergence in probability to zero implies convergence of its expectation.

The fresh-bank comparison in Proposition~\ref{prop:recap-mae-selection} uses both the stated regularity assumptions and validation parameters drawn from the same $\pi$ as the calibration banks, through \eqref{eq:recap-mae-threshold-identity}.
The proposition analyzes MAE minimization over a candidate set fixed before selection. 
A validation-dependent eligibility filter, another error criterion, or changes to the fitted representation or radius function after selection require separate analysis.

\section{Additional simulation: calibration with a shared low-dimensional representation}
\label{sec:supp-shared-g}

To compare calibration rules after a representation has been learned, we repeat the Poisson experiment in Section~\ref{subsec:simulation-poisson}, providing \texttt{RECAP}, \texttt{TRUST++}, and \texttt{CatBoost} with the same learned representation $g(\theta)$.
The model, Pearson statistic, proposal, twelve target parameters, and false candidates are unchanged.
We generate new simulation banks and relearn the representation independently of the main experiment.
The common representation is learned by \texttt{CF-RMAVE}, with dimension selected from $\{1,\ldots,6\}$ using the independent-location coverage-MAE criterion for \texttt{RECAP}.
The selected dimension in this experiment is $\widehat s=2$.
The exact oracle continues to calibrate at the full candidate parameter.
Section~\ref{subsec:supp-sim-poisson-shared-g} gives the implementation details.

Figure~\ref{fig:supp-poisson-shared-g} shows that the three learned calibrators have coverage close to $0.90$ at most targets and power that broadly follows the exact oracle.
\texttt{RECAP} and \texttt{CatBoost} have similar final-grid coverage MAEs, with somewhat larger error for \texttt{TRUST++} (Figure~\ref{fig:supp-poisson-shared-g-summary}(a)).
The coverage deviations of \texttt{TRUST++} and \texttt{CatBoost} are substantially smaller than in the original-parameter comparison in Figure~\ref{fig:simulation-poisson}.
These results show that the learned representation is useful for forest-based pooling and direct quantile regression as well as for moving-ball calibration.

\begin{figure}[t]
    \centering
    \includegraphics[width=0.8\textwidth]{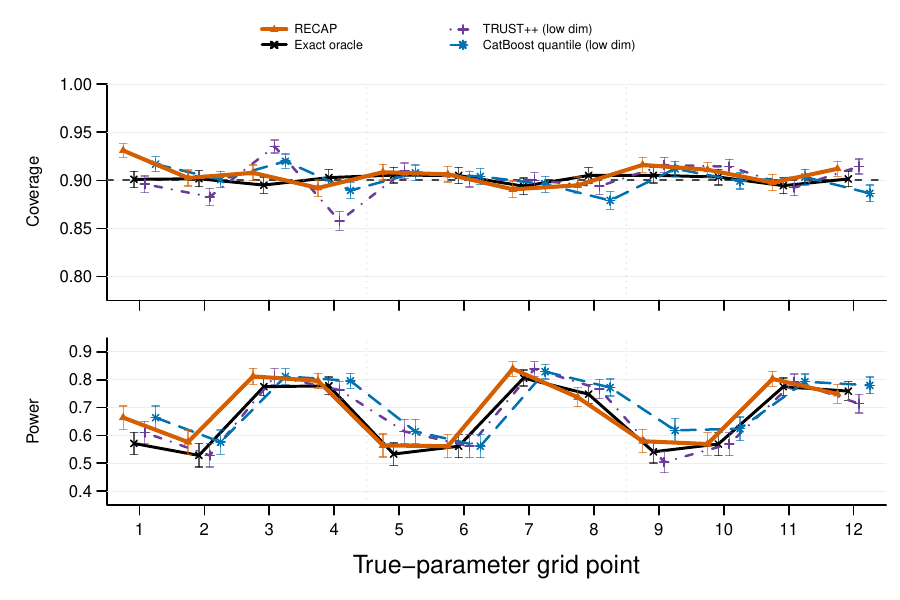}
    \caption{Poisson regression with a shared learned calibration representation: coverage and power at nominal coverage $0.90$.
    \texttt{RECAP}, \texttt{TRUST++}, and \texttt{CatBoost} use the same learned representation.
    The exact oracle uses the full Poisson model.}
    \label{fig:supp-poisson-shared-g}
\end{figure}

\begin{figure}[p]
    \centering
    \subfigure[Coverage error.]{%
        \includegraphics[width=\textwidth]{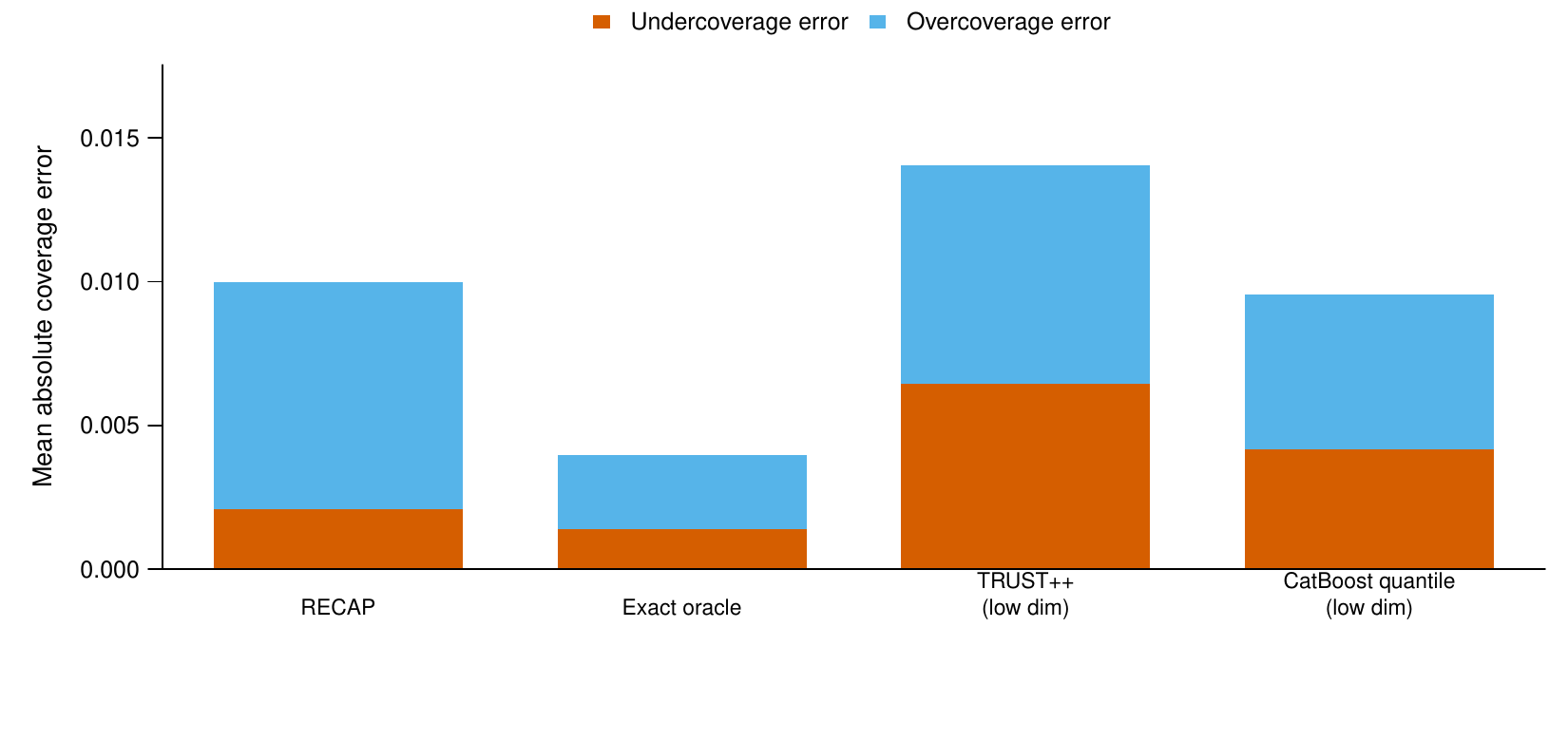}}
    \\[0.5ex]
    \subfigure[Simulation cost.]{%
        \includegraphics[width=0.78\textwidth]{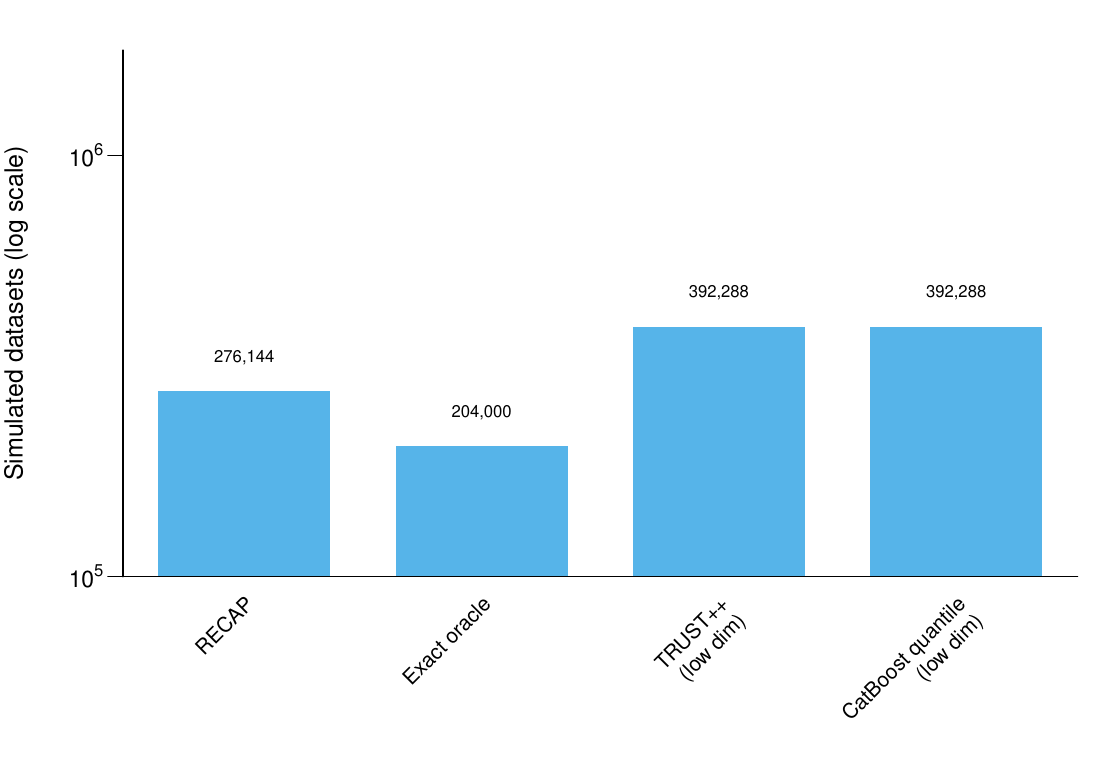}}
    \caption{Poisson regression with a shared learned calibration representation at nominal coverage $0.90$.
    (a) Mean absolute coverage error across twelve targets, decomposed into undercoverage and overcoverage contributions.
    (b) Total simulated datasets on a logarithmic scale, including representation learning, dimension selection, method-specific training and tuning, calibration, and final evaluation when required.
    The common representation-learning and selection costs are charged to each method that uses the learned representation.}
    \label{fig:supp-poisson-shared-g-summary}
\end{figure}

\section{Additional simulation: spectral clustering in the stochastic block model}
\label{subsec:simulation-sbm}

We investigate Example~\ref{ex:calibration-sbm}, spectral clustering for the Stochastic Block Model, using a simulation model with $30$ vertices in three communities of size ten, with edge probabilities $p_{\mathrm{out}}+\Delta\mathbf 1\{z_i=z_j\}$.
Spectral clustering uses the leading three adjacency eigenvectors, row normalization, and deterministic three-means clustering.
The score is the label-invariant misclustering fraction augmented by an independent $\operatorname{Unif}(0,1/30)$ variable to break ties.
\texttt{RECAP} uses the two model-based coordinates $(p_{\mathrm{out}},\Delta)$, rescaled by their proposal ranges.
The targets combine $p_{\mathrm{out}}\in\{0.08,0.12,0.16,0.20,0.24,0.28\}$ with $\Delta\in\{0.28,0.22\}$, holding a balanced community assignment fixed at each target.
The two false assignments remain balanced and have label-invariant distances $0.40$ and $0.50$ from the truth, with connection probabilities unchanged.
Supplement Section~\ref{subsec:supp-sim-sbm} specifies the clustering and alternative constructions.

In Figure~\ref{fig:simulation-sbm}, \texttt{RECAP} tracks the pointwise oracle's coverage and power.
\texttt{TRUST++} also remains reasonably close to nominal coverage, but has larger coverage deviations than \texttt{RECAP}, while \texttt{CatBoost} tends to undercover.
Ambient-space calibration is particularly conservative and less powerful.
The plug-in bootstrap shows substantial undercoverage at several targets, and increasing its budget from $1{,}200$ to $10{,}000$ does not materially change this pattern.

Supplement Figure~\ref{fig:supp-simulation-sbm-summary}(a) shows that \texttt{RECAP} has the smallest observed MAE among the nonoracle methods, with little undercoverage.
Panel~(b) reports $80{,}000$ simulated datasets for \texttt{RECAP}, compared with $112{,}144$ for each of \texttt{TRUST++} and \texttt{CatBoost} and $300{,}000$ for the pointwise oracle; the larger-budget bootstrap uses approximately $24$ million.

\begin{figure}[t]
    \centering
    \includegraphics[width=0.8\textwidth]{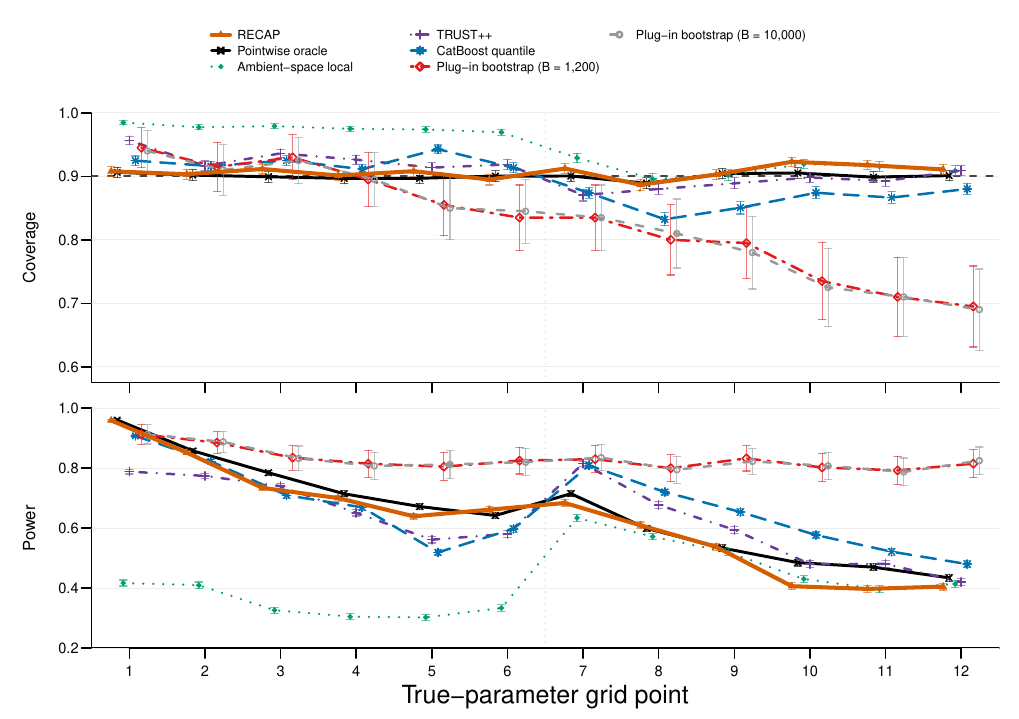}
    \caption{Stochastic block model: coverage and power.
    }
    \label{fig:simulation-sbm}
\end{figure}

\section{Simulation study: designs and implementation details}
\label{sec:supp-simulation}

This section provides further details for the four experiments in Section~\ref{sec:simulation-study}.

\subsection{Simulation banks and evaluation}

Table~\ref{tab:simulation-budgets} lists the simulation budgets.
Within each experiment, \texttt{RECAP}, ambient-space calibration, \texttt{TRUST++}, and \texttt{CatBoost} share the final score bank.
Separate simulations are used for representation learning, tuning, dimension selection, pointwise calibration, and final evaluation when needed.
\texttt{TRUST++} and \texttt{CatBoost} share their fitting and tuning banks, while the bootstrap uses new simulations for each evaluated dataset.

\begin{table}[tp]
    \centering
    \small
    \setlength{\tabcolsep}{4pt}
    \begin{tabular}{lrrrr}
        \toprule
        & Denoising & SBM & Completion & Poisson \\
        \midrule
        Final bank $M$ & $20{,}000$ & $20{,}000$ & $20{,}000$ & $100{,}000$ \\
        Parameter-only pilot & $100{,}000$ & $100{,}000$ & $100{,}000$ & $100{,}000$ \\
        Target neighborhood mass & $0.06$ & $0.06$ & $0.06$ & $0.01$ \\
        Representation training & -- & -- & -- & $5{,}000$ \\
        Representation tuning & -- & -- & -- & $5{,}000$ \\
        Dimension-selection calibration & -- & -- & -- & $100{,}000$ \\
        Dimension-selection locations & -- & -- & -- & $24$ \\
        Selection scores per location & -- & -- & -- & $256$ \\
        Comparator training & $6{,}000$ & $6{,}000$ & $6{,}000$ & $10{,}000$ \\
        Comparator tuning calibration & $20{,}000$ & $20{,}000$ & $20{,}000$ & $100{,}000$ \\
        Comparator validation locations & $24$ & $24$ & $24$ & $24$ \\
        Validation scores per location & $256$ & $256$ & $256$ & $256$ \\
        Pointwise calibration per null & $5{,}000$ & $20{,}000$ & $5{,}000$ & Exact \\
        Exact-oracle pilot per candidate & -- & -- & -- & $4{,}000$ \\
        Coverage observations per target & $1{,}000$ & $5{,}000$ & $1{,}000$ & $5{,}000$ \\
        Bootstrap observations per target & $200$ & $200$ & $200$ & $500$ \\
        Power observations, nonbootstrap & $200$ & $5{,}000$ & $200$ & $500$ \\
        Power observations, bootstrap & $200$ & $200$ & $200$ & $500$ \\
        Bootstrap resamples per observation & $1{,}200/20{,}000$ & $1{,}200/10{,}000$ & $1{,}200$ & $5{,}000$ \\
        \bottomrule
    \end{tabular}
    \caption{Simulation budgets and evaluation sizes.
    Entries count simulated datasets unless a row explicitly concerns parameter-only draws, validation locations, or neighborhood mass.
    A dash indicates that the stage is not used.
    Comparator training and tuning banks are shared by \texttt{TRUST++} and \texttt{CatBoost}.
    Bootstrap evaluation uses the first indicated observations from the corresponding final test bank.}
    \label{tab:simulation-budgets}
\end{table}

Coverage is evaluated at twelve fixed parameters using independent observations.
The displayed $95\%$ Monte Carlo intervals use the estimated coverage plus or minus $1.96$ standard errors, truncated to $[0,1]$.
For reusable cutoffs, these intervals describe evaluation uncertainty with the fitted procedures and calibration banks held fixed.
Power is averaged over two false candidates at each target; its standard error is computed from the paired rejection averages across observations.
The same test observations are used across methods where their evaluation budgets overlap.
Coverage MAE is the mean absolute deviation from $0.90$ across the twelve targets, and we report its undercoverage and overcoverage contributions separately.

\subsection{\texttt{RECAP}: pooling and representation learning}
\label{subsec:supp-sim-recap}

All experiments use moving-ball calibration with the adjusted empirical quantile in \eqref{eq:recap-cutoff}.
The pooling radii are determined by $100{,}000$ independent parameter-only draws, targeting neighborhood mass $0.06$ in denoising, the SBM, and completion, and $0.01$ in Poisson regression.
These choices give approximately $1{,}200$ and $1{,}000$ scores per cutoff, respectively.

Only Poisson regression uses a learned representation in the main experiments.
We fit the characteristic-function \texttt{RMAVE} ensemble \citep{yin2011ensemble} using sixteen frequencies between $0.20$ and $2.50$ and separate simulation banks for fitting and tuning.
We select the dimension from $\{1,\ldots,6\}$ by coverage MAE at 24 independently sampled parameters, with 256 observations per parameter and a separate selection-calibration bank.
The selected dimension is 3.
None of the dimension-selection simulations is reused in final calibration or evaluation.

\subsection{Shared comparison methods}
\label{subsec:supp-sim-comparisons}

The pointwise oracle calibrates directly at each tested parameter, using the budgets in Table~\ref{tab:simulation-budgets}.
For denoising and the SBM, parameters with the same null score distribution share a pointwise cutoff.
In Poisson regression, we instead compute the exact finite-sample quantile numerically.
Ambient-space calibration follows the same moving-ball rule as \texttt{RECAP}, but uses the original parameters rather than a reduced representation.

\texttt{TRUST++} fits a random forest to simulated scores and pools them by forest proximity \citep{cabezas2026tmlr-conformal}.
We tune its leaf size and proximity threshold using coverage MAE on independent simulations.
\texttt{CatBoost} directly estimates the conditional $0.90$ score quantile, with tree depth and regularization tuned by validation quantile loss and boosting iterations chosen by early stopping on separate simulations.
After tuning, \texttt{CatBoost} is refitted using its training bank and the shared final bank.

For the plug-in parametric bootstrap, we estimate the model parameter from each observed dataset and simulate at that estimate.
The fitted parameter comes from the truncated SVD in denoising, spectral clustering and fitted edge probabilities in the SBM, ridge alternating least squares in completion, and constrained maximum likelihood in Poisson regression.
The resulting bootstrap cutoff is used for testing both the true parameter and the two false candidates for that observation.

\subsection{Matrix denoising}
\label{subsec:supp-sim-denoising}

We generate rank-two signals $S=U\operatorname{diag}(\lambda_1,\lambda_2)V^\top$ with independent random orthonormal factors, $\lambda_2\sim\operatorname{Unif}(1.25,2.75)$, and $\lambda_1-\lambda_2\sim\operatorname{Unif}(0.5,1.5)$.
\texttt{RECAP} pools scores by Euclidean distance between the two ordered singular values.
For power, we rotate the left singular vectors to obtain false matrices at Frobenius distances $3$ and $3.5$ from the truth, leaving their singular values unchanged.
Thus, these candidates have the same null score distribution as the true matrix.

\subsection{Stochastic block model}
\label{subsec:supp-sim-sbm}

We generate balanced three-community graphs with $p_{\mathrm{out}}\sim\operatorname{Unif}(0.06,0.30)$ and $\Delta\sim\operatorname{Unif}(0.20,0.30)$, independently, and randomly permute the community labels.
The spectral-clustering procedure uses the leading three adjacency eigenvectors, normalized rows, and three-means clustering.
We add independent $\operatorname{Unif}(0,1/30)$ noise to each misclustering score to break ties.
\texttt{RECAP} pools scores using the rescaled coordinates $(p_{\mathrm{out}},\Delta)$.
The two false assignments remain balanced, have label-invariant misclustering distances $0.40$ and $0.50$ from the truth, and keep the same edge probabilities.

\subsection{Matrix completion}
\label{subsec:supp-sim-completion}

We use the rank-two matrix family described in Section~\ref{subsec:simulation-matrix-completion}, varying signal strength $\rho\sim\operatorname{Unif}(0.45,1.00)$ and row and column coherence parameters $\kappa_U,\kappa_V\sim\operatorname{Unif}(0,1.20)$ independently.
The matrices are formed from orthonormal Hadamard columns, with four rows or columns rescaled by $e^{\kappa_U}$ or $e^{\kappa_V}$ and with random row and column permutations.
The estimator starts from the rank-two SVD of the rescaled zero-filled data and applies five ridge alternating-least-squares sweeps, with ridge penalty $0.05$.

For the fixed representation, let $U$ and $V$ be the orthonormal left and right singular-vector matrices, set $\mu_U=16\max_i\|U_{i\cdot}\|_2^2$, and define $\mu_V$ in the same way from $V$.
We use the normalized coordinates
\[
\begin{aligned}
x_1&=\frac{\lambda_1(M)/32-0.45}{0.55},\\
x_2&=\frac{\log\mu_U}{\log\mu_{\max}},\qquad
x_3=\frac{\log\mu_V}{\log\mu_{\max}},\qquad
\mu_{\max}=\frac{e^{2.4}}{7/8+e^{2.4}/8}.
\end{aligned}
\]
We use $g(M)=(x_1,x_2,x_3,x_1x_2,x_1x_3,x_2x_3)^\top$ for \texttt{RECAP}, while the ambient-space and supervised comparators use the original matrix entries.
For power, we change $\rho$ by $0.12$ or $0.18$ toward the center of its range, leaving coherence and the row and column permutations fixed.

\subsection{Poisson regression}
\label{subsec:supp-sim-poisson}

We use a fixed full-rank design with two dominant directions and ten weak directions.
The baseline intensities are $e_j=\exp\{\log(0.0008)+(j-1)\log(20)/11\}$, $j=1,\ldots,12$.
In rotated coordinates, the two dominant parameters lie in $[-2,2]$ and the other ten lie in $[-9,9]$.
For calibration, we independently draw the first two coordinates from $N(0,0.8^2)$ truncated to $[-2,2]$ and the other ten from $N(0,6^2)$ truncated to $[-9,9]$.
The twelve targets combine four settings for the two dominant coordinates with three levels along a fixed nuisance direction.
For power, the two false candidates shift the first dominant coordinate by $-0.65$ and $-1$.
We standardize the original parameters before fitting \texttt{RMAVE} and the projected coordinates before pooling.

The pointwise oracle calculates the finite-sample Pearson-score quantile by enumerating the relevant joint Poisson probabilities until their cumulative mass reaches $0.90$.
The numerical calculation checks that omitted count vectors cannot change the quantile; 4,000 pilot simulations per tested parameter set the initial search range.
The chi-square comparator uses the asymptotic $\chi^2_{12,0.90}$ cutoff.
For the bootstrap, we fit the parameter by maximum likelihood within the parameter bounds and generate new Poisson observations at that fit.

\subsection{Poisson regression with a shared learned representation}
\label{subsec:supp-sim-poisson-shared-g}

We repeat the Poisson experiment with independent simulation banks and provide \texttt{RECAP}, \texttt{TRUST++}, and \texttt{CatBoost} with the same learned representation $g(\theta)$.
The representation is learned by the characteristic-function \texttt{RMAVE} ensemble, and independent coverage-MAE selection gives $\widehat s=2$.
All three methods use a new shared final bank of $100{,}000$ scores, with separate fitting and tuning simulations for \texttt{TRUST++} and \texttt{CatBoost}.
The pointwise oracle still uses the full Poisson model; the bootstrap and chi-square comparisons are omitted.

\subsection{Coverage MAE and simulation cost}
\label{subsec:supp-sim-diagnostics}

Figures~\ref{fig:supp-simulation-sbm-summary}--\ref{fig:supp-simulation-poisson-summary} show coverage MAE and simulation costs for the SBM, completion, and Poisson experiments.
The corresponding denoising results appear in Figure~\ref{fig:intro-matrix-denoising}.
MAE is split into undercoverage and overcoverage, and simulation totals include fitting, tuning, calibration, and final evaluation when used.
Parameter-only pilots are not counted; bootstrap totals cover the full repeated-evaluation experiments rather than a single observed dataset.
For Poisson regression, the exact oracle also uses 4,000 pilot simulations at each of the 36 tested parameters, although numerical enumeration itself is not counted as simulation.

\begin{figure}[p]
    \centering
    \subfigure[Coverage MAE.]{%
        \includegraphics[width=0.78\textwidth]{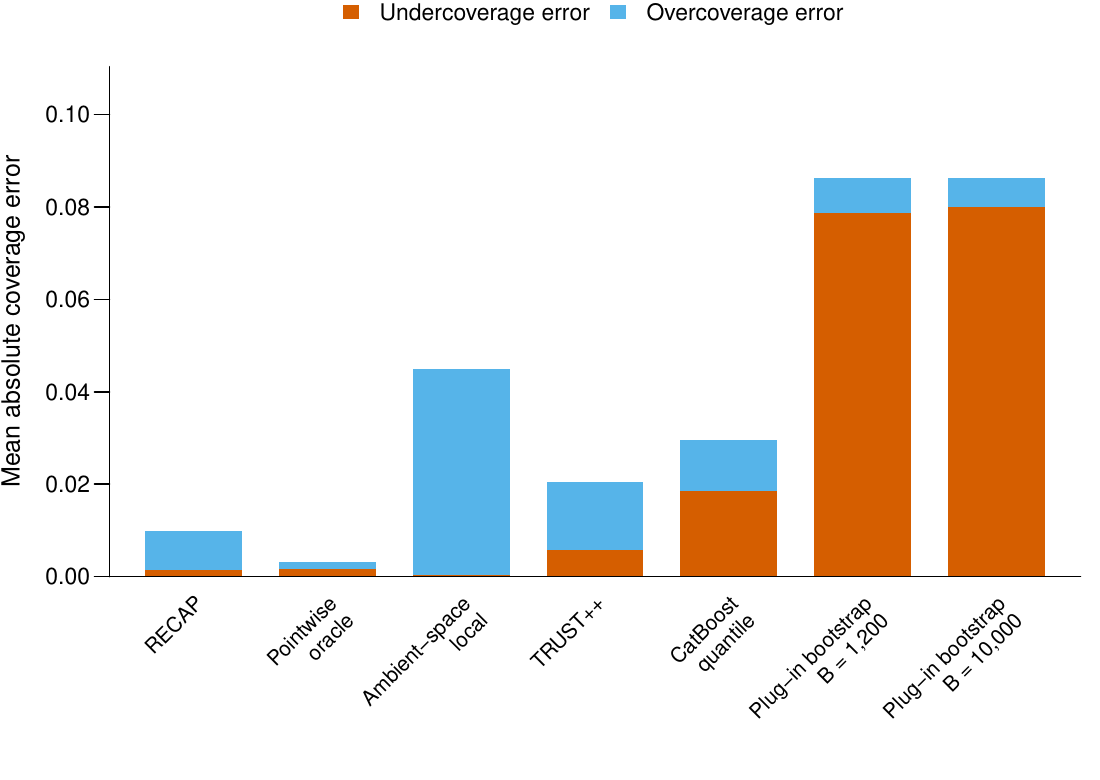}}
    \\[0.5ex]
    \subfigure[Simulation cost.]{%
        \includegraphics[width=0.78\textwidth]{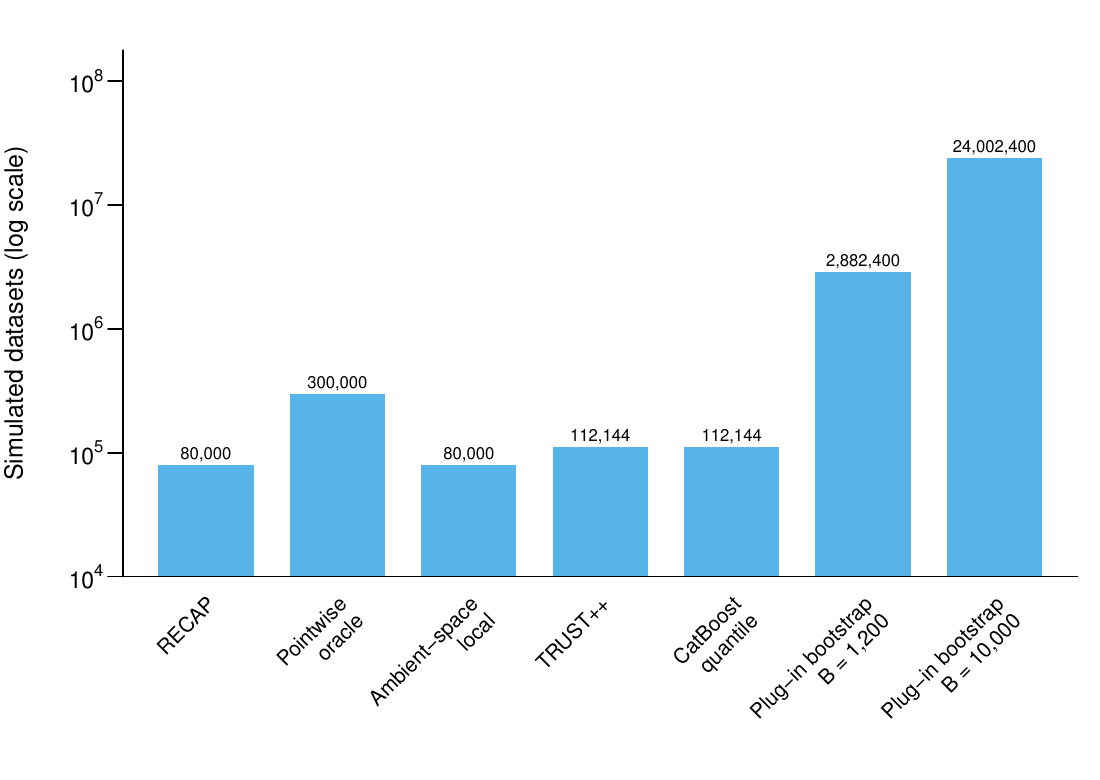}}
    \caption{Stochastic block model: (a) mean absolute coverage error at nominal level $0.90$ across twelve targets; (b) total simulated datasets on a logarithmic scale, including training, tuning, calibration, bootstrap resampling, and final evaluation.}
    \label{fig:supp-simulation-sbm-summary}
\end{figure}

\begin{figure}[p]
    \centering
    \subfigure[Coverage MAE.]{%
        \includegraphics[width=0.78\textwidth]{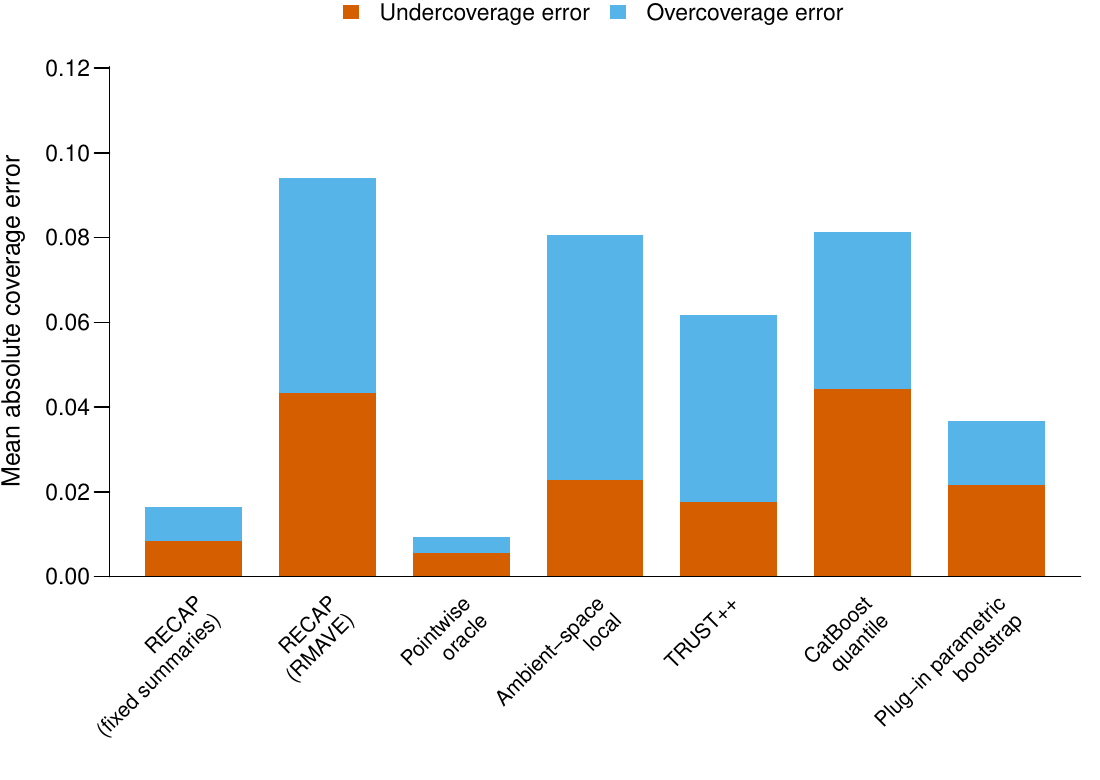}}
    \\[0.5ex]
    \subfigure[Simulation cost.]{%
        \includegraphics[width=0.78\textwidth]{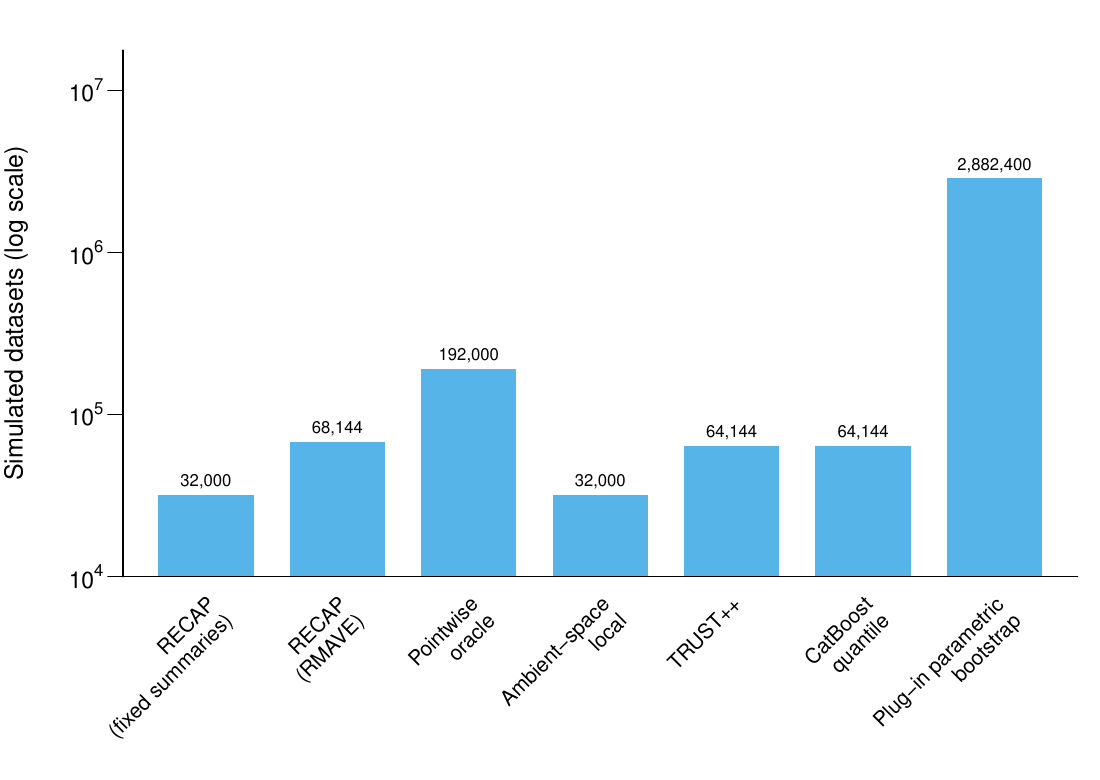}}
    \caption{Noisy matrix completion: (a) mean absolute coverage error at nominal level $0.90$; (b) total simulated datasets on a logarithmic scale.
    \texttt{RECAP} uses fixed summaries without representation learning.}
    \label{fig:supp-simulation-completion-summary}
\end{figure}

\begin{figure}[p]
    \centering
    \subfigure[Coverage MAE.]{%
        \includegraphics[width=\textwidth]{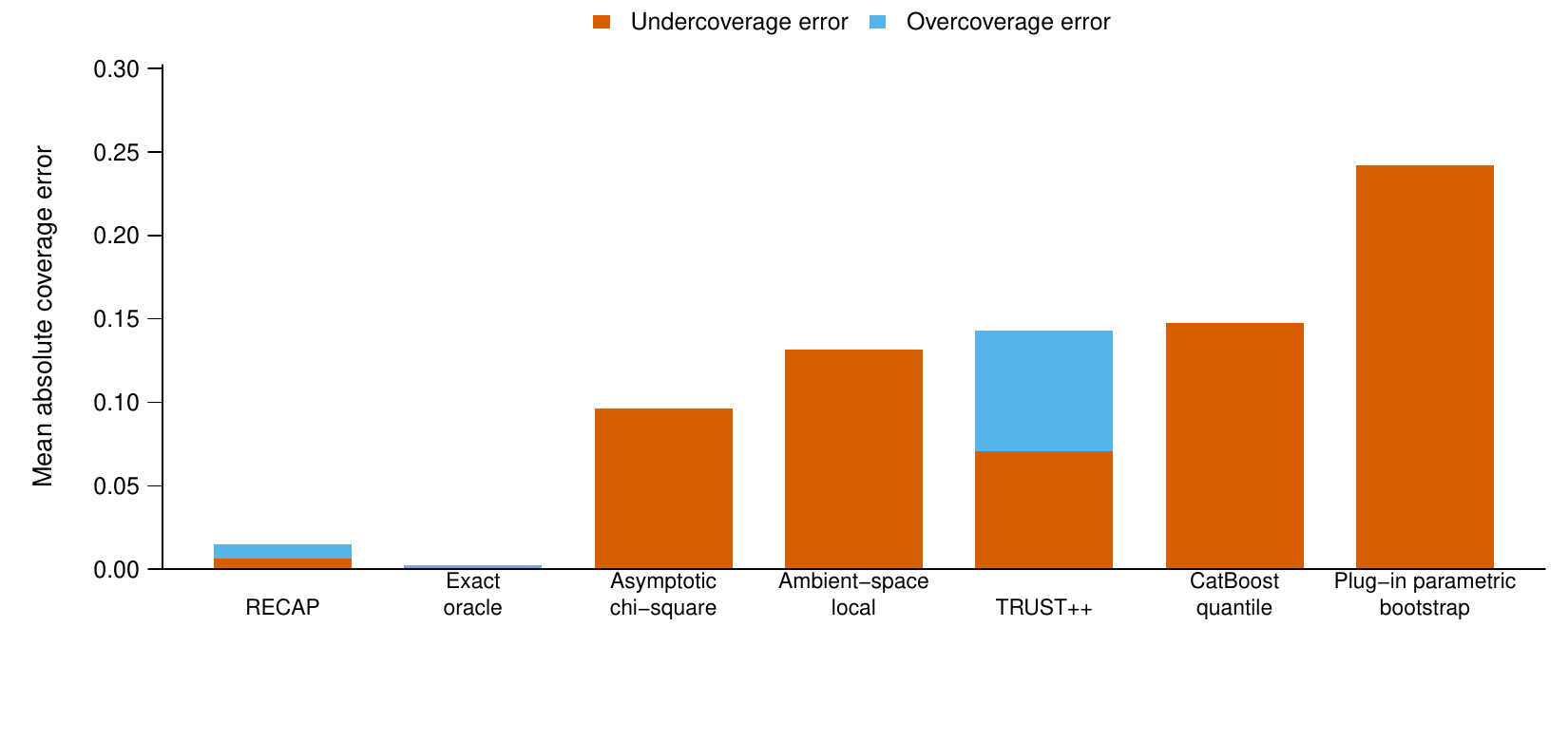}}
    \\[0.5ex]
    \subfigure[Simulation cost.]{%
        \includegraphics[width=0.78\textwidth]{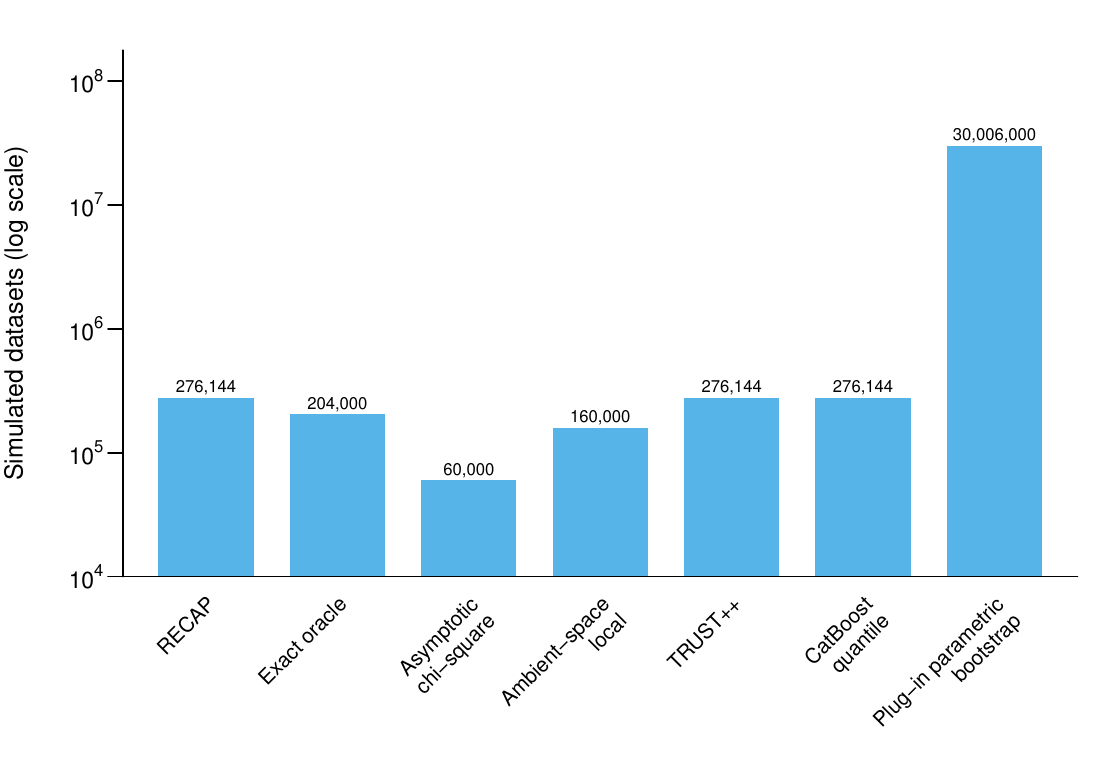}}
    \caption{Poisson regression: (a) final-test mean absolute coverage error at nominal level $0.90$; (b) total simulated datasets on a logarithmic scale, including representation learning and independent-location dimension selection for \texttt{RECAP}.}
    \label{fig:supp-simulation-poisson-summary}
\end{figure}

\section{Coalescent case study: model, validation, and implementation}
\label{sec:supp-real}

\subsection{Demographic specification and genomic summaries}
\label{subsec:supp-real-model}

We use a neutral two-population coalescent model with recombination \citep{Kingman1982,Hudson1983,Notohara1990}.
The two populations share a common ancestor before the split time $T_{\mathrm{split}}$.
The 35 parameters comprise ten population sizes for each present-day population, four ancestral population sizes, the split time, and five migration rates in each direction.
Population sizes and migration rates are constant within specified time intervals (Figure~\ref{fig:supp-real-demographic-model}).
Table~\ref{tab:real-size-bounds} gives the population-size bounds; the split time ranges from 1,600 to 3,600 generations, and migration rates range from $10^{-6}$ to $10^{-4}$ per generation.
We use the normalized log-parameters $u_j(\theta)=\log(\theta_j/L_j)/\log(U_j/L_j)$ from Section~\ref{subsec:real-application} and sample simulation locations across $[0,1]^{35}$ using Latin-hypercube designs.

\begin{figure}[t]
    \centering
    \includegraphics[width=0.95\textwidth]{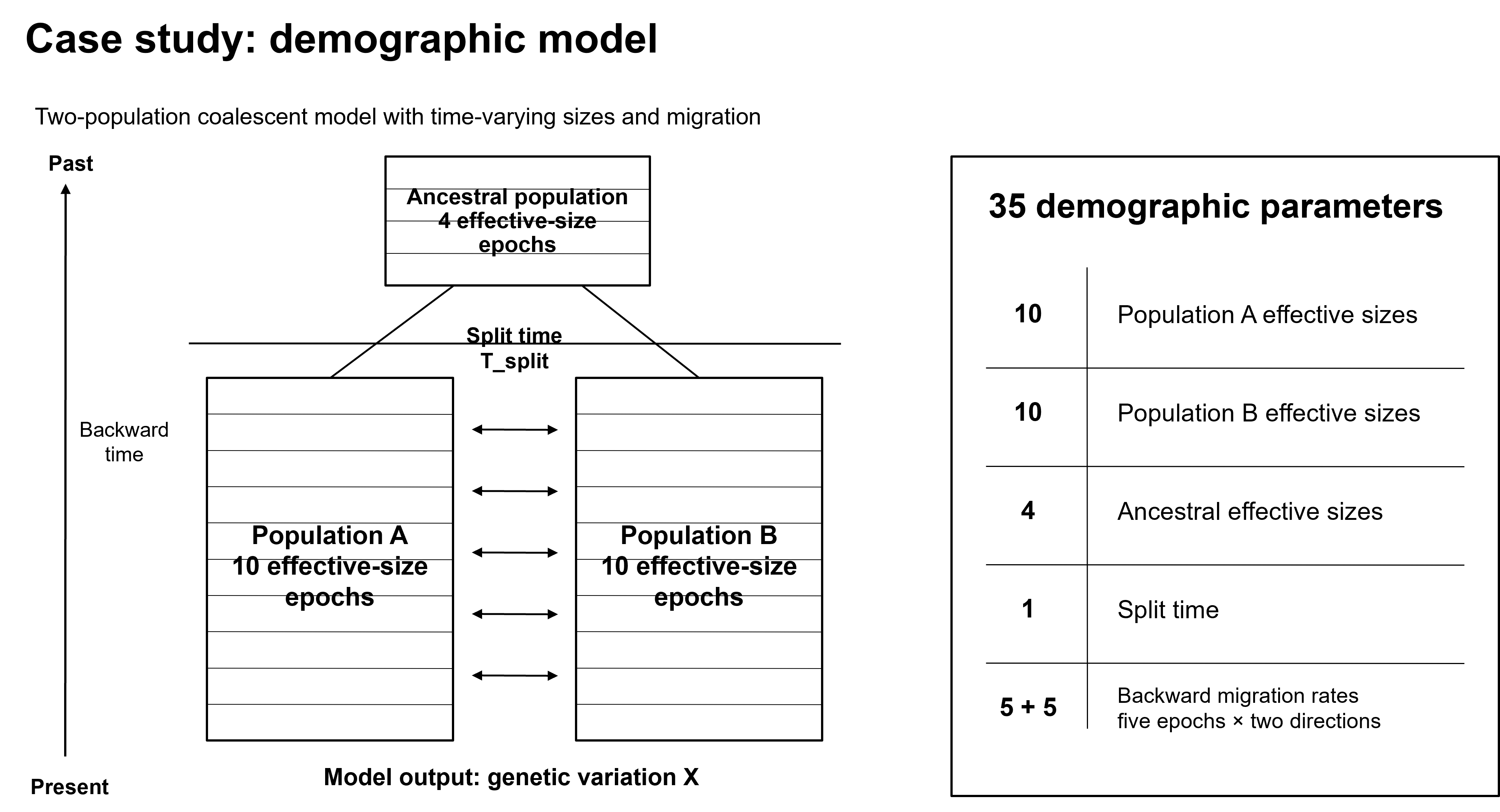}
    \caption{Two-population demographic model. Time runs backward from the present. Dashed lines show population-size epochs, and arrows show migration between the populations.}
    \label{fig:supp-real-demographic-model}
\end{figure}

\begin{table}[htbp]
    \centering
    \begin{tabular}{rrrr}
        \toprule
        Epoch & Population $A$ & Population $B$ & Ancestor \\
        \midrule
        1 & 15,000 & 22,500 & 6,000 \\
        2 & 12,000 & 15,000 & 7,500 \\
        3 & 10,000 & 10,000 & 9,000 \\
        4 & 8,500 & 6,000 & 9,000 \\
        5 & 7,500 & 3,500 & --- \\
        6 & 7,000 & 2,000 & --- \\
        7 & 6,500 & 1,250 & --- \\
        8 & 6,000 & 1,500 & --- \\
        9 & 5,500 & 2,500 & --- \\
        10 & 5,500 & 4,000 & --- \\
        \bottomrule
    \end{tabular}
    \caption{Lower bounds for effective population sizes. Each upper bound is four times its lower bound.}
    \label{tab:real-size-bounds}
\end{table}

Each observation $X$ contains four 20-kb genomic blocks for 20 diploid individuals from each population.
We simulate the blocks independently given $\theta$ using \texttt{msprime} \citep{Kelleher2016,Baumdicker2022}, with recombination rate $10^{-8}$ and mutation rate $1.25\times10^{-8}$ per base per generation.
We extract 86 features from each block: 12 marginal site-frequency-spectrum (SFS) features, 49 joint-SFS features, nine diversity and differentiation features, and 16 linkage-disequilibrium features.
The mean and standard deviation of each feature across the four blocks give 172 summaries, which we reduce to $Z(X)\in\mathbb R^{30}$ as described below.
Figure~\ref{fig:supp-real-workflow} shows the main steps.

\subsection{Statistic fitting, representation learning, and dimension selection}
\label{subsec:supp-real-learning}

We use 2,048 independent simulated parameter--observation pairs to fit the feature transformations, standardization, and principal components, retaining 30 components.
We estimate the mean $\widehat m$ in \eqref{eq:real-score} by quadratic ridge regression in standardized log-parameters.
Separate simulations are used to choose its penalty and estimate the residual corrections $\widehat b_\ell$ and scales $\widehat\tau_\ell$.
These fitted quantities are then held fixed.

We learn the representation $g$ using the characteristic-function \texttt{RMAVE} ensemble \citep{yin2011ensemble}, based on simulations at 512 parameter locations.
We choose the fitting settings by cross-validation and select the dimension from $\{1,2,3,4,6,8\}$ using coverage RMSE, corrected for Monte Carlo error, on independent simulations.
This gives $\widehat s=1$.

\begin{figure}[t]
    \centering
    \includegraphics[width=\textwidth]{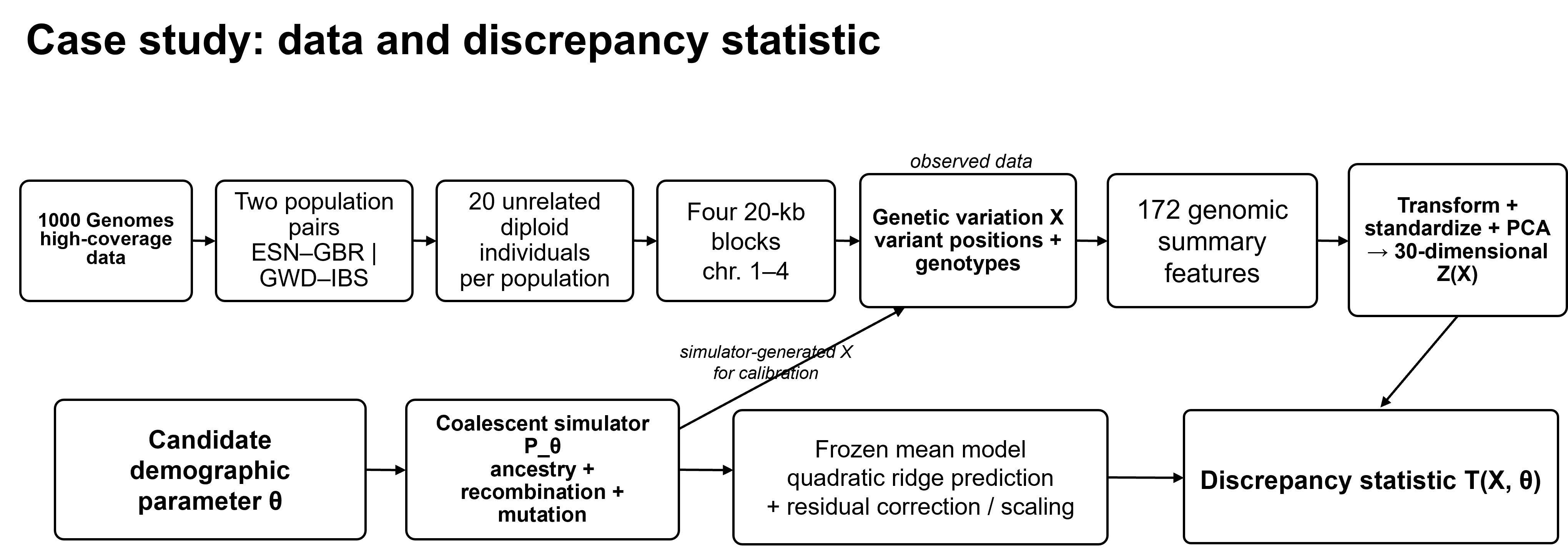}
    \caption{From genomic observations to the test statistic. The summary transformation and fitted score parameters are held fixed during calibration and inference.}
    \label{fig:supp-real-workflow}
\end{figure}

\subsection{Final calibration and simulation budgets}
\label{subsec:supp-real-calibration}

The final calibration bank contains 32 independent observations at each of 512 Latin-hypercube parameter locations, giving $16{,}384$ scores.
For each candidate $\theta$, \texttt{RECAP} pools scores from nearby values of the learned representation $g(\theta)$.
An independent parameter-only pilot sets the moving-ball radii to give about 512 pooled scores per candidate.
Ambient-space calibration uses the same score bank but pools in the original 35-dimensional parameter space.
The fixed-location theory in Supplement Section~\ref{sec:supp-fixed-design} describes this type of design under its stated assumptions; we also check coverage directly in simulations.

The pointwise oracle uses 256 null simulations at each tested parameter, and the plug-in bootstrap uses 512 simulations at a fitted parameter for each observed dataset.
\texttt{TRUST++} and \texttt{CatBoost} use the original 35-dimensional parameters, with separate fitting and tuning simulations, following the methods described in Supplement Section~\ref{subsec:supp-sim-comparisons}.
Table~\ref{tab:real-simulation-budget} gives the main simulation budgets.

\begin{table}[t]
    \centering
    \begin{tabular}{lrr}
        \toprule
        Stage or method & Design & Simulated observations \\
        \midrule
        Representation learning & $512\times128$ & 65,536 \\
        Dimension selection & Independent banks & 40,960 \\
        Comparator fitting and tuning & Independent banks & 26,624 \\
        Final shared calibration & $512\times32$ & 16,384 \\
        Final evaluation & $12\times144$ & 1,728 \\
        Pointwise calibration & $36\times256$ & 9,216 \\
        Bootstrap calibration & $12\times144\times512$ & 884,736 \\
        \bottomrule
    \end{tabular}
    \caption{Main simulation budgets for the coalescent study. Each simulated observation contains four genomic blocks.}
    \label{tab:real-simulation-budget}
\end{table}

Across the reported evaluation, bootstrap calibration requires 54 times as many simulations as the shared calibration bank.
The shared bank can also be reused for new candidates and independent observations.
Figure~\ref{fig:supp-real-cost} compares the final-calibration costs and the number of scores used for each cutoff.

\begin{figure}[t]
    \centering
    \subfigure[Final-calibration simulations]{%
        \includegraphics[width=0.48\linewidth]{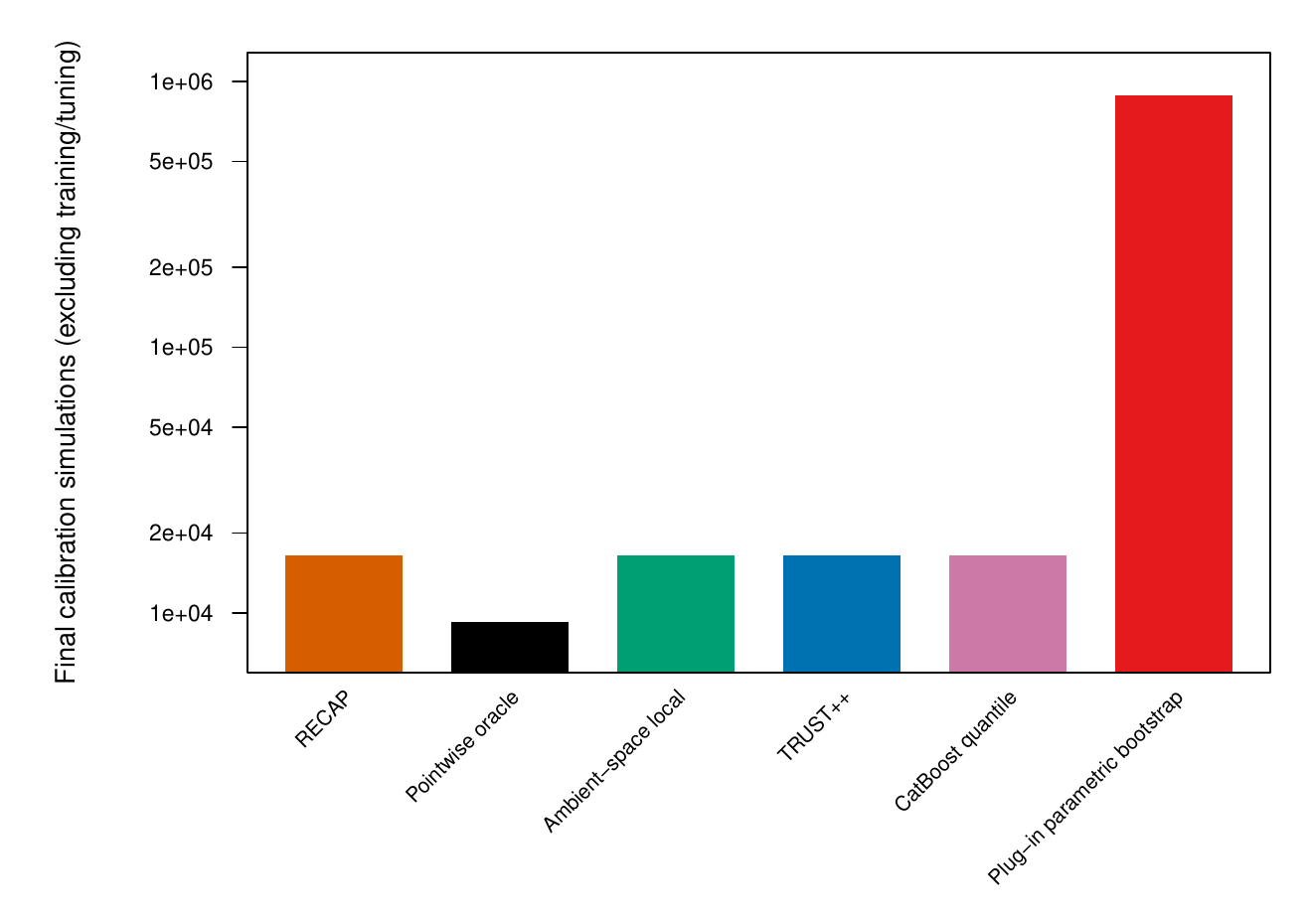}}
    \hfill
    \subfigure[Scores per cutoff]{%
        \includegraphics[width=0.48\linewidth]{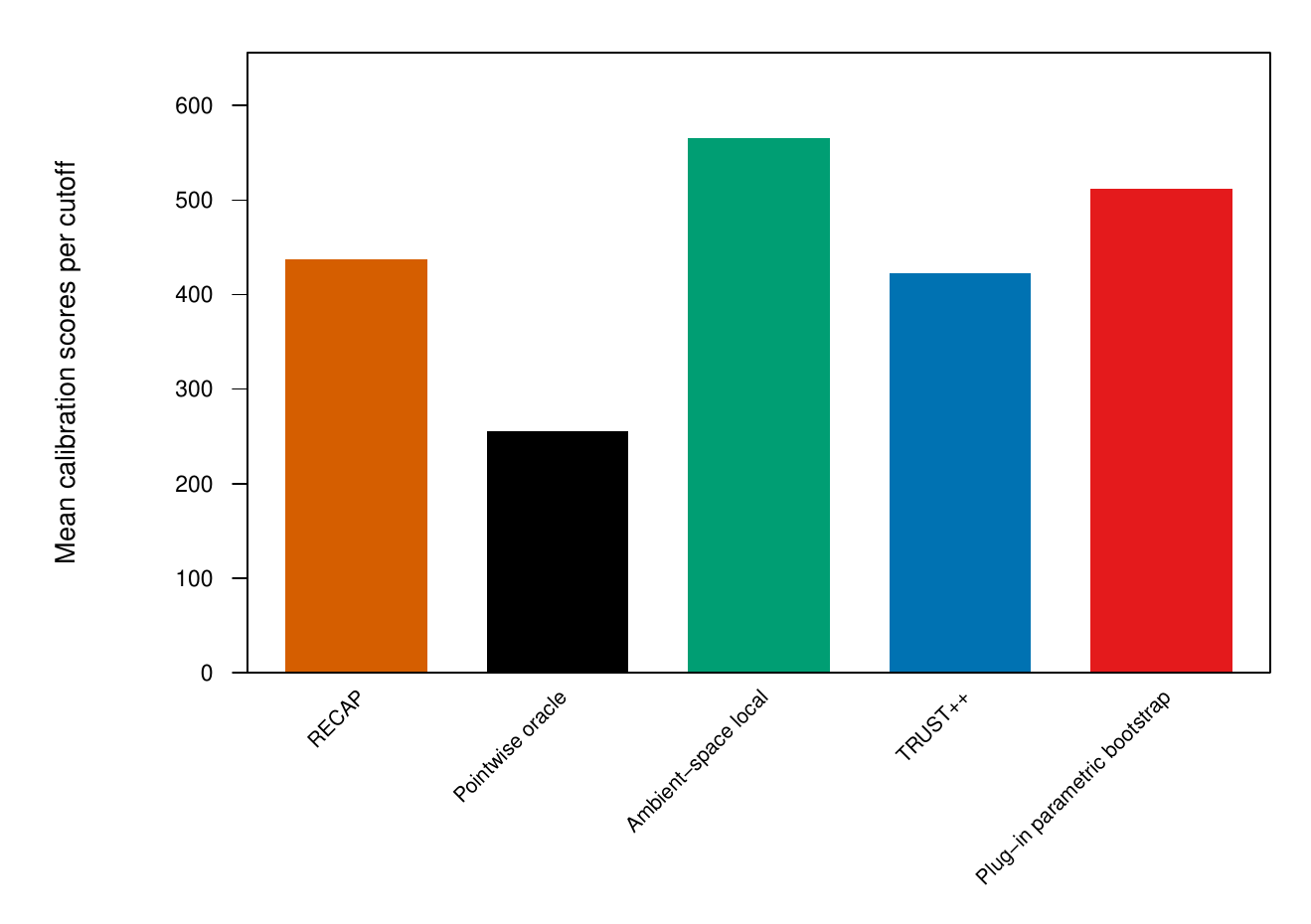}}
    \caption{Final-calibration costs and pool sizes in the coalescent comparison.}
    \label{fig:supp-real-cost}
\end{figure}

\subsection{Calibration validation: coverage and power}
\label{subsec:supp-real-validation}

We compare \texttt{RECAP} with ambient-space calibration, \texttt{TRUST++}, \texttt{CatBoost} quantile regression, the plug-in parametric bootstrap, and the pointwise oracle.
All methods use the same statistic.
Ambient-space calibration uses the same moving-ball rule as \texttt{RECAP}, but in the original 35-dimensional parameter space.
\texttt{TRUST++} fits a random forest to simulated scores, while \texttt{CatBoost} directly estimates their conditional quantile; both use the original parameters and are tuned on independent simulations.
The plug-in bootstrap calibrates at a nearest-neighbor \texttt{ABC} estimate \citep{Beaumont2002}, while the pointwise oracle simulates directly at each tested parameter.
The twelve evaluation targets are chosen to be well separated in the normalized 35-dimensional parameter space.
At each target, we estimate coverage using 144 independent observations.
For power, we test two false candidates whose expected summary statistics are well separated from those under the true parameter.

Figure~\ref{fig:real-validation} shows coverage and power, and Figure~\ref{fig:supp-real-mae} shows coverage mean absolute error (MAE) across the twelve targets.
\texttt{RECAP} stays close to the pointwise oracle and has the smallest coverage MAE among the nonoracle methods.
\texttt{TRUST++} and \texttt{CatBoost} have larger coverage errors at several targets, while the bootstrap shows substantial undercoverage.

\begin{figure}[t]
    \centering
    \includegraphics[width=0.8\textwidth]{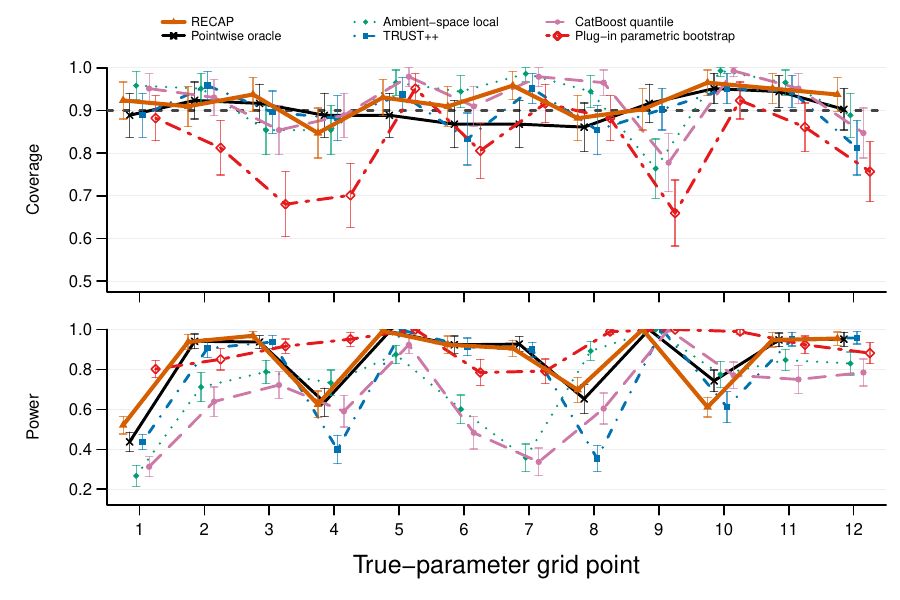}
    \caption{Coverage and power in the coalescent study at nominal coverage $0.90$. Each target uses 144 evaluation observations. Bars give pointwise 95\% Monte Carlo intervals.}
    \label{fig:real-validation}
\end{figure}

\begin{figure}[t]
    \centering
    \includegraphics[width=0.80\linewidth]{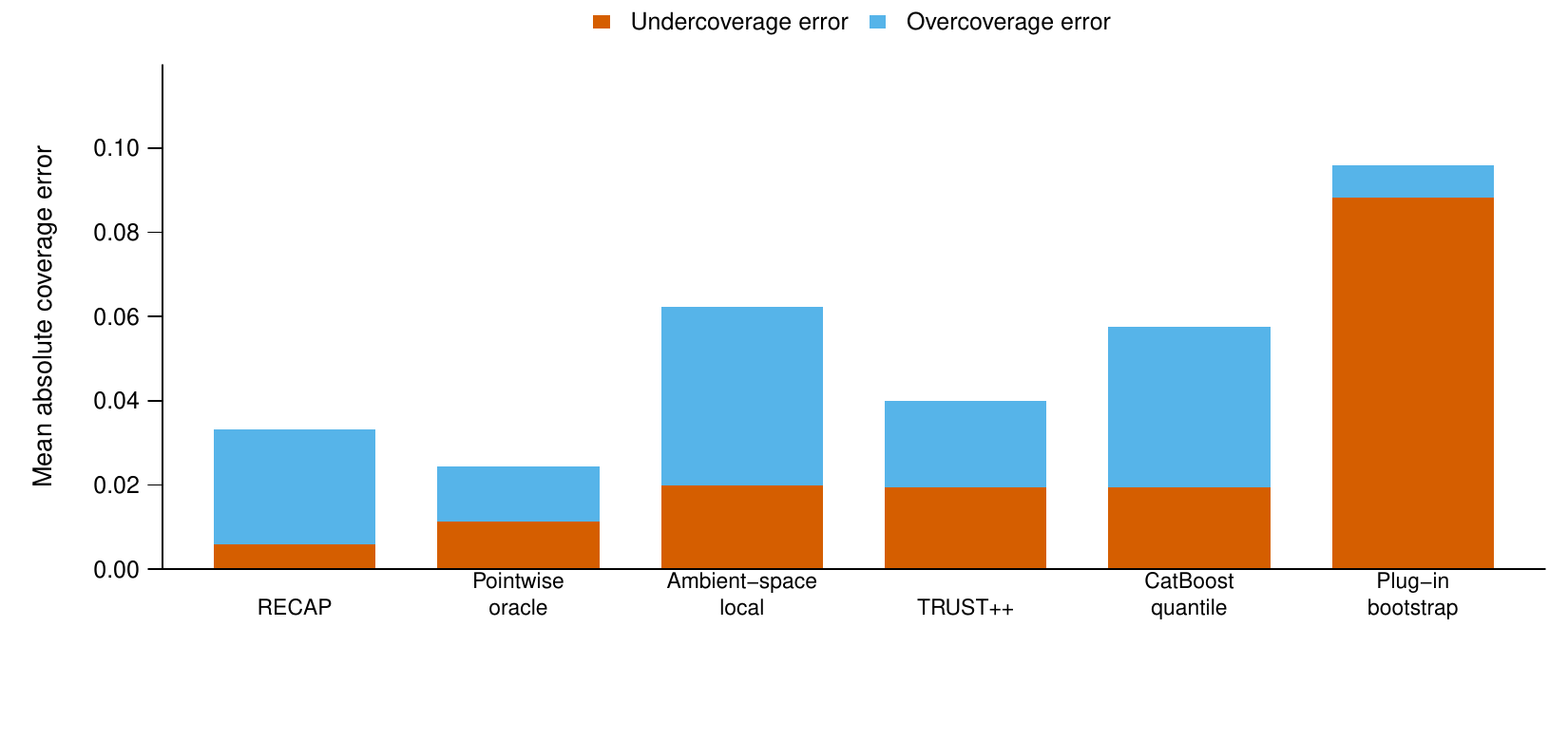}
    \caption{Coverage MAE across twelve targets, separated into undercoverage and overcoverage.}
    \label{fig:supp-real-mae}
\end{figure}

\subsection{Real-data preparation and selected model-adequacy checks}
\label{subsec:supp-real-data}

We restrict the analysis to unrelated individuals and genomic regions with high accessibility, using the same quality-control rules for both population pairs.
For inference, we select four 20-kb windows, one from each of chromosomes 1--4, based on genomic accessibility rather than genotype values.
The same windows are used for both population pairs (Table~\ref{tab:real-windows}).

\begin{table}[t]
    \centering
    \begin{tabular}{lrrr}
        \toprule
        Chromosome & Start & End & Accessible bases \\
        \midrule
        1 & 199,687,055 & 199,707,055 & 19,300 \\
        2 & 193,068,261 & 193,088,261 & 19,825 \\
        3 & 145,183,745 & 145,203,745 & 19,250 \\
        4 & 136,519,542 & 136,539,542 & 19,328 \\
        \bottomrule
    \end{tabular}
    \caption{GRCh38 inference windows with zero-based, half-open coordinates. Accessible-base counts include SNP quality control.}
    \label{tab:real-windows}
\end{table}

For model checking, we use separate fitting and held-out genomic blocks from chromosomes 1--8, with the same genomic masks applied to observed and simulated data.
For model checking, we use separate fitting and held-out genomic blocks from chromosomes 1--8, with the same genomic masks applied to observed and simulated data.

We fit the demographic parameters using summaries from 64 blocks and assess model fit on 64 held-out blocks.
The checks cover genetic differentiation and the marginal SFS.
We obtain approximate $p$-values using 256 parametric-bootstrap samples, refitting the model each time.
None of the four checks rejects at the 5\% level (Table~\ref{tab:real-model-fit}).
These checks provide some evidence of model fit, but do not test every aspect of the simulator.

\begin{table}[t]
    \centering
    \caption{Held-out model-check $p$-values from 256 refitted parametric-bootstrap samples per population pair.}
    \label{tab:real-model-fit}
    \begin{tabular}{lrr}
        \toprule
        Population pair & Differentiation & Marginal folded SFS \\
        \midrule
        ESN--GBR & $0.5564$ & $0.3191$ \\
        GWD--IBS & $0.1167$ & $0.4125$ \\
        \bottomrule
    \end{tabular}
\end{table}

\subsection{Numerical inversion and interpretation}
\label{subsec:supp-real-inversion}

We approximate the confidence sets by searching a Latin-hypercube sample of parameters and then refining around promising candidates.
The reported interval $I_g(X_0)$ in \eqref{eq:real-g-interval} uses the smallest and largest values of $g(\theta)$ among the accepted candidates found in this search.
The full parameter-box range of $g$ is $[-\sqrt3\|\widehat B\|_1,\sqrt3\|\widehat B\|_1]$, approximately $[-3.596,3.596]$.

For the two-dimensional slices in Figure~\ref{fig:real-application}, we vary split time and deep ancestral population size while fixing the other 33 normalized parameters at $0.5$.
The slices describe acceptance under this fixed background.

\end{document}